\documentclass[a4paper,UKenglish,cleveref, autoref, thm-restate]{lipics-v2021}
\usepackage[utf8]{inputenc}
\usepackage{tikz}
\usetikzlibrary{arrows.meta,positioning,calc}
\usetikzlibrary{shapes}
\usepackage{adjustbox}
\usepackage{subcaption}
\usepackage{float}

\usepackage[ruled,vlined,linesnumbered,longend, noend]{algorithm2e}
\newcommand{\algosize}{\fontsize{9pt}{11pt}\selectfont}

\SetAlFnt{\algosize}        
\SetAlCapFnt{\algosize}     
\SetAlCapNameFnt{\algosize} 

\usepackage{url}

\usepackage{etoolbox}
\makeatletter
\patchcmd{\@algocf@start}{\vbox}{\vbox}{}{}
\makeatother

\usetikzlibrary{shapes}
\usepackage{adjustbox}

\usepackage{bm}

\SetKwComment{Comment}{/* }{ */}
\title{How Long Can the Escaping Ant Be Confined?}

\author{Kossi Roland Etse}{DAVID Lab., UVSQ, Université Paris Saclay, 45 avenue des Etats-Unis, 78000, Versailles, France. \and \url{https://www.david.uvsq.fr/m-kossi-roland-etse} }{kossi-roland.etse@uvsq.fr}{https://orcid.org/0009-0002-3583-6708}{}

\authorrunning{K. R. Etse}

\Copyright{Kossi Roland Etse}



\category{} 

\relatedversion{} 

\acknowledgements{I would like to thank David Auger and Pierre Coucheney for valuable discussions and insightful feedback. I also thank Jérémie Cabessa and Yann Strozecki for their help with the simulations, and anonymous reviewers for their comments and suggestions.}

\nolinenumbers 

\EventEditors{John Q. Open and Joan R. Access}
\EventNoEds{2}
\EventLongTitle{42nd Conference on Very Important Topics (CVIT 2016)}
\EventShortTitle{CVIT 2016}
\EventAcronym{CVIT}
\EventYear{2016}
\EventDate{December 24--27, 2016}
\EventLocation{Little Whinging, United Kingdom}
\EventLogo{}
\SeriesVolume{42}
\ArticleNo{23}

\begin{document}

\maketitle


\begin{abstract}
Langton's ant is a simple two-dimensional cellular automaton whose long-term behavior exhibits
remarkable complexity. While it is known that the ant eventually escapes any finite
connected region of the grid, the quantitative aspects of this escape remain poorly
understood. In this paper, we study the \emph{escaping time} of Langton's ant,
defined as the maximum number of steps the ant can perform within a finite connected
domain before leaving it.

We establish general upper bounds on the escaping time as a function of the domain
size, and derive improved bounds for rectangular domains. In
particular, we obtain a factorial upper bound for square domains via an inductive
decomposition argument. We also obtain linear upper bounds for rectangular domains of height
two and three via a column-by-column analysis. More generally, for rectangular domains with a fixed height, we establish a polynomial upper bound in the number of columns.
These results are
complemented by exact values computed through an optimized simulation algorithm
that exploits the geometric symmetries of the grid and employs a backtracking
branching strategy to avoid exhaustive search over all color configurations. We also provide lower-bound constructions, proving that the linear upper bounds for rectangular domains of heights two and three are asymptotically optimal.
\end{abstract}

\section{Introduction}

Many natural and artificial systems can be modeled as collections of simple interacting elements evolving over discrete time. Each element updates its state according to local rules that depend only on a small neighborhood, yet the repeated application of these rules can generate complex global behavior. Understanding how such macroscopic patterns arise from microscopic interactions, often referred to as \emph{emergence}, is a central question in the study of discrete dynamical systems. 

A classical example of emergence in a discrete dynamical system is \emph{Langton's ant},
a two-dimensional cellular automaton introduced by Christopher G. Langton in 1986 in the context
of artificial life \cite{langton1986artificial}. 
Langton's ant is defined on the infinite square grid whose cells are in one of two states: \emph{black} (turn-left) or \emph{white} (turn-right). The ant occupies a cell and has an orientation among the four cardinal directions. At each step, the ant turns left on a black cell and right on a white cell, flips the state of the cell, and moves forward to the adjacent cell in the new direction.

Despite the simplicity of these local rules, the trajectory of the ant exhibits highly unpredictable behavior. Starting from a uniformly white (or black) initial configuration, the ant's trajectory shows rotational symmetry during the first 500 steps, followed by an apparently erratic phase lasting until around step 10 000, after which the trajectory suddenly becomes regular: the ant advances indefinitely along one of the four diagonal directions, tracing a periodic pattern of period 104 known as the \emph{highway} (Figure \ref{fig:highway}). Numerical simulations consistently show that this highway emerges from any finite initial configuration (configuration in which all but a finite number of cells are in the same state), yet this experimental observation has never been formally established.

Beyond the standard two-dimensional square lattice, Langton's ant has been studied on a variety of other regular structures and topologies, including one-dimensional lattices \cite{gajardo2004dynamics, grosfils1999propagation}, triangular and hexagonal grids \cite{grosfils1999propagation, tsukiji2011pspace}, infinite bi-regular and hyperbolic graphs, finite and planar graphs \cite{gajardo2001stacs}, twisted tori \cite{hagiwara2020torus}, and higher-dimensional lattices \cite{bunimovich1996dim}. Those works show that the dynamical behavior of the ant is strongly influenced by the topology of the underlying graph. Generalizations have also been proposed along other axes: variants with a richer set of movements and multi-state extensions \cite{gajardo2025highways, gale1995further}.

While much attention has been devoted to the ant's behavior on infinite grids, its dynamics on finite domains remain poorly understood. It is known that Langton’s ant eventually escapes any finite connected region of the grid, yet the quantitative aspects of this escape —and more broadly, the computational complexity of predicting its behavior on finite two-color grids— are still unclear.  A natural and fundamental question thus arises:\emph{ how long can the ant remain confined within a finite region before escaping?} Bounding the escaping time provides insight into the complexity of reachability problems and prediction difficulty for the ant on finite domains.

\begin{figure}[!t]
    \centering
    \resizebox{0.8\linewidth}{!}{%
    \begin{minipage}{\linewidth}
    \begin{subfigure}[b]{0.32\textwidth}
        \includegraphics[width=\textwidth]{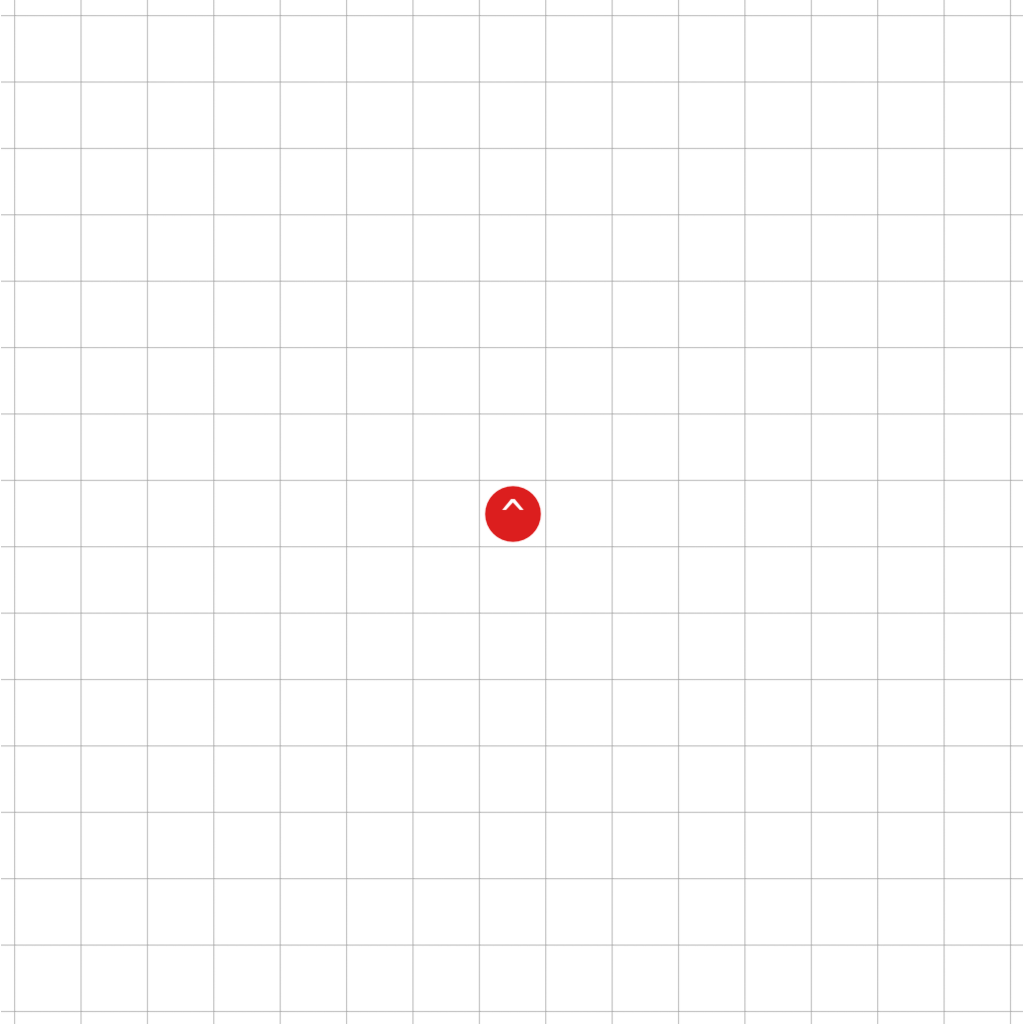}
        \caption*{(a) 0 Steps }
    \end{subfigure}
    \hfill
    \begin{subfigure}[b]{0.32\textwidth}
        \includegraphics[width=\textwidth]{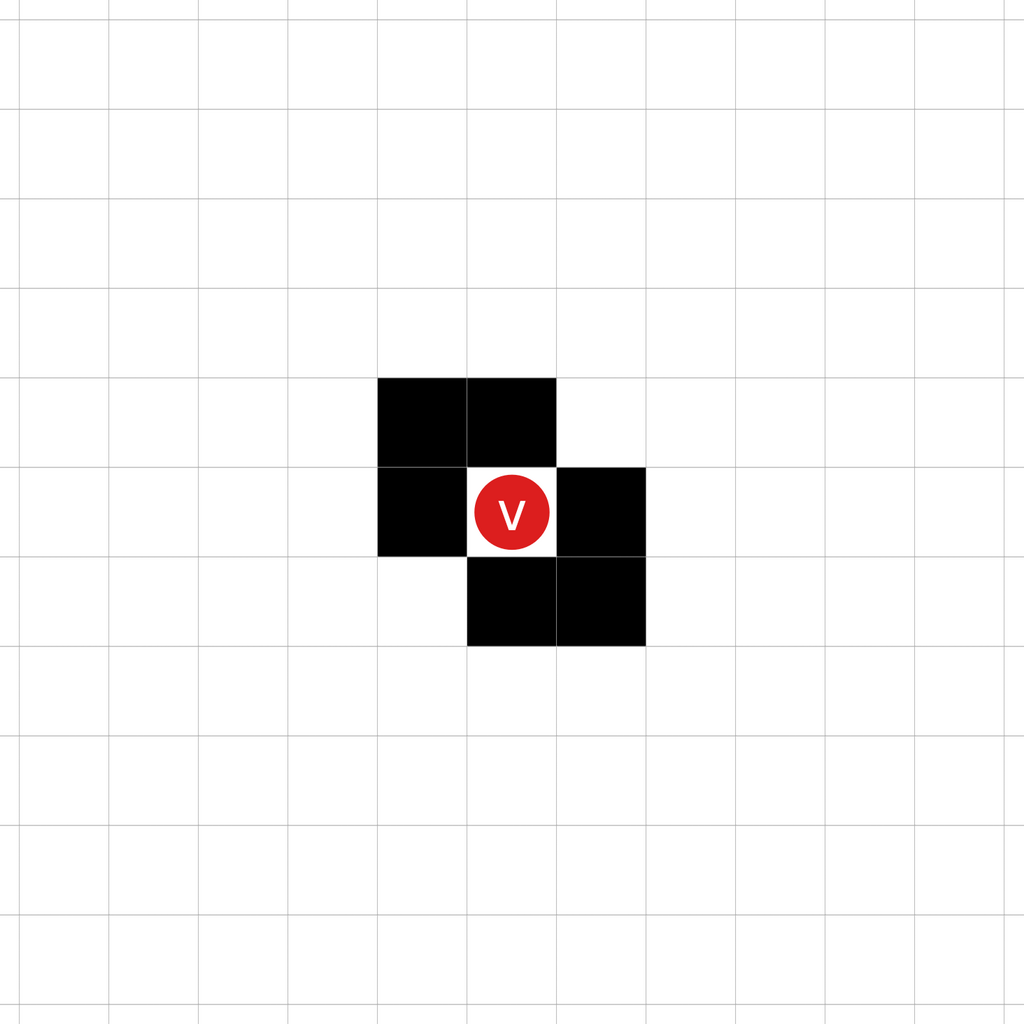}
        \caption*{(b) 8 Steps }
    \end{subfigure}
    \hfill
    \begin{subfigure}[b]{0.32\textwidth}
        \includegraphics[width=\textwidth]{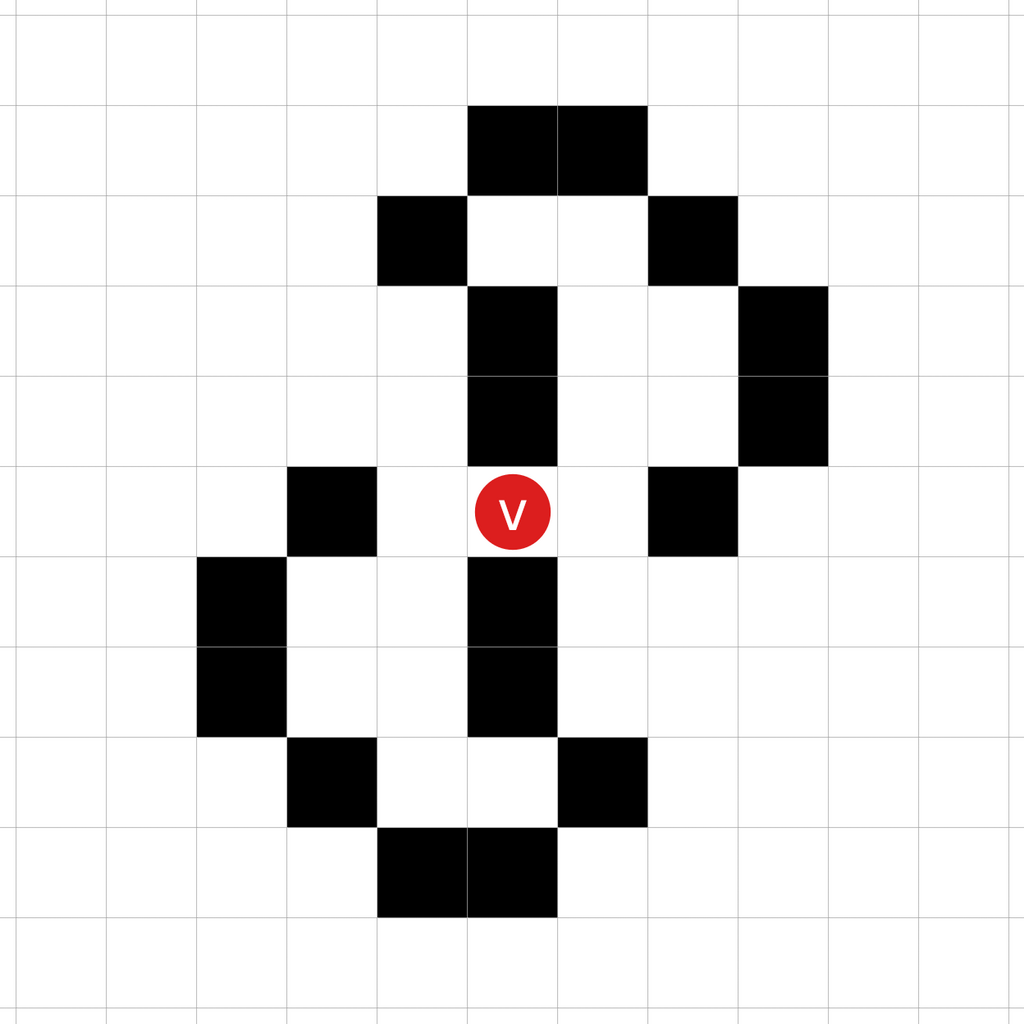}
        \caption*{(c) 96 Steps}
    \end{subfigure}

    \vspace{0.5em}

    \begin{subfigure}[b]{0.32\textwidth}
        \includegraphics[width=\textwidth]{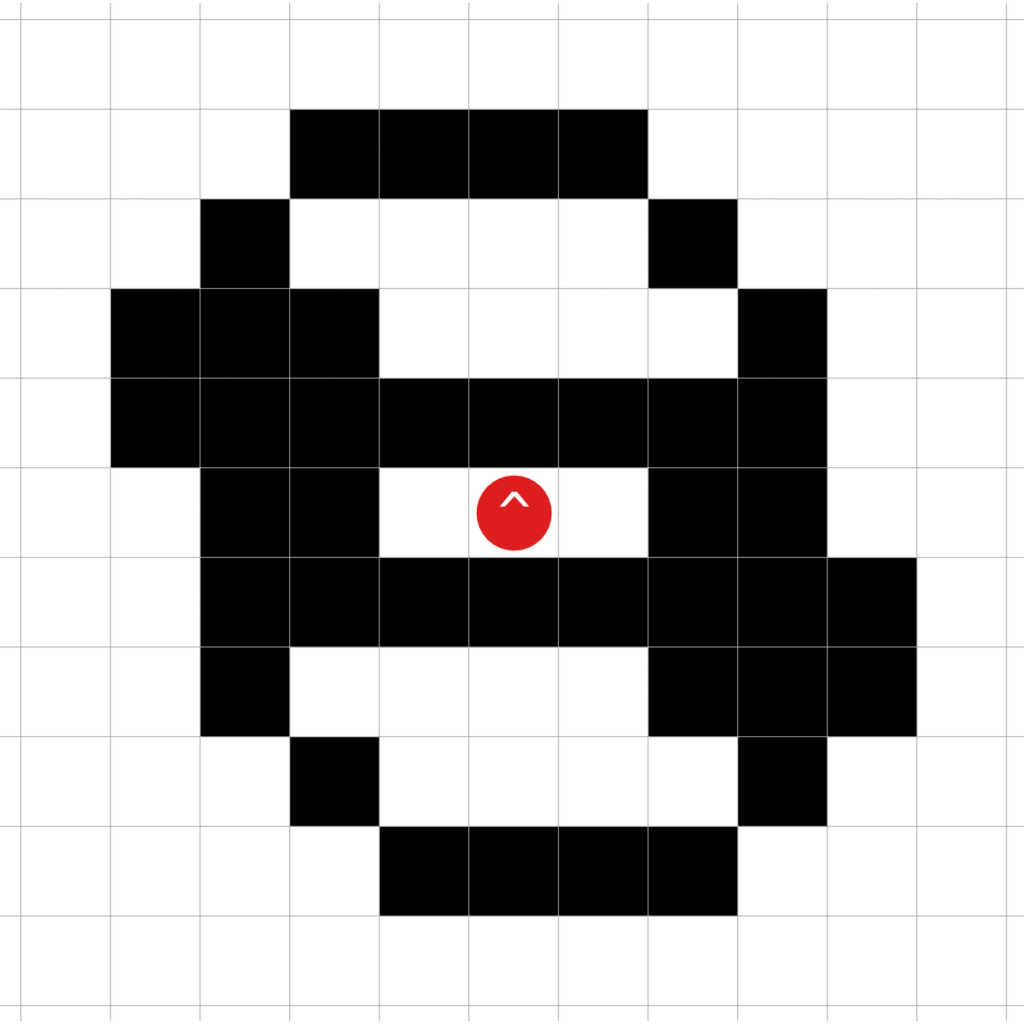}
        \caption*{(d) 184 Steps}
    \end{subfigure}
    \hfill
    \begin{subfigure}[b]{0.32\textwidth}
        \includegraphics[width=\textwidth]{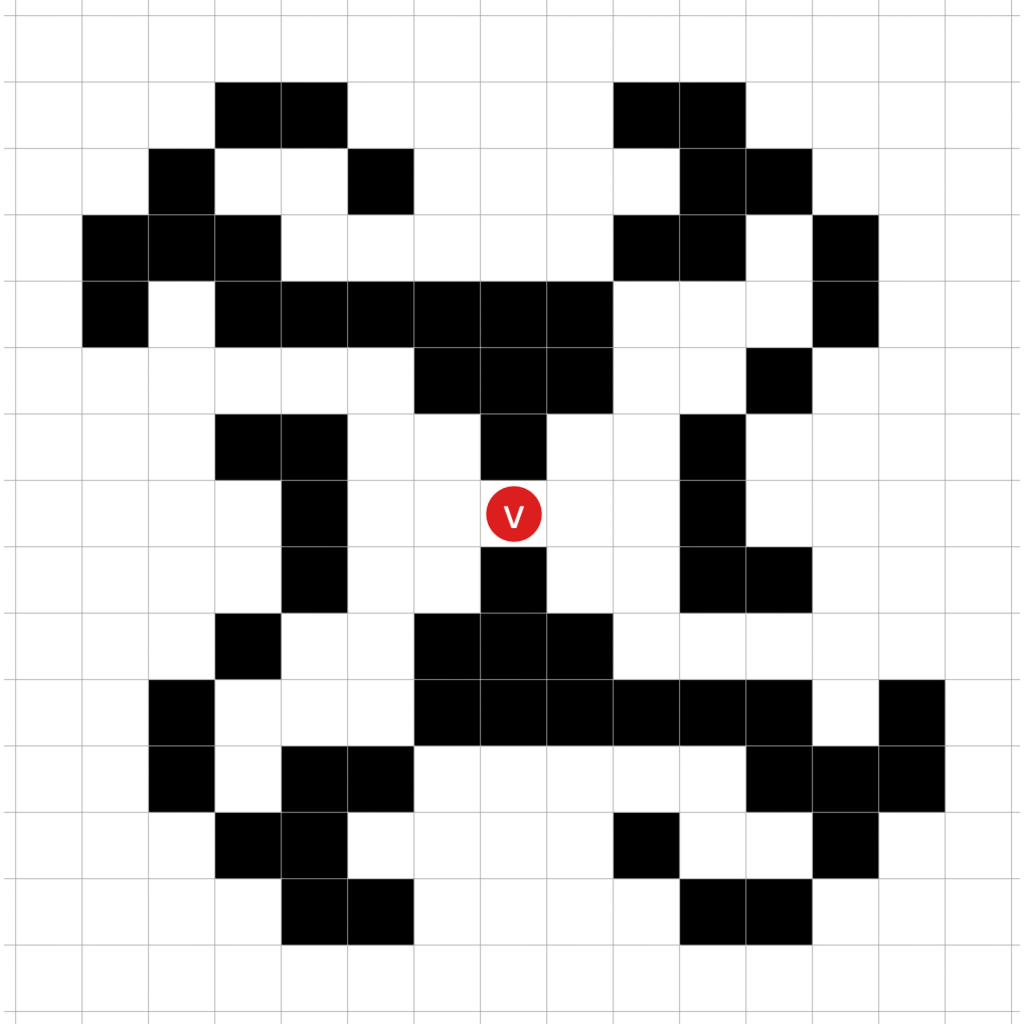}
        \caption*{(e) 368 Steps}
    \end{subfigure}
    \hfill
    \begin{subfigure}[b]{0.32\textwidth}
        \includegraphics[width=\textwidth]{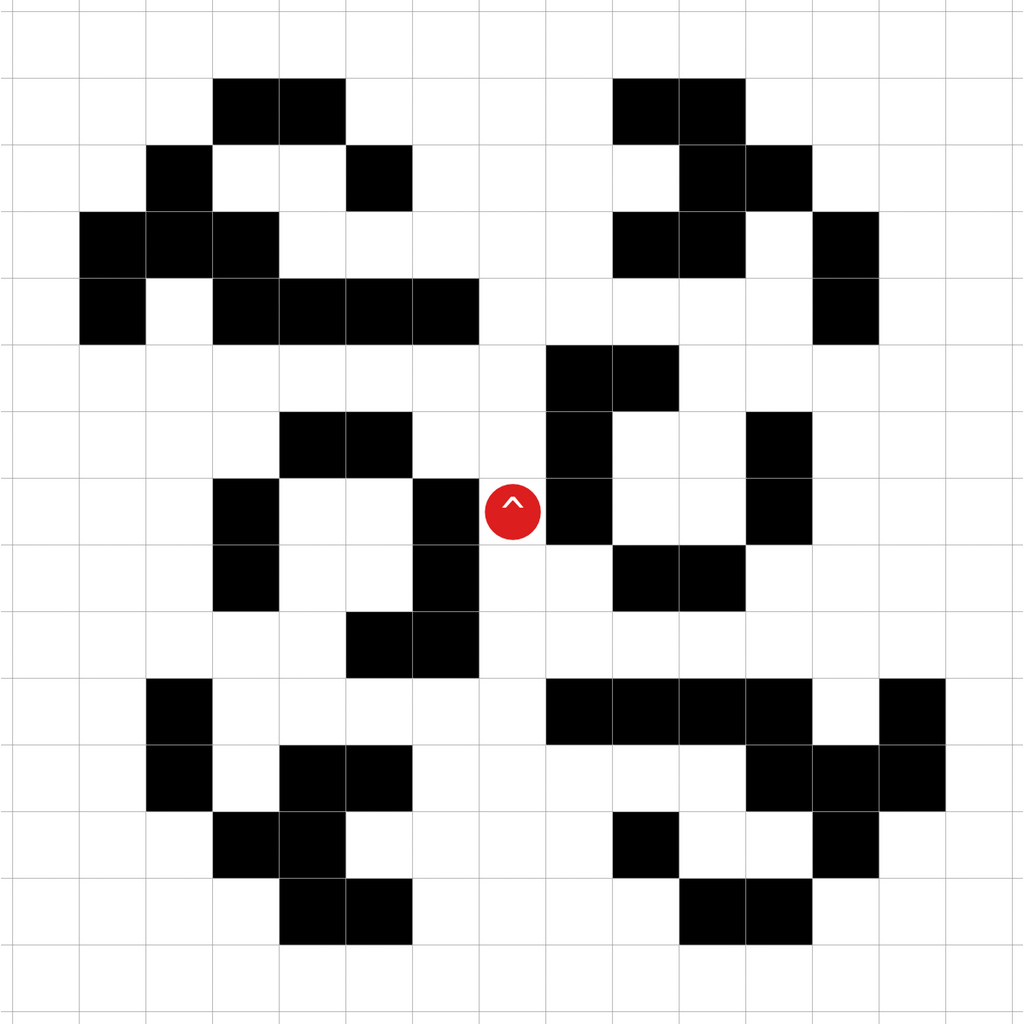}
        \caption*{(f) 472 Steps}
    \end{subfigure}

    \vspace{0.5em}

    \begin{center}
    \begin{subfigure}[b]{0.5\textwidth}
        \includegraphics[width=\textwidth]{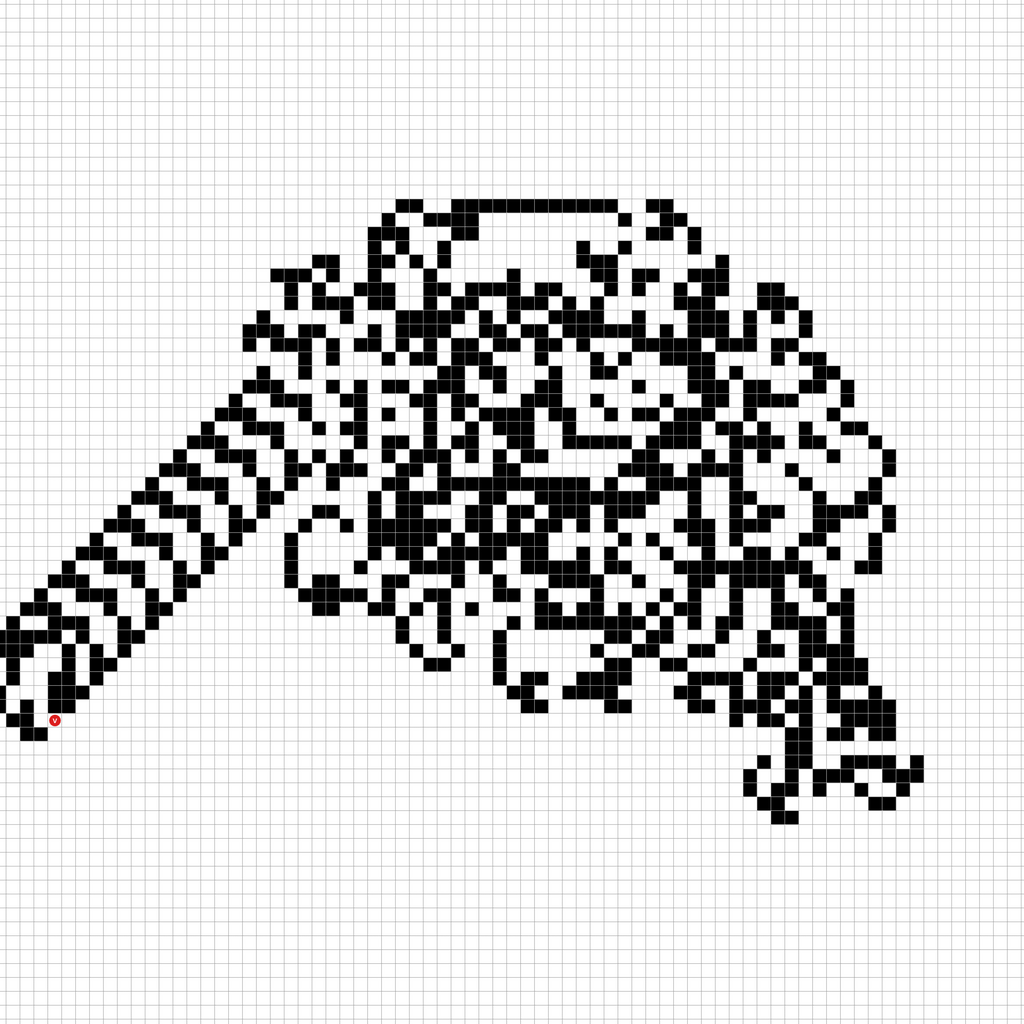}
        \caption*{(g) Highway}
    \end{subfigure}
    \end{center}
    \end{minipage}
    }

    \caption{The ant starting from a uniform white configuration facing north. 
Snapshots (a--f) after steps 8, 96, 184, 368, and 472 illustrate the early 
rotational symmetry, while (g) shows the eventual \emph{highway} regime.}
    \label{fig:highway}
\end{figure}

\paragraph*{The Present Work}
This paper studies the escaping time of Langton's ant on the standard two-color
square grid from a quantitative perspective, without modifying the cell states
or the underlying topology. Our contributions are as follows.

We first establish that, while confined to a finite connected domain, the same
color configuration cannot appear twice in the ant's trajectory. This
immediately yields a general upper bound of $2^d$ on the escaping time for any
domain of $d$ cells.

We then derive improved bounds depending on the shape of the domain. For square
domains with $n$ rows and $n$ columns, an inductive decomposition into a
boundary layer and interior yields a factorial upper bound $(n+1)!$. For
rectangular domains of height two, a column-by-column analysis gives the linear
bound $6(n-1)$ for a grid with $n$ columns, which is tight: we exhibit explicit
configurations achieving exactly $6(n-1)$ steps. For rectangular domains of
height three, the same approach yields the linear bound $10n-4$, with
near-matching constructions.
More generally, for fixed $k$, an inductive
argument gives for rectangular domains with $k$ rows and $n$
columns an upper bound $\left(\frac{n+1}{\sqrt{2}}\right)^k$, 
which is polynomial in $n$ for every fixed $k$.

These results are complemented by exact escaping times computed for all
rectangular grids up to size $9 \times 9$, using an optimized simulation
algorithm that exploits grid symmetries and a backtracking branching strategy,
which in particular corrects and extends the OEIS sequence
A282425~\cite{oeisA282425}.

\medskip
The remainder of this paper is organized as follows. 
Section \ref{sec:preliminaries} introduces the formal framework, definitions, and known results.
Section \ref{sec:general} presents our general results on finite connected domains, including the proof that color configurations cannot repeat. Section \ref{sec:squarerectangular} details our simulations (\ref{sec:simulation}), the inductive derivation of the square domain bound (\ref{sec:square}), the linear bounds for two-row  (\ref{sec:tworow}) and three-row  (\ref{sec:threerow}) rectangular grids and the inductive derivation of the bound of rectangular domain with fixed height (\ref{sec:rectupperbound}).

\section{Preliminaries} \label{sec:preliminaries}

This section introduces the concepts, definitions, and notation that will be used throughout the paper.

\subsection{Graph Representation and Transition Function}
We represent the infinite grid as a planar directed graph  $G = (V, A)$,  where the vertex set $V = \mathbb{Z}^2$ corresponds to the cells, and $A$ consists of arcs in both directions between vertices corresponding to adjacent cells (i.e. pairs of vertices differing by exactly one unit in exactly one coordinate). See Figure \ref{fig:gg} for an illustration, where each pair of opposite arcs is represented as single undirected edge.
In this representation, the position of the ant (the cell it occupies and its orientation) is encoded by an arc $(u,v) \in A$, specifying both the cell $u$ from which the ant arrives and the cell $v$ toward which it is heading. 
We identify vertices with cells of the grid and arcs with oriented positions of the ant. Thus, we may use ‘cell’ for vertices and ‘position’ for arcs.

\begin{figure}[H]
\centering
\begin{tikzpicture}[
scale=0.5,
vertex/.style={circle, draw, minimum size=2.5mm, inner sep=0pt},
fvertex/.style={circle, draw, fill=black, minimum size=2.5mm, inner sep=0pt}, 
>=Stealth,
thick
]
\def\k{5}
\def\n{5}
\def\spacing{1.5}
\def\externalarrow{1.0}

\pgfmathsetmacro{\gridleft}{-\spacing/2}
\pgfmathsetmacro{\gridright}{(\n-1)*\spacing+\spacing/2}
\pgfmathsetmacro{\gridbottom}{-\spacing/2}
\pgfmathsetmacro{\gridtop}{(\k-1)*\spacing+\spacing/2}

\foreach \i in {0,...,\n} {
    \pgfmathsetmacro{\xline}{(\i-0.5)*\spacing}
    \draw[dashed, gray] (\xline,\gridbottom) -- (\xline,\gridtop);
}
\foreach \j in {0,...,\k} {
    \pgfmathsetmacro{\yline}{(\j-0.5)*\spacing}
    \draw[dashed, gray] (\gridleft,\yline) -- (\gridright,\yline);
}

\foreach \i in {0,...,\numexpr\n-1} {
    \foreach \j in {0,...,\numexpr\k-1} {
        \pgfmathsetmacro{\x}{\i*\spacing}
        \pgfmathsetmacro{\y}{\j*\spacing}
        \node[vertex] (\i v\j) at (\x,\y) {};
    }
}

\foreach \j in {0,...,\numexpr\k-1} {
    \pgfmathsetmacro{\y}{\j*\spacing}
    \draw (-\externalarrow,\y) -- (0v\j);
    \foreach \i in {0,...,\numexpr\n-2} {
        \draw (\i v\j) -- (\the\numexpr\i+1\relax v\j);
    }
    \pgfmathsetmacro{\xright}{(\n-1)*\spacing+\externalarrow}
    \draw (\the\numexpr\n-1\relax v\j) -- (\xright,\y);
}

\foreach \i in {0,...,\numexpr\n-1} {
    \pgfmathsetmacro{\x}{\i*\spacing}
    \draw (\x,-\externalarrow) -- (\i v0);
    \foreach \j in {0,...,\numexpr\k-2} {
        \draw (\i v\j) -- (\i v\the\numexpr\j+1\relax);
    }
    \pgfmathsetmacro{\ytop}{(\k-1)*\spacing+\externalarrow}
    \draw (\i v\the\numexpr\k-1\relax) -- (\x,\ytop);
}

\end{tikzpicture}
\caption{Representation of a portion of the infinite grid as a planar directed 
graph $G=(V,A)$: vertices correspond to cells, and each undirected edge 
represents a pair of arcs in both directions between adjacent cells.}
\label{fig:gg}
\end{figure}
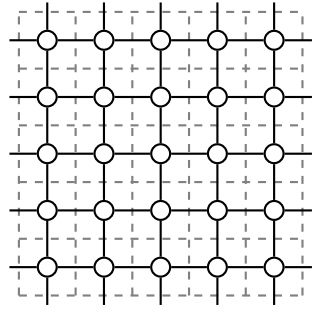

In all figures cited in the remainder of this paper, black cells are in the $L$ state and white cells in the $R$ state. Gray cells represent cells whose state may be either $L$ or $R$ and is not yet fixed.

A \emph{color configuration} of the grid is a function $C :  \mathbb{Z}^2 \rightarrow \{L,R\}$ assigning to each cell its state, where $L$ denotes the turn-left  state and $R$ denotes the turn-right state.
A \emph{configuration}  is a  pair $(C, pos)$ consisting of a color configuration of the grid together with the position $pos \in A$ currently occupied by the ant.

The global \emph{transition function} $T$ is defined by $T(C, pos) = (C', pos')$,
where, given $pos = (u, v)$:
\begin{itemize}
    \item $C'(w) = C(w)$ for all $w \neq v$, and $C'(v) = \overline{C(v)}$, 
    where $\overline{L} = R$ and $\overline{R} = L$;
    \item $pos' = (v, w)$, where $w$ is the neighbor of $v$ obtained by 
    turning right from the direction $(u \to v)$ if $C(v) = R$, or turning 
    left if $C(v) = L$.
\end{itemize}

Applying $T$ once corresponds to executing one step of the ant's motion. See Figure \ref{fig:transitions} for an illustration of the application of transition function.

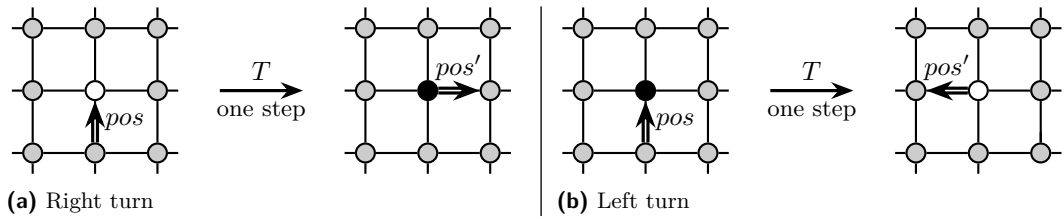
\begin{figure}[htbp]
\centering

\begin{minipage}{0.48\textwidth}
\centering
\begin{tikzpicture}[
scale=0.55,
vertex/.style={circle, draw, minimum size=2.5mm, inner sep=0pt},
fvertex/.style={circle, draw, fill=black, minimum size=2.5mm, inner sep=0pt},
gvertex/.style={circle, draw, fill=gray!40, minimum size=2.5mm, inner sep=0pt}, 
>=Stealth, thick
]

\node[gvertex] (0v0) at (0,0) {};
\node[gvertex] (0v1) at (0,1.5) {};
\node[gvertex] (0v2) at (0,3) {};

\node[gvertex] (1v0) at (1.5,0) {};
\node[vertex] (1v1) at (1.5,1.5) {};
\node[gvertex] (1v2) at (1.5,3) {};

\node[gvertex] (2v0) at (3,0) {};
\node[gvertex] (2v1) at (3,1.5) {};
\node[gvertex] (2v2) at (3,3) {};

\draw (0v0)--(1v0)--(2v0);
\draw (0v1)--(1v1)--(2v1);
\draw (0v2)--(1v2)--(2v2);

\draw (0v0)--(0v1)--(0v2);
\draw (1v0)--(1v1)--(1v2);
\draw (2v0)--(2v1)--(2v2);

\draw (-0.5,0)--(0v0);
\draw (-0.5,1.5)--(0v1);
\draw (-0.5,3)--(0v2);

\draw (3.5,0)--(2v0);
\draw (3.5,1.5)--(2v1);
\draw (3.5,3)--(2v2);

\draw (0,-0.5)--(0v0);
\draw (1.5,-0.5)--(1v0);
\draw (3,-0.5)--(2v0);

\draw (0,3.5)--(0v2);
\draw (1.5,3.5)--(1v2);
\draw (3,3.5)--(2v2);

\draw[->, very thick, double] (1v0) -- node[right] {$pos$} (1v1);

\draw[->, very thick] (4.5,1.5) -- node[above] {$T$} node[below] {\small one step} (6.5,1.5);

\node[gvertex] (R0v0) at (8,0) {};
\node[gvertex] (R0v1) at (8,1.5) {};
\node[gvertex] (R0v2) at (8,3) {};

\node[gvertex] (R1v0) at (9.5,0) {};
\node[fvertex] (R1v1) at (9.5,1.5) {};
\node[gvertex] (R1v2) at (9.5,3) {};

\node[gvertex] (R2v0) at (11,0) {};
\node[gvertex] (R2v1) at (11,1.5) {};
\node[gvertex] (R2v2) at (11,3) {};

\draw (R0v0)--(R1v0)--(R2v0);
\draw (R0v1)--(R1v1)--(R2v1);
\draw (R0v2)--(R1v2)--(R2v2);

\draw (R0v0)--(R0v1)--(R0v2);
\draw (R1v0)--(R1v1)--(R1v2);
\draw (R2v0)--(R2v1)--(R2v2);

\draw (7.5,0)--(R0v0);
\draw (7.5,1.5)--(R0v1);
\draw (7.5,3)--(R0v2);

\draw (11.5,0)--(R2v0);
\draw (11.5,1.5)--(R2v1);
\draw (11.5,3)--(R2v2);

\draw (8,-0.5)--(R0v0);
\draw (9.5,-0.5)--(R1v0);
\draw (11,-0.5)--(R2v0);

\draw (8,3.5)--(R0v2);
\draw (9.5,3.5)--(R1v2);
\draw (11,3.5)--(R2v2);

\draw[->, very thick, double] (R1v1) -- node[above] {$pos'$} (R2v1);

\end{tikzpicture}
\subcaption{Right turn}
\end{minipage}
\hfill
\vrule width 0.5pt
\hfill
\begin{minipage}{0.48\textwidth}
\centering
\begin{tikzpicture}[
scale=0.55,
vertex/.style={circle, draw, minimum size=2.5mm, inner sep=0pt},
fvertex/.style={circle, draw, fill=black, minimum size=2.5mm, inner sep=0pt},
gvertex/.style={circle, draw, fill=gray!40, minimum size=2.5mm, inner sep=0pt}, 
>=Stealth, thick
]

\node[gvertex] (0v0) at (0,0) {};
\node[gvertex] (0v1) at (0,1.5) {};
\node[gvertex] (0v2) at (0,3) {};

\node[gvertex] (1v0) at (1.5,0) {};
\node[fvertex] (1v1) at (1.5,1.5) {};
\node[gvertex] (1v2) at (1.5,3) {};

\node[gvertex] (2v0) at (3,0) {};
\node[gvertex] (2v1) at (3,1.5) {};
\node[gvertex] (2v2) at (3,3) {};

\draw (0v0)--(1v0)--(2v0);
\draw (0v1)--(1v1)--(2v1);
\draw (0v2)--(1v2)--(2v2);

\draw (0v0)--(0v1)--(0v2);
\draw (1v0)--(1v1)--(1v2);
\draw (2v0)--(2v1)--(2v2);

\draw (-0.5,0)--(0v0);
\draw (-0.5,1.5)--(0v1);
\draw (-0.5,3)--(0v2);

\draw (3.5,0)--(2v0);
\draw (3.5,1.5)--(2v1);
\draw (3.5,3)--(2v2);

\draw (0,-0.5)--(0v0);
\draw (1.5,-0.5)--(1v0);
\draw (3,-0.5)--(2v0);

\draw (0,3.5)--(0v2);
\draw (1.5,3.5)--(1v2);
\draw (3,3.5)--(2v2);

\draw[->, very thick, double] (1v0) -- node[right] {$pos$} (1v1);

\draw[->, very thick] (4.5,1.5) -- node[above] {$T$} node[below] {\small one step} (6.5,1.5);

\node[gvertex] (R0v0) at (8,0) {};
\node[gvertex] (R0v1) at (8,1.5) {};
\node[gvertex] (R0v2) at (8,3) {};

\node[gvertex] (R1v0) at (9.5,0) {};
\node[vertex] (R1v1) at (9.5,1.5) {};
\node[gvertex] (R1v2) at (9.5,3) {};

\node[gvertex] (R2v0) at (11,0) {};
\node[gvertex] (R2v1) at (11,1.5) {};
\node[gvertex] (R2v2) at (11,3) {};

\draw (R0v0)--(R1v0)--(R2v0);
\draw (R0v1)--(R1v1)--(R2v1);
\draw (R0v2)--(R1v2)--(R2v2);

\draw (R0v0)--(R0v1)--(R0v2);
\draw (R1v0)--(R1v1)--(R1v2);
\draw (R2v0)--(R2v1)--(R2v2);

\draw (7.5,0)--(R0v0);
\draw (7.5,1.5)--(R0v1);
\draw (7.5,3)--(R0v2);

\draw (11.5,0)--(R2v0);
\draw (11.5,1.5)--(R2v1);
\draw (11.5,3)--(R2v2);

\draw (8,-0.5)--(R0v0);
\draw (9.5,-0.5)--(R1v0);
\draw (11,0.5)--(R2v0);

\draw (8,3.5)--(R0v2);
\draw (9.5,3.5)--(R1v2);
\draw (11,3.5)--(R2v2);

\draw[->, very thick, double] (R1v1) -- node[above] {$pos'$} (R0v1);

\end{tikzpicture}
\subcaption{Left turn}
\end{minipage}

\caption{Configuration before and after one application of the transition function $T$. In~(a), the ant arrives at a white cell, turns right, flips the cell to black, and moves forward. In~(b), the ant arrives at a black cell, turns left, flips the cell to white, and moves 
forward.}
\label{fig:transitions}
\end{figure}

\begin{remark}
An important property is that the dynamics of the ant is deterministic and \emph{reversible}, i.e. the transition function is bijective. Configurations at each time in the future and in the past are entirely determined by the current configuration.
To determine the previous configuration, it is enough to invert the direction of the ant, to apply the transition function $T$, and to invert once more the direction of the ant.
\end{remark}

\subsection{HV-partition Property}
Another fundamental structural property of Langton's ant is what we call the \emph{HV-partition property}. The ant’s orientation alternates between horizontal and vertical, since its direction is rotated by a quarter turn at each step.  Color the cells of $G$ in a checkerboard pattern by assigning one color to all cells $(i,j)$ with $i+j$ even, and the other color to those with $i+j$ odd. If the ant starts in a horizontal orientation and is pointing toward a cell of the first color, then it will always be in a horizontal orientation when pointing toward a cell of that color, and in a vertical orientation when pointing toward a cell of the other color.
This induces a natural partition of $\mathbb{Z}^2$ into two classes, \emph{H-cells} and \emph{V-cells}: when the ant points toward an H-cell (resp.\ a V-cell), it arrives horizontally (resp.\ vertically) and departs vertically  (resp.\ horizontally).  See Figure \ref{fig:HVpartition} for illustration.

\begin{figure}[H]
\centering
\begin{tikzpicture}[
scale=0.5,
vcell/.style={circle, draw, fill=gray!40, minimum size=2.5mm, inner sep=0pt},
hcell/.style={rectangle, draw, fill=gray!40, minimum size=2.5mm, inner sep=0pt},
>=Stealth,
thick
]
\def\k{5}
\def\n{5}
\def\spacing{1.5}
\def\externalarrow{1.0}

\foreach \i in {0,...,\numexpr\n-1} {
    \foreach \j in {0,...,\numexpr\k-1} {
        \pgfmathsetmacro{\x}{\i*\spacing}
        \pgfmathsetmacro{\y}{\j*\spacing}
        \pgfmathsetmacro{\parity}{int(mod(\i+\j,2))}
        \ifnum\parity=0
            \node[hcell] (\i v\j) at (\x,\y) {};
        \else
            \node[vcell] (\i v\j) at (\x,\y) {};
        \fi
    }
}

\foreach \j in {0,...,\numexpr\k-1} {
    \foreach \i in {0,...,\numexpr\n-2} {
        \pgfmathsetmacro{\parity}{int(mod(\i+\j,2))}
        \ifnum\parity=0
            \draw[<-] (\i v\j) -- (\the\numexpr\i+1\relax v\j);
        \else
            \draw[->] (\i v\j) -- (\the\numexpr\i+1\relax v\j);
        \fi
    }
}

\foreach \i in {0,...,\numexpr\n-1} {
    \foreach \j in {0,...,\numexpr\k-2} {
        \pgfmathsetmacro{\parity}{int(mod(\i+\j,2))}
        \ifnum\parity=0
            \draw[->] (\i v\j) -- (\i v\the\numexpr\j+1\relax);
        \else
            \draw[<-] (\i v\j) -- (\i v\the\numexpr\j+1\relax);
        \fi
    }
}

\foreach \j in {0,...,\numexpr\k-1} {
    \pgfmathsetmacro{\y}{\j*\spacing}
    \pgfmathsetmacro{\parity}{int(mod(\j,2))}
    \ifnum\parity=0
        \draw[->] (-\externalarrow,\y) -- (0v\j);
    \else
        \draw[<-] (-\externalarrow,\y) -- (0v\j);
    \fi
}

\foreach \j in {0,...,\numexpr\k-1} {
    \pgfmathsetmacro{\y}{\j*\spacing}
    \pgfmathsetmacro{\xright}{(\n-1)*\spacing+\externalarrow}
    \pgfmathsetmacro{\parity}{int(mod(\n-1+\j,2))}
    \ifnum\parity=0
        \draw[->] (\xright,\y) -- (\the\numexpr\n-1\relax v\j);
    \else
        \draw[->] (\the\numexpr\n-1\relax v\j) -- (\xright,\y);
    \fi
}

\foreach \i in {0,...,\numexpr\n-1} {
    \pgfmathsetmacro{\x}{\i*\spacing}
    \pgfmathsetmacro{\parity}{int(mod(\i,2))}
    \ifnum\parity=0
        \draw[->] (\i v0) -- (\x,-\externalarrow);
    \else
        \draw[->] (\x,-\externalarrow) -- (\i v0);
    \fi
}

\foreach \i in {0,...,\numexpr\n-1} {
    \pgfmathsetmacro{\x}{\i*\spacing}
    \pgfmathsetmacro{\ytop}{(\k-1)*\spacing+\externalarrow}
    \pgfmathsetmacro{\parity}{int(mod(\i+\k-1,2))}
    \ifnum\parity=0
        \draw[->] (\i v\the\numexpr\k-1\relax) -- (\x,\ytop);
    \else
        \draw[->] (\x,\ytop) -- (\i v\the\numexpr\k-1\relax);
    \fi
}

\end{tikzpicture}
\caption{Illustration of the HV-partition property under the assumption that the ant enters horizontally into the cell located in the last row of the first column.  V-cells are represented as circles and H-cells as rectangles. Horizontal arcs point toward H-cells, and vertical arcs point toward V-cells. The set of arcs thus specifies exactly all admissible positions of the ant.}
\label{fig:HVpartition}
\end{figure}
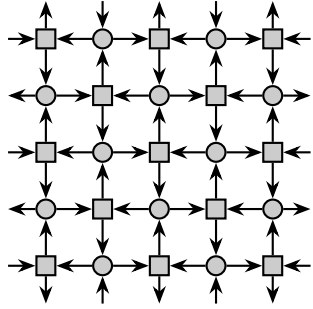
This partition of the cells induces a corresponding partition of the arc set $A$ 
into two subsets $A_H$ and $A_V$, where $A_H$ (resp.\ $A_V$) is the set of admissible positions of the ant when the origin cell $(0,0)$ is taken as an H-cell (resp.\ a V-cell). So according to the starting position of the ant the motion of the ant will be either supported by the graph  $(\mathbb{Z}^2,A_H)$ (horizontal $G$) or $(\mathbb{Z}^2,A_V)$ (vertical $G$). In these graphs, every cell has exactly two incoming arcs and exactly two outgoing arcs.

\subsection{Finite Connected Domain}
We say that a finite set of cells $V_D \subset \mathbb{Z}^2$ is 
\emph{connected} if for any two cells $u, v \in V_D$, there exists a 
sequence $u = v_0, v_1, \ldots, v_k = v$ in $V_D$ with $k\in \mathbb{N}$ such that $v_{i+1}$ is adjacent to $v_i$ for every $0\le i \le k-1$.

For any such set, we define the associated \emph{finite connected domain} $D = (V_D, A_D)$, where $A_D$ consists of all arcs of $A$ incident to at least one cell of $V_D$. Three types of 
arcs in $A_D$ can be distinguished:
\emph{interior arcs}: both endpoints in $V_D$; \emph{in-boundary arcs}: tail outside $V_D$, head in $V_D$ and \emph{out-boundary arcs}: tail in $V_D$, head outside $V_D$.

The HV-partition of $G$ restricts naturally to $D$: the arc set $A_D$ 
is partitioned into $A_D \cap A_H$ and $A_D \cap A_V$, identifying 
within $D$ the admissible positions of the ant under each of the two 
partitions. We define \emph{horizontal} $D$ 
as $(V_D, A_D \cap A_H)$ and \emph{vertical} $D$ as $(V_D, A_D \cap 
A_V)$. Henceforth, unless stated 
otherwise, $D$ refers to horizontal $D$ (i.e $A_D= A_D \cap A_H$).

A cell of $D$ is called a \emph{boundary cell} if it is incident to 
at least one boundary arc (either an in-boundary or an out-boundary arc), and a \emph{corner cell} if it is incident to both an in-boundary and an out-boundary arc. See 
Figure~\ref{fig:finite} for an illustration.

\begin{figure}[htbp]
\centering
\begin{tikzpicture}[
scale=0.5,
vertex/.style={circle, draw, fill=gray!40, minimum size=2.5mm, inner sep=0pt},
boundary/.style={rectangle, draw, fill=gray!40, minimum size=2.5mm, inner sep=0pt},
corner/.style={diamond, draw, fill=gray!40, minimum size=2.5mm, inner sep=0pt},
>=Stealth,
thick
]

\def\k{4}
\def\n{6}
\def\spacing{1.5}
\def\externalarrow{1.2}

\foreach \i in {0,...,\numexpr\n-1} {
    \foreach \j in {0,...,\numexpr\k-1} {
        \pgfmathsetmacro{\x}{\i*\spacing}
        \pgfmathsetmacro{\y}{\j*\spacing}
        \node[vertex] (\i v\j) at (\x,\y) {};
    }
}

\foreach \i/\j in {
0/0, 0/1, 0/2, 0/3,
5/0, 5/1, 5/2, 5/3,
1/0, 2/0, 3/0, 4/0,
1/3, 2/3, 3/3, 4/3
} {
    \pgfmathsetmacro{\x}{\i*\spacing}
    \pgfmathsetmacro{\y}{\j*\spacing}
    \node[boundary] (\i v\j) at (\x,\y) {};
}

\foreach \i/\j in {0/0, 0/3, 5/0, 5/3} {
    \pgfmathsetmacro{\x}{\i*\spacing}
    \pgfmathsetmacro{\y}{\j*\spacing}
    \node[corner] (\i v\j) at (\x,\y) {};
}

\foreach \j in {0,...,\numexpr\k-1} {
    \pgfmathsetmacro{\y}{\j*\spacing}
    
    \pgfmathsetmacro{\dir}{int(mod(\j,2)==0?1:0)}
    \ifnum\dir=1
        \draw[->, dashed] (-\externalarrow,\y) -- (0v\j);
    \else
        \draw[<-, dotted] (-\externalarrow,\y) -- (0v\j);
    \fi
    
    \foreach \i in {0,...,\numexpr\n-2} {
        \pgfmathsetmacro{\arrowdir}{int(mod(\i+\j,2)==0?1:0)}
        \ifnum\arrowdir=1
            \draw[<-] (\i v\j) -- (\the\numexpr\i+1\relax v\j);
        \else
            \draw[->] (\i v\j) -- (\the\numexpr\i+1\relax v\j);
        \fi
    }
    
    \pgfmathsetmacro{\xright}{(\n-1)*\spacing+\externalarrow}
    \pgfmathsetmacro{\lastdir}{int(mod(\n-1+\j,2)==0?1:0)}
    \ifnum\lastdir=1
        \draw[<-, dashed] (\the\numexpr\n-1\relax v\j) -- (\xright,\y);
    \else
        \draw[->, dotted] (\the\numexpr\n-1\relax v\j) -- (\xright,\y);
    \fi
}

\foreach \i in {0,...,\numexpr\n-1} {
    \pgfmathsetmacro{\x}{\i*\spacing}
    
    \pgfmathsetmacro{\dir}{int(mod(\i,2)==0?1:0)}
    \ifnum\dir=1
        \draw[<-, dotted] (\x,-\externalarrow) -- (\i v0);
    \else
        \draw[->, dashed] (\x,-\externalarrow) -- (\i v0);
    \fi
    
    \foreach \j in {0,...,\numexpr\k-2} {
        \pgfmathsetmacro{\arrowdir}{int(mod(\i+\j,2)==0?1:0)}
        \ifnum\arrowdir=1
            \draw[->] (\i v\j) -- (\i v\the\numexpr\j+1\relax);
        \else
            \draw[<-] (\i v\j) -- (\i v\the\numexpr\j+1\relax);
        \fi
    }
    
    \pgfmathsetmacro{\yup}{(\k-1)*\spacing+\externalarrow}
    \pgfmathsetmacro{\lastdir}{int(mod(\i+\k-1,2)==0?1:0)}
    \ifnum\lastdir=1
        \draw[->, dotted] (\i v\the\numexpr\k-1\relax) -- (\x,\yup);
    \else
        \draw[<-, dashed] (\i v\the\numexpr\k-1\relax) -- (\x,\yup);
    \fi
}

\end{tikzpicture}

\caption{A horizontal finite connected domain of four rows and six columns. Interior cells are shown as circles, boundary cells as rectangles, and corner cells as diamonds; interior arcs are drawn as solid lines, in-boundary arcs as dashed lines and out-boundary arcs as dotted lines. Arc directions are determined by HV-partition.}
\label{fig:finite}
\end{figure}
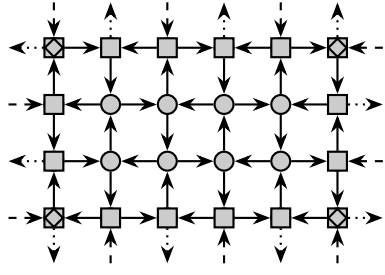

A color configuration of $D$ is a function $C_D: V_D \rightarrow \{L,R\}$, and a configuration of $D$ is a pair $(C_D, pos)$, where  $pos \in A_D$.
The ant is said to be \emph{confined} in $D$ as long as its position is an interior or in-boundary arc of $D$, and it \emph{escapes} $D$ at the first step at which its position is an out-boundary arc of $D$. The transition function $T$ restricts naturally to $D$.

Given an initial configuration $(C_0, pos_0)$ of $D$ with $pos_0=(v_0, v_1)$,  for $k\in \mathbb{N}$, we denote by $(C_k, pos_k)=T^k (C_0, pos_0)$ with $pos_k=(v_k,v_{k+1})$, the configuration after $k$ steps.
The \emph{trajectory} of the ant from $(C_0, pos_0)$  to  $(C_k, pos_k)$ is the sequence of configurations $(C_i, pos_i)_{0 \le i \le k}$. The \emph{set of exited cells} during this trajectory is  the set $\{v_i \mid 1 \le i \le k \}$ consisting of all cells that the ant leaves during its motion.
The positions $pos_0$ and $pos_k$ are the starting and ending positions of the trajectory, respectively, and $v_1$ and $v_k$ are called the starting and ending cells.

The \emph{motion of the ant within $D$} from $(C_0, pos_0)$ 
refers to any trajectory $(C_i, pos_i)_{0 \le i \le k}$ such that each  $pos_i$ is an arc of $D$.

The \emph{number of steps performed} by 
the ant within $D$ starting from $(C_0, pos_0)$ is:
$$S_D(C_0, pos_0) = \min\,\{k \in \mathbb{N} \mid v_{k+1} \notin V_D\},$$ the number of steps during which it remains confined in $D$.
The \emph{escaping time} of $D$ is defined as: $$S_D = \max\,\{S_D(C_0, pos_0) \mid (C_0, pos_0) \text{ configuration of } D\},$$ that is, the maximum number of steps the ant can perform within $D$ over 
all initial configurations. The central question addressed in this paper is the following:

\begin{quote}
\textbf{Escaping Time Problem.} \emph{Given a finite connected domain $D$ 
of the grid, what is the escaping time of $D$?}
\end{quote}

\subsection{Related Work}

The most fundamental result on Langton's ant on the infinite square grid, due to Troubetzkoy~\cite{troubetzkoy1997ant}, states that for any initial configuration, if the ant moves indefinitely, the set of cells exited infinitely often contains no corner cell. As a corollary~\cite{bunimovich1992recurrence}, the ant's trajectory is always \emph{unbounded} (non-periodic). Consequently, the ant escapes any finite connected domain in a finite number of steps. The main objective of this paper is to quantify this escaping time.

The other main theoretical results concern the computational complexity of 
the ant. 
Gajardo, Moreira, and Goles~\cite{gajardo2002complexity} showed that Boolean circuits can be simulated within Langton's ant by embedding logic gadgets into the initial color configuration of the grid, with the ant's trajectory serving as the computational signal. This yields the $\mathsf{P}$-hardness under 
log-space reductions of the reachability problem: given an initial 
configuration, does the ant ever visit a specified target cell? By further 
composing these gadgets to simulate one-dimensional cellular automata, they 
established that Langton's ant is computationally universal, which in turn 
implies the existence of undecidable problems about its trajectory.

When the ant is restricted to a finite domain, this reachability problem becomes  decidable. But what is its precise complexity? In particular, is it $\mathsf{P}$-complete? If the escaping time was polynomially bounded, then reachability on finite grids can be decided by direct simulation in 
polynomial time, suggesting that the problem lies strictly within $\mathsf{P}$.

\medskip
Tsukiji and Hagiwara established $\mathsf{PSPACE}$-hardness results concerning the predictability and macroscopic behavior of the ant.
On the square grid ~\cite{tsukiji2011pspace}, by introducing a third cell state in which the ant moves straight ahead without changing the cell's color, they proved that determining whether a given finite configuration is repeatable (i.e. the ant's trajectory is bounded  and eventually cycles through a finite sequence of configurations) is $\mathsf{PSPACE}$-complete.  They also established the $\mathsf{PSPACE}$-hardness of this problem on a hexagonal grid. These results are obtained through a reduction from the Quantified Boolean Formula (QBF) evaluation problem.
In a subsequent study on a twisted torus ~\cite{hagiwara2020torus}, they proved that it is $\mathsf{PSPACE}$-hard to determine whether the ant will ever visit almost all vertices or nearly none of them via a reduction from the Quantified Conjunctive Normal Form (QCNF) problem.
The introduction of the third state on the square grid and the use of the twisted torus topology are crucial for these reductions. These break the standard HV-partition of the square grid and allows the ant to be "reversed", so that it can traverse arcs in both directions. This property is exploited in the reductions through the reversible behavior of the system. 

On the standard two-color square grid, however, the HV-partition is a rigid 
structural constraint: the ant cannot be reversed, and every arc is traversed 
in a fixed direction. This raises the question of whether 
$\mathsf{PSPACE}$-hardness can be achieved without breaking the 
HV-partition, or whether every natural decision problem about the ant on 
the finite two-color square grid is solvable in polynomial time. 
Quantifying the escaping time is a necessary first step toward answering this question.

\section{Escaping Time of Finite Connected Domains}\label{sec:general}

Let $D$ be a finite connected domain of $d \in \mathbb{N}^*$ cells. Since 
each cell can be in one of two states, there are $2^d$ color configurations 
of $D$. By the HV-partition, each cell is the head of exactly two arcs in 
$A_D$, giving $2d$ admissible positions. Together with $2^d$ possible 
color configurations, the total number of configurations of $D$ is 
$2d \cdot 2^d$.
The unboundedness result ensures that the ant's trajectory cannot be periodic within $D$, so no configuration of $D$ can be encountered twice in the trajectory of the ant within $D$, and the escaping time of $D$ is bounded by $2d \cdot 2^d$. However, this bound is far from tight, and the following results provide a slight improvement.

\begin{proposition}\label{pr:no_repeated_coloring}
Starting from a configuration of a finite connected domain $D$, no color configuration of $D$ can appear twice in the motion of the ant within $D$.
\end{proposition}

\begin{proof}
Suppose, for contradiction, that some color configuration $C_D$ appears 
twice during the ant's motion within $D$. Let $Traj$ denote the trajectory of the ant between the first and second occurrences of $C_D$. Since the successive positions of the ant during its trajectory form a connected path, the set of exited cells during $Traj$ induces a finite connected subdomain of $D$, which we denote by $D'$. As the configuration $C_D$ is identical at the beginning and end of $Traj$, the state of each cell in $D'$ must be flipped an even number of times during $Traj$. Consequently, the ant must exit each cell of $D'$ a positive even number of times. The contradiction follows from the two claims below.

\medskip
\textbf{Claim 1.} 
\emph{Every corner cell of $D'$ is either a starting cell or an ending cell of $Traj$.}

\medskip
\textbf{Claim 2.} \emph{$D'$ has at least three distinct corner cells.}

\medskip
Claims~1 and~2 together imply that the trajectory $Traj$ has at
least two distinct starting cells or at least two distinct ending cells, contradicting the fact that each ant trajectory has a unique starting and a unique ending cell.
Thus, no coloring of $D$ can appear twice during the ant motion within $D$.

\medskip
\emph{Proof of Claim 1.}
Every cell of $D'$ has exactly two incoming arcs and exactly two outgoing arcs. 
Let $v_c$ be a corner cell of $D'$, and let $b_{in}$ and $b_{out}$ denote respectively an in-boundary arc and an out-boundary arc of $D'$ incident to $v_c$.  We denote by $a_{in}$ the other incoming arc of $v_c$, and by $a_{out}$ the other outgoing arc of $v_c$. See Figure~\ref{fig:cornercell_a}.

Since the state of $v_c$ alternates at every visit, for the ant to 
exit $v_c$ a positive even number of times, it must exit exactly twice, in 
one of these two possible ways:
\begin{itemize}
    \item The ant enters once via $a_{in}$ and once via $b_{in}$, and 
    exits both times via $a_{out}$ (Figure~\ref{fig:cornercell_b}). 
    Then $b_{in}$ must be the starting position of  $Traj$, so $v_c$ is 
    a starting cell.
    \item The ant enters twice via $a_{in}$, and exits once via 
    $a_{out}$ and once via $b_{out}$ (Figure~\ref{fig:cornercell_c}). 
    Then $b_{out}$ must be the ending position of  $Traj$, so $v_c$ is 
    an ending cell.
\end{itemize}
In both cases, $a_{in}$ and $a_{out}$ are interior arcs of $D'$, since 
 $Traj$ cannot have two starting positions or two ending positions.
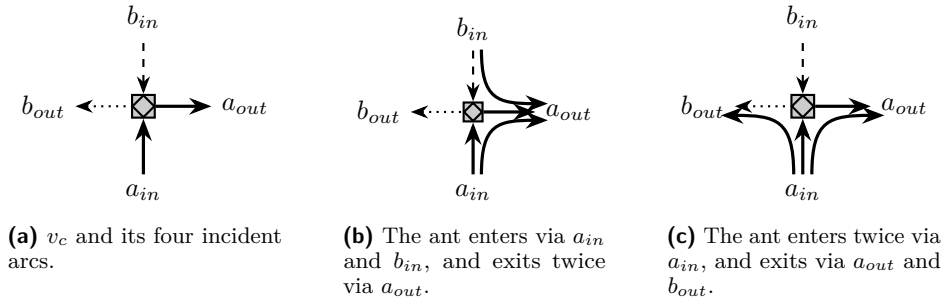
\begin{figure}[H]
\centering
\subcaptionbox{$v_c$ and its four incident arcs.\label{fig:cornercell_a} }{%
\begin{tikzpicture}[
scale=0.6,
corner/.style={diamond, draw, fill=gray!40, minimum size=3mm, inner sep=0pt},
boundary/.style={rectangle, draw, fill=gray!40, minimum size=3mm, inner sep=0pt},
>=Stealth,
thick
]
\def\ext{1.5}
\node[boundary] (vc) at (0,0) {};
\node[corner] (vc) at (0,0) {};
\draw[<-, dotted] (-\ext,0) -- node[at start, left] {$b_{out}$} (vc);
\draw[very thick,<-] (\ext,0) -- node[at start, right] {$a_{out}$} (vc);
\draw[<-, dashed] (vc) -- node[at end, above] {$b_{in}$} (0,\ext);
\draw[very thick,<-] (vc) -- node[at end, below] {$a_{in}$} (0,-\ext);
\end{tikzpicture}
}%
\hspace{2em}
\subcaptionbox{The ant enters via $a_{in}$ and $b_{in}$, and 
exits twice via $a_{out}$.\label{fig:cornercell_b}}{%
\begin{tikzpicture}[
scale=0.55,
corner/.style={diamond, draw, fill=gray!40, minimum size=2.5mm, inner sep=0pt},
boundary/.style={rectangle, draw, fill=gray!40, minimum size=2.5mm, inner sep=0pt},
>=Stealth,
thick
]
\def\ext{1.5}
\node[boundary] (vc) at (0,0) {};
\node[corner] (vc) at (0,0) {};
\draw[<-, dotted] (-\ext,0) -- node[at start, left] {$b_{out}$} (vc);
\draw[very thick,<-] (\ext,0) -- node[at start, right] {$a_{out}$} (vc);
\draw[<-, dashed] (vc) -- node[at end, above] {$b_{in}$} (0,\ext);
\draw[very thick,<-] (vc) -- node[at end, below] {$a_{in}$} (0,-\ext);
\draw[very thick, ->] (0.2,\ext) to[out=-90, in=-180, looseness=1.5] (\ext+0.3,0.2);
\draw[very thick, ->] (0.2,-\ext) to[out=90, in=180, looseness=1.5] (\ext+0.3,-0.2);
\end{tikzpicture}
}%
\hspace{2em}
\subcaptionbox{The ant enters twice via $a_{in}$, and exits 
via $a_{out}$ and $b_{out}$.\label{fig:cornercell_c}}{%
\begin{tikzpicture}[
scale=0.6,
corner/.style={diamond, draw, fill=gray!40, minimum size=3mm, inner sep=0pt},
boundary/.style={rectangle, draw, fill=gray!40, minimum size=3mm, inner sep=0pt},
>=Stealth,
thick
]
\def\ext{1.5}
\node[boundary] (vc) at (0,0) {};
\node[corner] (vc) at (0,0) {};
\draw[<-, dotted] (-\ext,0) -- node[at start, left] {$b_{out}$} (vc);
\draw[very thick,<-] (\ext,0) -- node[at start, right] {$a_{out}$} (vc);
\draw[<-, dashed] (vc) -- node[at end, above] {$b_{in}$} (0,\ext);
\draw[very thick,<-] (vc) -- node[at end, below] {$a_{in}$} (0,-\ext);

\draw[very thick, ->] (-0.2,-\ext) to[out=90, in=0, looseness=1.5] (-\ext-0.3,-0.2);
\draw[very thick, ->] (0.2,-\ext) to[out=90, in=180, looseness=1.5] (\ext+0.3,-0.2);
\end{tikzpicture}
}%
\caption{A boundary cell $v_c$ (here a V-cell) and its four incident arcs (solid arcs are interior arcs, dashed arcs are in-boundary 
arcs, and dotted arcs are out-boundary arcs)~(a); together 
with the two possible ways the ant can visit $v_c$ an even positive number of 
times~(b,c). Curved solid arrows indicate the motion of the ant 
through $v_c$.}
\label{fig:cornercell}
\end{figure}

\emph{Proof of Claim 2.}
Without loss of generality, assume $v_c$ is the topmost among the leftmost cells of $D'$. By Claim~1, the two possibilities for the ant to exit $v_c$ a positive even number of times involves $a_{in}$ and $a_{out}$ as interior arcs, so $D'$ contains at least one cell to the right of $v_c$ and at least one cell below $v_c$. Consequently, the topmost among the rightmost cells of $D'$ and the leftmost among the bottommost cells of $D'$ are both distinct from $v_c$ and from each other, giving at least three distinct corner cells of $D'$.

This completes the proof of Proposition~\ref{pr:no_repeated_coloring}.
\end{proof}

\begin{corollary}
\label{pr:exit_bound}
For any finite connected domain $D$ of $d \in \mathbb{N}^*$ cells, one has  $S_D \le 2^d$.
\end{corollary}

\begin{proof}
Suppose, for the sake of contradiction, that the ant is confined in $D$ for more than $2^d$ steps. Since there are exactly $2^d$ distinct color configurations of $D$, the pigeonhole principle implies that some coloring $C_D$ must occur at least twice during the ant's motion within $D$. This contradicts Proposition~\ref{pr:no_repeated_coloring}.
\end{proof}
This bound is, however, also far from being achieved in practice. In the 
following sections, we present simulation results on the escaping time over rectangular domains, which suggest that the true maximum grows 
significantly slower than $2^d$, and we establish rigorous upper bounds for square connected domains and some rectangular connected domains.

We end this section by pointing out that:
\begin{remark}\label{rm:inb_max}
In order to maximize the number of steps performed within a
finite connected domain $D$, the ant must start from an in-boundary arc of $D$. 
\end{remark}
Indeed, if the maximal number of steps were achieved by an ant starting from an interior arc, then, by reversibility of the transition function, this number could be extended, contradicting maximality.

\section{Escaping Time of Square and Rectangular Connected Domains}\label{sec:squarerectangular}

We now focus on rectangular finite connected domains. For $k, n \in \mathbb{N}^*$, we denote by $G_{k,n}$ the rectangular domain with $k$ rows and $n$ columns. The \emph{horizontal} $G_{k,n}$ is the graph equipped with the HV-partition in which the cell in the last row of the first column is an H-cell; the alternative partition is referred to as the \emph{vertical} $G_{k,n}$.
We denote by $S_{k,n}$ the maximum escaping time over both horizontal $G_{k,n}$ and vertical $G_{k,n}$.

The goal of this section is twofold: we first present exact values of 
$S_{k,n}$ obtained by computer simulation, which serve both as a 
reference and as a guide for the theoretical bounds established 
thereafter; we then prove rigorous bounds for square connected domains and some rectangular connected domains.

\subsection{Exact Escape Times via Simulation on Small Grids}\label{sec:simulation}

We computed the exact values of $S_{k,n}$ for $1 \le k, n \le 9$ 
(Table~\ref{tab:S99}). The  values $S_{n,n} - 1$ up to $n = 7$ 
were previously listed as sequence A282425 in the OEIS~\cite{oeisA282425}, 
with an error of 1 for $S_{7,7}$; our values correct this entry.

\begin{table}[H]
\centering
\begin{tabular}{c|c|c|c|c|c|c|c|c|c|}
 & 1 & 2 & 3 & 4 & 5 & 6 & 7 & 8 & 9 \\
\hline
1 & 1 & 2 & 2 & 2 & 2 & 2 & 2 & 2 & 2\\
\hline
2 & & 6 & 12 & 18 & 24 & 30 & 36 & 42 & 48 \\
\hline
3 & & & 17 & 28 & 37 & 48 & 57 & 68 & 77 \\
\hline
4 & & & & 46 & 64 & 86 & 104 & 146 & 156\\
\hline
5 & & & & & 85 & 130 & 145 & 214 & 221\\
\hline
6 & & & & & & 164 & 220 & 280 &  315\\
\hline
7 & & & & & & & 262 & 356 &410 \\
\hline
8 & & & & & & & & 488 &618 \\
\hline
9 & & & & & & & &  & 679 \\
\hline
\end{tabular}
\caption{Values of $S_{k,n}$ for $1\le k,n \le 9$.}
\label{tab:S99}
\end{table}

\paragraph*{Simulation Algorithm}
By Remark~\ref{rm:inb_max}, it suffices to consider initial 
configurations where the ant starts from an in-boundary arc. Computing the escape time $S_{k,n}$ amounts to evaluating, over all initial configurations $(C, pos)$, the number of steps before escape, where $pos$ ranges over in-boundary arcs and $C$ over all possible color configurations of $G_{k,n}$.
A naive exhaustive search over all $2^{kn}$ color configurations for each 
starting position is prohibitively expensive; we reduce this cost via 
two complementary optimizations.

\textbf{Branching exploration.}  
For a fixed starting position $pos$, instead of simulating the ant independently for every color configuration $C$, we exploit the fact that the choice of the state of a cell only matters upon its first visit. The exploration can thus be structured as a binary decision tree: whenever the ant visits a cell for the first time, the algorithm branches into two cases corresponding to the two possible initial states of that cell. Each branch then evolves deterministically according to Langton’s rule.
In this way, all configurations that coincide on the set of already visited cells share the same trajectory prefix and are explored simultaneously. This avoids redundant recomputation and allows the entire configuration space to be traversed in a single branching exploration.

\begin{figure}[H]
\centering

%
%

\newcommand{\diagS}{%
\begin{tikzpicture}[
scale=0.5,
vertex/.style={circle, draw, minimum size=2.5mm, inner sep=0pt},
fvertex/.style={circle, draw, fill=black, minimum size=2.5mm, inner sep=0pt},
gvertex/.style={circle, draw, fill=gray!60, minimum size=2.5mm, inner sep=0pt},
>=Stealth, thick]
\node[gvertex] (0v0) at (0,0) {};
\node[gvertex] (1v0) at (1.5,0) {};
\node[gvertex] (0v1) at (0,1.5) {};
\node[gvertex] (1v1) at (1.5,1.5) {};
\draw[gray] (-0.75,-0.75) -- (2.25,-0.75);
\draw[gray] (-0.75, 2.25) -- (2.25, 2.25);
\draw[gray] (-0.75,-0.75) -- (-0.75, 2.25);
\draw[gray] ( 2.25,-0.75) -- ( 2.25, 2.25);
\draw[<-] (-1.2,0) -- (0v0);
\draw[->] (0v0) -- (1v0);
\draw[<-] (1v0) -- (2.7,0);
\draw[->, double, very thick, red] (-1.2,1.5) -- (0v1);
\draw[<-] (0v1) -- (1v1);
\draw[->] (1v1) -- (2.7,1.5);
\draw[->] (0,-1.2) -- (0v0);
\draw[<-] (0v0) -- (0v1);
\draw[->] (0v1) -- (0,2.7);
\draw[<-] (1.5,-1.2) -- (1v0);
\draw[->] (1v0) -- (1v1);
\draw[<-] (1v1) -- (1.5,2.7);
\node at (0.75,-1.6) {\small $S$};
\end{tikzpicture}}

\newcommand{\diagR}{%
\begin{tikzpicture}[
scale=0.5,
vertex/.style={circle, draw, minimum size=2.5mm, inner sep=0pt},
fvertex/.style={circle, draw, fill=black, minimum size=2.5mm, inner sep=0pt},
gvertex/.style={circle, draw, fill=gray!60, minimum size=2.5mm, inner sep=0pt},
>=Stealth, thick]
\node[gvertex] (0v0) at (0,0) {};
\node[gvertex] (1v0) at (1.5,0) {};
\node[fvertex] (0v1) at (0,1.5) {};
\node[gvertex] (1v1) at (1.5,1.5) {};
\draw[gray] (-0.75,-0.75) -- (2.25,-0.75);
\draw[gray] (-0.75, 2.25) -- (2.25, 2.25);
\draw[gray] (-0.75,-0.75) -- (-0.75, 2.25);
\draw[gray] ( 2.25,-0.75) -- ( 2.25, 2.25);
\draw[<-] (-1.2,0) -- (0v0);
\draw[->] (0v0) -- (1v0);
\draw[<-] (1v0) -- (2.7,0);
\draw[->] (-1.2,1.5) -- (0v1);
\draw[<-] (0v1) -- (1v1);
\draw[->] (1v1) -- (2.7,1.5);
\draw[->] (0,-1.2) -- (0v0);
\draw[<-, double, very thick, red] (0v0) -- (0v1);
\draw[->] (0v1) -- (0,2.7);
\draw[<-] (1.5,-1.2) -- (1v0);
\draw[->] (1v0) -- (1v1);
\draw[<-] (1v1) -- (1.5,2.7);
\node at (0.75,-1.6) {\small $R$};
\end{tikzpicture}}

\newcommand{\diagL}{%
\begin{tikzpicture}[
scale=0.5,
vertex/.style={circle, draw, minimum size=2.5mm, inner sep=0pt},
fvertex/.style={circle, draw, fill=black, minimum size=2.5mm, inner sep=0pt},
gvertex/.style={circle, draw, fill=gray!60, minimum size=2.5mm, inner sep=0pt},
>=Stealth, thick]
\node[gvertex] (0v0) at (0,0) {};
\node[gvertex] (1v0) at (1.5,0) {};
\node[vertex]  (0v1) at (0,1.5) {};
\node[gvertex] (1v1) at (1.5,1.5) {};
\draw[gray] (-0.75,-0.75) -- (2.25,-0.75);
\draw[gray] (-0.75, 2.25) -- (2.25, 2.25);
\draw[gray] (-0.75,-0.75) -- (-0.75, 2.25);
\draw[gray] ( 2.25,-0.75) -- ( 2.25, 2.25);
\draw[<-] (-1.2,0) -- (0v0);
\draw[->] (0v0) -- (1v0);
\draw[<-] (1v0) -- (2.7,0);
\draw[->] (-1.2,1.5) -- (0v1);
\draw[<-] (0v1) -- (1v1);
\draw[->] (1v1) -- (2.7,1.5);
\draw[->] (0,-1.2) -- (0v0);
\draw[<-] (0v0) -- (0v1);
\draw[->, double, very thick, red] (0v1) -- (0,2.7);
\draw[<-] (1.5,-1.2) -- (1v0);
\draw[->] (1v0) -- (1v1);
\draw[<-] (1v1) -- (1.5,2.7);
\node at (0.75,-1.6) {\small $L$};
\end{tikzpicture}}

\newcommand{\diagRR}{%
\begin{tikzpicture}[
scale=0.5,
vertex/.style={circle, draw, minimum size=2.5mm, inner sep=0pt},
fvertex/.style={circle, draw, fill=black, minimum size=2.5mm, inner sep=0pt},
gvertex/.style={circle, draw, fill=gray!60, minimum size=2.5mm, inner sep=0pt},
>=Stealth, thick]
\node[fvertex] (0v0) at (0,0) {};
\node[gvertex] (1v0) at (1.5,0) {};
\node[fvertex] (0v1) at (0,1.5) {};
\node[gvertex] (1v1) at (1.5,1.5) {};
\draw[gray] (-0.75,-0.75) -- (2.25,-0.75);
\draw[gray] (-0.75, 2.25) -- (2.25, 2.25);
\draw[gray] (-0.75,-0.75) -- (-0.75, 2.25);
\draw[gray] ( 2.25,-0.75) -- ( 2.25, 2.25);
\draw[<-, double, very thick, red] (-1.2,0) -- (0v0);
\draw[->] (0v0) -- (1v0);
\draw[<-] (1v0) -- (2.7,0);
\draw[->] (-1.2,1.5) -- (0v1);
\draw[<-] (0v1) -- (1v1);
\draw[->] (1v1) -- (2.7,1.5);
\draw[->] (0,-1.2) -- (0v0);
\draw[<-] (0v0) -- (0v1);
\draw[->] (0v1) -- (0,2.7);
\draw[<-] (1.5,-1.2) -- (1v0);
\draw[->] (1v0) -- (1v1);
\draw[<-] (1v1) -- (1.5,2.7);
\node at (0.75,-1.6) {\small $R.R$};
\end{tikzpicture}}

\newcommand{\diagRL}{%
\begin{tikzpicture}[
scale=0.5,
vertex/.style={circle, draw, minimum size=2.5mm, inner sep=0pt},
fvertex/.style={circle, draw, fill=black, minimum size=2.5mm, inner sep=0pt},
gvertex/.style={circle, draw, fill=gray!60, minimum size=2.5mm, inner sep=0pt},
>=Stealth, thick]
\node[vertex]  (0v0) at (0,0) {};
\node[gvertex] (1v0) at (1.5,0) {};
\node[fvertex] (0v1) at (0,1.5) {};
\node[gvertex] (1v1) at (1.5,1.5) {};
\draw[gray] (-0.75,-0.75) -- (2.25,-0.75);
\draw[gray] (-0.75, 2.25) -- (2.25, 2.25);
\draw[gray] (-0.75,-0.75) -- (-0.75, 2.25);
\draw[gray] ( 2.25,-0.75) -- ( 2.25, 2.25);
\draw[<-] (-1.2,0) -- (0v0);
\draw[->, double, very thick, red] (0v0) -- (1v0);
\draw[<-] (1v0) -- (2.7,0);
\draw[->] (-1.2,1.5) -- (0v1);
\draw[<-] (0v1) -- (1v1);
\draw[->] (1v1) -- (2.7,1.5);
\draw[->] (0,-1.2) -- (0v0);
\draw[<-] (0v0) -- (0v1);
\draw[->] (0v1) -- (0,2.7);
\draw[<-] (1.5,-1.2) -- (1v0);
\draw[->] (1v0) -- (1v1);
\draw[<-] (1v1) -- (1.5,2.7);
\node at (0.75,-1.6) {\small $R.L$};
\end{tikzpicture}}

\newcommand{\diagRLR}{%
\begin{tikzpicture}[
scale=0.5,
vertex/.style={circle, draw, minimum size=2.5mm, inner sep=0pt},
fvertex/.style={circle, draw, fill=black, minimum size=2.5mm, inner sep=0pt},
gvertex/.style={circle, draw, fill=gray!60, minimum size=2.5mm, inner sep=0pt},
>=Stealth, thick]
\node[vertex]  (0v0) at (0,0) {};
\node[fvertex] (1v0) at (1.5,0) {};
\node[fvertex] (0v1) at (0,1.5) {};
\node[gvertex] (1v1) at (1.5,1.5) {};
\draw[gray] (-0.75,-0.75) -- (2.25,-0.75);
\draw[gray] (-0.75, 2.25) -- (2.25, 2.25);
\draw[gray] (-0.75,-0.75) -- (-0.75, 2.25);
\draw[gray] ( 2.25,-0.75) -- ( 2.25, 2.25);
\draw[<-] (-1.2,0) -- (0v0);
\draw[->] (0v0) -- (1v0);
\draw[<-] (1v0) -- (2.7,0);
\draw[->] (-1.2,1.5) -- (0v1);
\draw[<-] (0v1) -- (1v1);
\draw[->] (1v1) -- (2.7,1.5);
\draw[->] (0,-1.2) -- (0v0);
\draw[<-] (0v0) -- (0v1);
\draw[->] (0v1) -- (0,2.7);
\draw[<-, double, very thick, red] (1.5,-1.2) -- (1v0);
\draw[->] (1v0) -- (1v1);
\draw[<-] (1v1) -- (1.5,2.7);
\node at (0.75,-1.6) {\small $R.L.R$};
\end{tikzpicture}}

\newcommand{\diagRLL}{%
\begin{tikzpicture}[
scale=0.5,
vertex/.style={circle, draw, minimum size=2.5mm, inner sep=0pt},
fvertex/.style={circle, draw, fill=black, minimum size=2.5mm, inner sep=0pt},
gvertex/.style={circle, draw, fill=gray!60, minimum size=2.5mm, inner sep=0pt},
>=Stealth, thick]
\node[vertex]  (0v0) at (0,0) {};
\node[vertex]  (1v0) at (1.5,0) {};
\node[fvertex] (0v1) at (0,1.5) {};
\node[gvertex] (1v1) at (1.5,1.5) {};
\draw[gray] (-0.75,-0.75) -- (2.25,-0.75);
\draw[gray] (-0.75, 2.25) -- (2.25, 2.25);
\draw[gray] (-0.75,-0.75) -- (-0.75, 2.25);
\draw[gray] ( 2.25,-0.75) -- ( 2.25, 2.25);
\draw[<-] (-1.2,0) -- (0v0);
\draw[->] (0v0) -- (1v0);
\draw[<-] (1v0) -- (2.7,0);
\draw[->] (-1.2,1.5) -- (0v1);
\draw[<-] (0v1) -- (1v1);
\draw[->] (1v1) -- (2.7,1.5);
\draw[->] (0,-1.2) -- (0v0);
\draw[<-] (0v0) -- (0v1);
\draw[->] (0v1) -- (0,2.7);
\draw[<-] (1.5,-1.2) -- (1v0);
\draw[->, double, very thick, red] (1v0) -- (1v1);
\draw[<-] (1v1) -- (1.5,2.7);
\node at (0.75,-1.6) {\small $R.L.L$};
\end{tikzpicture}}

\newcommand{\diagRLLR}{%
\begin{tikzpicture}[
scale=0.5,
vertex/.style={circle, draw, minimum size=2.5mm, inner sep=0pt},
fvertex/.style={circle, draw, fill=black, minimum size=2.5mm, inner sep=0pt},
gvertex/.style={circle, draw, fill=gray!60, minimum size=2.5mm, inner sep=0pt},
>=Stealth, thick]
\node[vertex]  (0v0) at (0,0) {};
\node[vertex]  (1v0) at (1.5,0) {};
\node[fvertex] (0v1) at (0,1.5) {};
\node[fvertex] (1v1) at (1.5,1.5) {};
\draw[gray] (-0.75,-0.75) -- (2.25,-0.75);
\draw[gray] (-0.75, 2.25) -- (2.25, 2.25);
\draw[gray] (-0.75,-0.75) -- (-0.75, 2.25);
\draw[gray] ( 2.25,-0.75) -- ( 2.25, 2.25);
\draw[<-] (-1.2,0) -- (0v0);
\draw[->] (0v0) -- (1v0);
\draw[<-] (1v0) -- (2.7,0);
\draw[->] (-1.2,1.5) -- (0v1);
\draw[<-] (0v1) -- (1v1);
\draw[->, double, very thick, red] (1v1) -- (2.7,1.5);
\draw[->] (0,-1.2) -- (0v0);
\draw[<-] (0v0) -- (0v1);
\draw[->] (0v1) -- (0,2.7);
\draw[<-] (1.5,-1.2) -- (1v0);
\draw[->] (1v0) -- (1v1);
\draw[<-] (1v1) -- (1.5,2.7);
\node at (0.75,-1.6) {\small $R.L.L.R$};
\end{tikzpicture}}

\newcommand{\diagRLLL}{%
\begin{tikzpicture}[
scale=0.5,
vertex/.style={circle, draw, minimum size=2.5mm, inner sep=0pt},
fvertex/.style={circle, draw, fill=black, minimum size=2.5mm, inner sep=0pt},
gvertex/.style={circle, draw, fill=gray!60, minimum size=2.5mm, inner sep=0pt},
>=Stealth, thick]
\node[vertex]  (0v0) at (0,0) {};
\node[vertex]  (1v0) at (1.5,0) {};
\node[fvertex] (0v1) at (0,1.5) {};
\node[vertex]  (1v1) at (1.5,1.5) {};
\draw[gray] (-0.75,-0.75) -- (2.25,-0.75);
\draw[gray] (-0.75, 2.25) -- (2.25, 2.25);
\draw[gray] (-0.75,-0.75) -- (-0.75, 2.25);
\draw[gray] ( 2.25,-0.75) -- ( 2.25, 2.25);
\draw[<-] (-1.2,0) -- (0v0);
\draw[->] (0v0) -- (1v0);
\draw[<-] (1v0) -- (2.7,0);
\draw[->] (-1.2,1.5) -- (0v1);
\draw[<-, double, very thick, red] (0v1) -- (1v1);
\draw[->] (1v1) -- (2.7,1.5);
\draw[->] (0,-1.2) -- (0v0);
\draw[<-] (0v0) -- (0v1);
\draw[->] (0v1) -- (0,2.7);
\draw[<-] (1.5,-1.2) -- (1v0);
\draw[->] (1v0) -- (1v1);
\draw[<-] (1v1) -- (1.5,2.7);
\node at (0.75,-1.6) {\small $R.L.L.L$};
\end{tikzpicture}}

\newcommand{\diagRLLLtwo}{%
\begin{tikzpicture}[
scale=0.5,
vertex/.style={circle, draw, minimum size=2.5mm, inner sep=0pt},
fvertex/.style={circle, draw, fill=black, minimum size=2.5mm, inner sep=0pt},
gvertex/.style={circle, draw, fill=gray!60, minimum size=2.5mm, inner sep=0pt},
>=Stealth, thick]
\node[vertex]  (0v0) at (0,0) {};
\node[vertex]  (1v0) at (1.5,0) {};
\node[vertex]  (0v1) at (0,1.5) {};
\node[vertex]  (1v1) at (1.5,1.5) {};
\draw[gray] (-0.75,-0.75) -- (2.25,-0.75);
\draw[gray] (-0.75, 2.25) -- (2.25, 2.25);
\draw[gray] (-0.75,-0.75) -- (-0.75, 2.25);
\draw[gray] ( 2.25,-0.75) -- ( 2.25, 2.25);
\draw[<-] (-1.2,0) -- (0v0);
\draw[->] (0v0) -- (1v0);
\draw[<-] (1v0) -- (2.7,0);
\draw[->] (-1.2,1.5) -- (0v1);
\draw[<-] (0v1) -- (1v1);
\draw[->] (1v1) -- (2.7,1.5);
\draw[->] (0,-1.2) -- (0v0);
\draw[<-, double, very thick, red] (0v0) -- (0v1);
\draw[->] (0v1) -- (0,2.7);
\draw[<-] (1.5,-1.2) -- (1v0);
\draw[->] (1v0) -- (1v1);
\draw[<-] (1v1) -- (1.5,2.7);
\node at (0.75,-1.6) {\small $R.L.L.L2$};
\end{tikzpicture}}

\newcommand{\diagRLLLthree}{%
\begin{tikzpicture}[
scale=0.5,
vertex/.style={circle, draw, minimum size=2.5mm, inner sep=0pt},
fvertex/.style={circle, draw, fill=black, minimum size=2.5mm, inner sep=0pt},
gvertex/.style={circle, draw, fill=gray!60, minimum size=2.5mm, inner sep=0pt},
>=Stealth, thick]
\node[fvertex] (0v0) at (0,0) {};
\node[vertex]  (1v0) at (1.5,0) {};
\node[vertex]  (0v1) at (0,1.5) {};
\node[vertex]  (1v1) at (1.5,1.5) {};
\draw[gray] (-0.75,-0.75) -- (2.25,-0.75);
\draw[gray] (-0.75, 2.25) -- (2.25, 2.25);
\draw[gray] (-0.75,-0.75) -- (-0.75, 2.25);
\draw[gray] ( 2.25,-0.75) -- ( 2.25, 2.25);
\draw[<-, double, very thick, red] (-1.2,0) -- (0v0);
\draw[->] (0v0) -- (1v0);
\draw[<-] (1v0) -- (2.7,0);
\draw[->] (-1.2,1.5) -- (0v1);
\draw[<-] (0v1) -- (1v1);
\draw[->] (1v1) -- (2.7,1.5);
\draw[->] (0,-1.2) -- (0v0);
\draw[<-] (0v0) -- (0v1);
\draw[->] (0v1) -- (0,2.7);
\draw[<-] (1.5,-1.2) -- (1v0);
\draw[->] (1v0) -- (1v1);
\draw[<-] (1v1) -- (1.5,2.7);
\node at (0.75,-1.6) {\small $R.L.L.L3$};
\end{tikzpicture}}

%


\setlength{\tabcolsep}{6pt}
\renewcommand{\arraystretch}{2}

\begin{tikzpicture}[remember picture]

\node (matrix) {
\begin{tabular}{ccccc}

\tikz[remember picture]\node (S) {\diagS}; &
\tikz[remember picture]\node (R) {\diagR}; &
\tikz[remember picture]\node (RR) {\diagRR}; &
\tikz[remember picture]\node (RLR) {\diagRLR}; &
\tikz[remember picture]\node (RLLR) {\diagRLLR}; \\

&
\tikz[remember picture]\node (L) {\diagL}; &
\tikz[remember picture]\node (RL) {\diagRL}; &
\tikz[remember picture]\node (RLL) {\diagRLL}; &
\tikz[remember picture]\node (RLLL) {\diagRLLL}; \\

&
&
&
\tikz[remember picture]\node (RLLLthree) {\diagRLLLthree}; &
\tikz[remember picture]\node (RLLLtwo) {\diagRLLLtwo}; \\

\end{tabular}
};


\draw[->,thick, shorten >=-2mm, shorten <=-2mm, blue]
(S.east) -- (R.west);

\draw[->,thick, shorten >=-4mm, shorten <=-4mm, blue]
(S.south east) to[bend right=20] (L.north west);

\draw[->,thick, shorten >=-2mm, shorten <=-2mm, blue]
(R.east) -- (RR.west);

\draw[->,thick, shorten >=-4mm, shorten <=-4mm, blue]
(R.south east) to[bend right=20] (RL.north west);

\draw[->,thick, shorten >=-2mm, shorten <=-2mm, blue]
(RL.east) -- (RLL.west);

\draw[->,thick, shorten >=-4mm, shorten <=-4mm, blue]
(RL.north east) to[bend right=20] (RLR.south west);

\draw[->,thick, shorten >=-2mm, shorten <=-2mm, blue]
(RLL.east) -- (RLLL.west);

\draw[->,thick, shorten >=-4mm, shorten <=-4mm, blue]
(RLL.north east) to[bend right=20] (RLLR.south west);

\draw[->,thick, shorten >=-2mm, shorten <=-2mm, blue]
(RLLL.south) -- (RLLLtwo.north);

\draw[->,thick, shorten >=-2mm, shorten <=-2mm, blue]
(RLLLtwo.west) -- (RLLLthree.east);

\end{tikzpicture}

\caption{Illustration of the branching exploration on $G_{2,2}$. Blue arrows indicate transitions in the exploration tree. Some transitions correspond to branching decisions when the ant first visits a previously unvisited cell, while others represent deterministic continuations of the exploration. The double red arc highlights the current position of the ant. Black and white vertices represent cells in states $L$ and $R$, respectively, while gray vertices correspond to cells whose initial state has not yet been determined. Labels such as R, L, R.L, R.L.L, etc., record the sequence of state assignments made during branching, starting from the root (S) of the exploration tree.}
\label{fig:branchingexplo}
\end{figure}
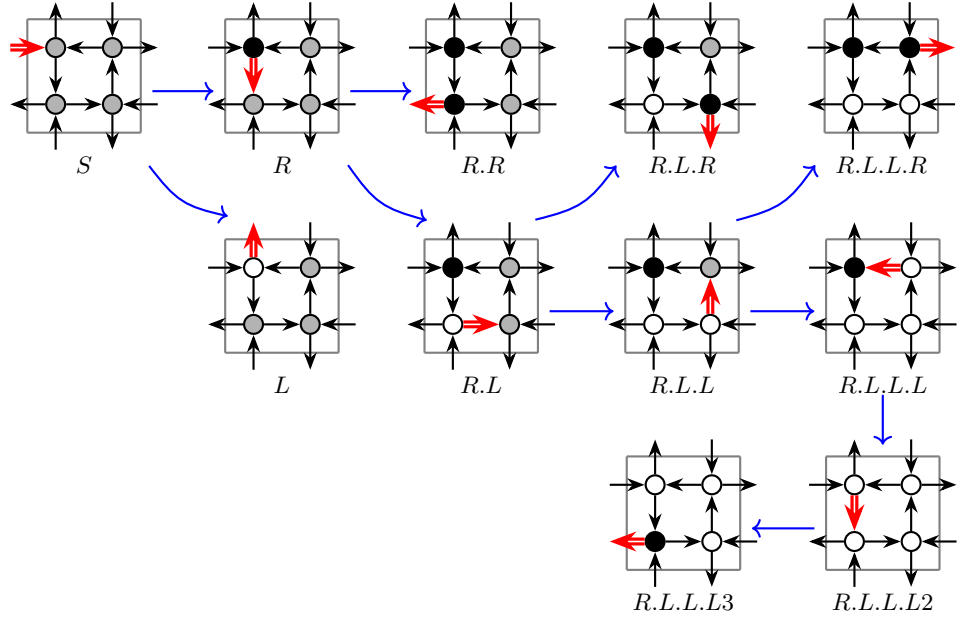

\textbf{Symmetry reduction.}  
The number of starting positions can also be reduced using geometric symmetries of the grid. Since the ant’s dynamics are invariant under reflections and rotations of the rectangular domain, many in-boundary arcs lead to equivalent trajectories.
It therefore suffices to consider only the in-boundary arcs associated with $\lceil k/2 \rceil$ consecutive cells of a boundary column (for horizontal starting positions), beginning at a corner cell, together with those associated with $\lceil n/2 \rceil$ consecutive cells of a boundary row (for vertical starting positions), also beginning at a corner cell.
When $k = n$, the additional symmetry between rows and columns makes these two families equivalent, so that it is enough to consider a single boundary side.
Moreover, instead of launching a separate branching exploration for each of these representative starting positions of a side, we aggregate them into a single exploration. For the horizontal case, this is achieved by extending $G_{k,n}$ upward with $\left\lceil \frac{k}{2} \right\rceil - 1$ additional rows, yielding a working grid of height $k + \left\lceil \frac{k}{2} \right\rceil - 1$. The exploration is then initiated from the in-boundary arc of the cell that originally lay on the top row of $G_{k,n}$.
During the exploration, each branch is followed until the vertical span of the visited cells exceeds $k$. This condition precisely captures the escape of the ant for one of the original starting positions represented in the merged construction.
A symmetric construction is used for vertical starting positions when $k \neq n$, and the overall escape time is obtained by taking the maximum over both orientations.

\begin{figure}[H]
\centering

\begin{minipage}{0.48\textwidth}
\centering

\begin{tikzpicture}[
scale=0.6,
vertex/.style={circle, draw, minimum size=2.5mm, inner sep=0pt},
fvertex/.style={circle, draw, fill=black, minimum size=2.5mm, inner sep=0pt},
gvertex/.style={circle, draw, fill=gray!60, minimum size=2.5mm, inner sep=0pt},
>=Stealth,
thick
]

\def\k{5}
\def\n{5}
\def\spacing{1.5}
\def\externalarrow{1.0}

\pgfmathsetmacro{\gridleft}{-\spacing/2}
\pgfmathsetmacro{\gridright}{(\n-1)*\spacing+\spacing/2}
\pgfmathsetmacro{\gridbottom}{-\spacing/2}
\pgfmathsetmacro{\gridtop}{(\k-1)*\spacing+\spacing/2}

\foreach \i in {0,...,\n} {
    \pgfmathsetmacro{\xline}{(\i-0.5)*\spacing}
    \draw[gray] (\xline,\gridbottom) -- (\xline,\gridtop);
}

\foreach \j in {0,...,\k} {
    \pgfmathsetmacro{\yline}{(\j-0.5)*\spacing}
    \draw[gray] (\gridleft,\yline) -- (\gridright,\yline);
}

\foreach \i in {0,...,\numexpr\n-1} {
    \foreach \j in {0,...,\numexpr\k-1} {
        \pgfmathsetmacro{\x}{\i*\spacing}
        \pgfmathsetmacro{\y}{\j*\spacing}
        \node[gvertex] (\i v\j) at (\x,\y) {};
    }
}

\foreach \j in {0,...,\numexpr\k-1} {
    \pgfmathsetmacro{\y}{\j*\spacing}
    \draw (-\externalarrow,\y) -- (0v\j);
    \foreach \i in {0,...,\numexpr\n-2} {
        \draw (\i v\j) -- (\the\numexpr\i+1\relax v\j);
    }
    \pgfmathsetmacro{\xright}{(\n-1)*\spacing+\externalarrow}
    \draw (\the\numexpr\n-1\relax v\j) -- (\xright,\y);
}

\foreach \i in {0,...,\numexpr\n-1} {
    \pgfmathsetmacro{\x}{\i*\spacing}
    \draw (\x,-\externalarrow) -- (\i v0);
    \foreach \j in {0,...,\numexpr\k-2} {
        \draw (\i v\j) -- (\i v\the\numexpr\j+1\relax);
    }
    \pgfmathsetmacro{\ytop}{(\k-1)*\spacing+\externalarrow}
    \draw (\i v\the\numexpr\k-1\relax) -- (\x,\ytop);
}

\draw[->, double, ultra thick, blue]
(-\externalarrow -0.5, 2*\spacing) -- (0v2);

\draw[->, double, ultra thick, red]
(-\externalarrow -0.5, 3*\spacing) -- (0v3);

\draw[->, double, ultra thick, green!60!black]
(-\externalarrow -0.5, 4*\spacing) -- (0v4);

\end{tikzpicture}

\end{minipage}
\hfill
\begin{minipage}{0.48\textwidth}
\centering

\begin{tikzpicture}[
scale=0.6,
vertex/.style={circle, draw, minimum size=2.5mm, inner sep=0pt},
fvertex/.style={circle, draw, fill=black, minimum size=2.5mm, inner sep=0pt},
gvertex/.style={circle, draw, fill=gray!60, minimum size=2.5mm, inner sep=0pt},
>=Stealth,
thick
]

\def\k{7}
\def\n{5}
\def\spacing{1.5}
\def\externalarrow{1.0}

\pgfmathsetmacro{\gridleft}{-\spacing/2}
\pgfmathsetmacro{\gridright}{(\n-1)*\spacing+\spacing/2}
\pgfmathsetmacro{\gridbottom}{-\spacing/2}
\pgfmathsetmacro{\gridtop}{(\k-1)*\spacing+\spacing/2}

\foreach \i in {0,...,\n} {
    \pgfmathsetmacro{\xline}{(\i-0.5)*\spacing}
    \draw[gray] (\xline,\gridbottom) -- (\xline,\gridtop);
}

\foreach \j in {0,...,\k} {
    \pgfmathsetmacro{\yline}{(\j-0.5)*\spacing}
    \draw[gray] (\gridleft,\yline) -- (\gridright,\yline);
}

\foreach \i in {0,...,\numexpr\n-1} {
    \foreach \j in {0,...,\numexpr\k-1} {
        \pgfmathsetmacro{\x}{\i*\spacing}
        \pgfmathsetmacro{\y}{\j*\spacing}
        \node[gvertex] (\i v\j) at (\x,\y) {};
    }
}

\foreach \j in {0,...,\numexpr\k-1} {
    \pgfmathsetmacro{\y}{\j*\spacing}
    \draw (-\externalarrow,\y) -- (0v\j);
    \foreach \i in {0,...,\numexpr\n-2} {
        \draw (\i v\j) -- (\the\numexpr\i+1\relax v\j);
    }
    \pgfmathsetmacro{\xright}{(\n-1)*\spacing+\externalarrow}
    \draw (\the\numexpr\n-1\relax v\j) -- (\xright,\y);
}

\foreach \i in {0,...,\numexpr\n-1} {
    \pgfmathsetmacro{\x}{\i*\spacing}
    \draw (\x,-\externalarrow) -- (\i v0);
    \foreach \j in {0,...,\numexpr\k-2} {
        \draw (\i v\j) -- (\i v\the\numexpr\j+1\relax);
    }
    \pgfmathsetmacro{\ytop}{(\k-1)*\spacing+\externalarrow}
    \draw (\i v\the\numexpr\k-1\relax) -- (\x,\ytop);
}

\draw[very thick, green!60!black]
(-0.8,-0.8) rectangle (5*\spacing -0.7,5*\spacing -0.7);

\draw[very thick, red]
(-0.9,0.6) rectangle (5*\spacing -0.6,5*\spacing +0.9);

\draw[very thick, blue]
(-1.05,2.1) rectangle (5*\spacing -0.45,5*\spacing +2.4);

\draw[->, double, ultra thick]
(-\externalarrow -0.5, 4*\spacing) -- (0v4);

\end{tikzpicture}

\end{minipage}

\caption{Illustration of the symmetry reduction for $G_{5,5}$. Left: by symmetry, it suffices to consider the three representative starting positions shown by the colored arrows. Right: these starting positions are merged into a single branching exploration by extending the grid upward. The exploration is initiated from the unique starting position indicated by the double black arrow. The colored rectangles represent the original  grid translated to match the three representative starting positions. A branch is terminated as soon as the vertical span of its visited cells exceeds, which corresponds to the ant escaping from the original grid for one of the represented starting positions.}

\label{fig:symreduction}

\end{figure}
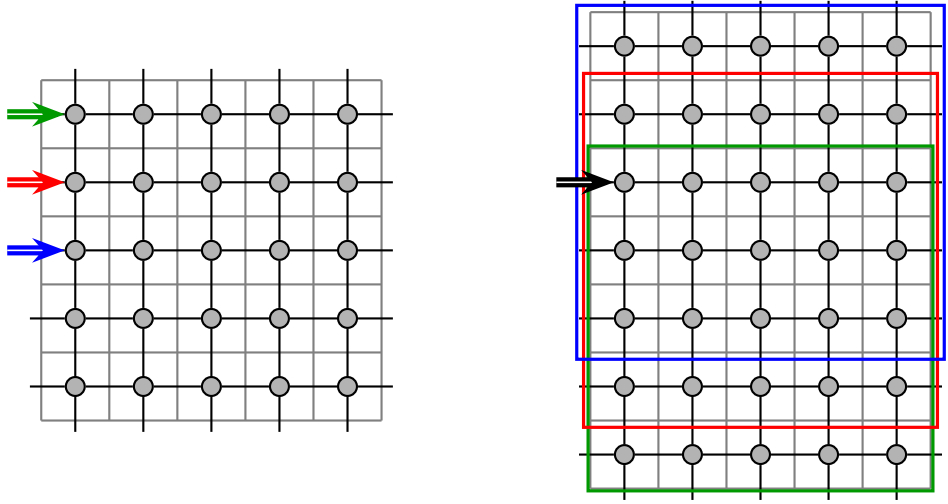
\medskip
The branching exploration is implemented via a depth-first backtracking algorithm. During the forward phase, the ant is simulated until it escapes the admissible region. Each first visit to a cell is recorded as a branching point together with the current height and vertical bounds. When a branch terminates, the algorithm retraces its steps backward, restoring the previous state and identifying the most recent cell where an unexplored alternative remains. The exploration then resumes from this point with the alternate state. This systematic traversal guarantees that every feasible branch (and thus every configuration consistent with the symmetry reduction) is explored.

\subsection{Upper Bound of the Escaping time of Square Domains}\label{sec:square}
In this section, we consider the  square finite connected domain $G_n = G_{n,n}$ with $n \in \mathbb{N}^*$ and investigate the growth of $S_n = S_{n,n}$. Note that $G_n$ has $2 n^2 \cdot 2^{n^2} $ configurations and by Corollary~\ref{pr:exit_bound}, $S_n \le 2^{n^2} $.  In this section, we derive an improved upper bound using an inductive approach, expressing $S_n$ in terms of $S_{n-2}$.

\begin{theorem}
For each positive integer $n $, we have $S_n \le (n+1)!$.
\end{theorem}
\begin{proof}
By simulation:  $S_1=1$ and $(1+1)!=2$, $S_2=6$ and $(2+1)!=6$ thus the inequality is true for $n\in \{1,2\}$. 

For $n \ge 3$, observe that $G_n$ can be decomposed as the union of $G_{n-2}$ and $B_n$, where $B_n$ denotes the connected domain on the boundary cells of $G_n$ (see Figure~\ref{fig:sqgrid_decomposition}). This decomposition provides the basis for our inductive argument.

\begin{figure}[htbp]
\centering
\begin{tikzpicture}[
scale=0.5,
vertex/.style={circle, draw, minimum size=2.5mm, inner sep=0pt},
fvertex/.style={circle, draw, fill=black, minimum size=2.5mm, inner sep=0pt}, 
gvertex/.style={circle, draw, fill=gray!60, minimum size=2.5mm, inner sep=0pt}, 
>=Stealth,
thick
]
\def\k{6} 
\def\n{6} 
\def\spacing{1.5}
\def\externalarrow{1.0}

\pgfmathsetmacro{\gridleft}{-\spacing/2}
\pgfmathsetmacro{\gridright}{(\n-1)*\spacing+\spacing/2}
\pgfmathsetmacro{\gridbottom}{-\spacing/2}
\pgfmathsetmacro{\gridtop}{(\k-1)*\spacing+\spacing/2}

\fill[gray!50] (\gridleft, \gridbottom) rectangle (\gridright, \gridtop);
\fill[white] (0.5*\spacing, 0.5*\spacing) rectangle ({(\n-2)*\spacing+0.5*\spacing}, {(\k-2)*\spacing+0.5*\spacing});

\foreach \i in {0,...,\numexpr\n-1} {
    \foreach \j in {0,...,\numexpr\k-1} {
        \pgfmathsetmacro{\x}{\i*\spacing}
        \pgfmathsetmacro{\y}{\j*\spacing}
        \node[gvertex] (\i v\j) at (\x,\y) {};
    }
}

\foreach \j in {0,...,\numexpr\k-1} {
    \pgfmathsetmacro{\y}{\j*\spacing}
    \pgfmathsetmacro{\dir}{int(mod(\j,2)==0?1:0)}
    \ifnum\dir=1
        \draw[->] (-\externalarrow,\y) -- (0v\j);
    \else
        \draw[<-] (-\externalarrow,\y) -- (0v\j);
    \fi
    \foreach \i in {0,...,\numexpr\n-2} {
        \pgfmathsetmacro{\arrowdir}{int(mod(\i+\j,2)==0?1:0)}
        \ifnum\arrowdir=1
            \draw[<-] (\i v\j) -- (\the\numexpr\i+1\relax v\j);
        \else
            \draw[->] (\i v\j) -- (\the\numexpr\i+1\relax v\j);
        \fi
    }
    \pgfmathsetmacro{\xright}{(\n-1)*\spacing+\externalarrow}
    \pgfmathsetmacro{\lastdir}{int(mod(\n-1+\j,2)==0?1:0)}
    \ifnum\lastdir=1
        \draw[<-] (\the\numexpr\n-1\relax v\j) -- (\xright,\y);
    \else
        \draw[->] (\the\numexpr\n-1\relax v\j) -- (\xright,\y);
    \fi
}

\foreach \i in {0,...,\numexpr\n-1} {
    \pgfmathsetmacro{\x}{\i*\spacing}
    \pgfmathsetmacro{\dir}{int(mod(\i,2)==0?1:0)}
    \ifnum\dir=1
        \draw[<-] (\x,-\externalarrow) -- (\i v0);
    \else
        \draw[->] (\x,-\externalarrow) -- (\i v0);
    \fi
    \foreach \j in {0,...,\numexpr\k-2} {
        \pgfmathsetmacro{\arrowdir}{int(mod(\i+\j,2)==0?1:0)}
        \ifnum\arrowdir=1
            \draw[->] (\i v\j) -- (\i v\the\numexpr\j+1\relax);
        \else
            \draw[<-] (\i v\j) -- (\i v\the\numexpr\j+1\relax);
        \fi
    }
    \pgfmathsetmacro{\ytop}{(\k-1)*\spacing+\externalarrow}
    \pgfmathsetmacro{\lastdir}{int(mod(\i+\k-1,2)==0?1:0)}
    \ifnum\lastdir=1
        \draw[->] (\i v\the\numexpr\k-1\relax) -- (\x,\ytop);
    \else
        \draw[<-] (\i v\the\numexpr\k-1\relax) -- (\x,\ytop);
    \fi
}
\end{tikzpicture}
\caption{Decomposition of $G_6$ as the union of boundary cells $B_6$ (gray shaded region)  and interior grid $G_4$ (white region)}
\label{fig:sqgrid_decomposition}
\end{figure}
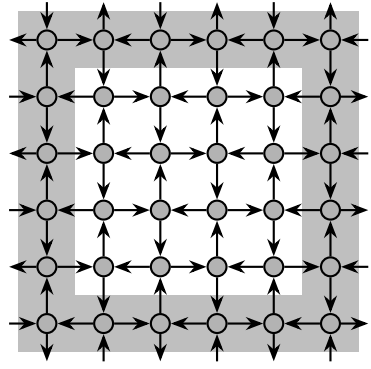

To  maximize the number of steps on $G_n$, the ant must start at an in-boundary arc of $G_n$ (Remark \ref{rm:inb_max}). From this starting position, it reaches an in-boundary arc of $G_{n-2}$ within at most two or three steps when a corner cell is involved. Then it performs some number of steps within $G_{n-2}$ before reaching an out-boundary arc of $G_{n-2}$ thereby returning to $B_n$.
Once in $B_n$, the ant undergoes a  bouncing phase before re-entering $G_{n-2}$. Each such bouncing phase requires at most two steps or three steps when a corner cell is involved. The process then repeats: the ant alternates between phases inside $G_{n-2}$ and bouncing phases within $B_n$ until it finally exits $G_{n-2}$ for the last time and reaches an out-boundary arc of $G_{n}$ in at most three  additional steps.  Suppose that the ant returns to $G_{n-2}$ at most $X_n$ times after its initial entry, then we obtain:

$$S_n \le 3 + (X_n+1) \cdot S_{n-2} + 3X_n + 3,$$
where the four terms account, respectively, for: 
the initial steps to enter $G_{n-2}$ from an in-boundary arc of $G_n$; 
the $(X_n+1)$ visits to $G_{n-2}$ (initial entry plus $X_n$ returns); 
the $X_n$ bouncing phases on $B_n$; and the final steps to escape $G_{n-2}$ to an out-boundary arc of $G_n$. 
This simplifies to:
$$S_n \le (X_n+1) \cdot S_{n-2} + 3X_n + 6.$$

It remains to bound $X_n$. 
To this end, we analyze the behavior of the ant on the top row of $B_n$ and determine the maximal number of times it can bounce there while still being able to return to $G_{n-2}$. During a bouncing phase, the ant moves from $G_{n-2}$, traverses boundary cells, and returns to $G_{n-2}$, without occupying any in-boundary or out-boundary arc of $G_n$. 

For example, Figure~\ref{fig:bouncingbound} illustrates all possible bounces in the top row of $B_{10}$ along with the maximum number of times each can be performed.

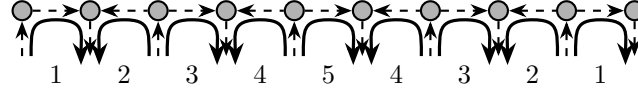
\begin{figure}[htbp]
\centering
\begin{tikzpicture}[
scale=0.6,
vertex/.style={circle, draw, minimum size=2.5mm, inner sep=0pt},
fvertex/.style={circle, draw, fill=black, minimum size=2.5mm, inner sep=0pt}, 
gvertex/.style={circle, draw, fill=gray!60, minimum size=2.5mm, inner sep=0pt}, 
>=Stealth,
thick
]
\def\n{10} 
\def\spacing{1.5}
\def\externalarrow{1.0}
\def\margin{0.5}
\pgfmathsetmacro{\gridleft}{-\spacing/2}
\pgfmathsetmacro{\gridright}{(\n-1)*\spacing+\spacing/2}
\pgfmathsetmacro{\gridbottom}{-\spacing/2}
\pgfmathsetmacro{\gridtop}{\spacing/2}

\foreach \i in {0,...,\numexpr\n-1} {
    \pgfmathsetmacro{\x}{\i*\spacing}
    \node[gvertex] (\i v0) at (\x,0) {};
}
\foreach \i in {0,...,\numexpr\n-2} {
    \pgfmathsetmacro{\arrowdir}{int(mod(\i+5,2)==0?1:0)}
    \ifnum\arrowdir=1
        \draw[dashed,<-] (\i v0) -- (\the\numexpr\i+1\relax v0);
    \else
        \draw[dashed,->] (\i v0) -- (\the\numexpr\i+1\relax v0);
    \fi
}
\foreach \i in {0,...,\numexpr\n-1} {
    \pgfmathsetmacro{\x}{\i*\spacing}
    \pgfmathsetmacro{\dir}{int(mod(\i,2)==0?1:0)}
    \ifnum\dir=1
        \draw[dashed,->] (\x,-\externalarrow) -- (\i v0);
    \else
        \draw[dashed,<-] (\x,-\externalarrow) -- (\i v0);
    \fi
}

\foreach \i in {0,...,\numexpr\n-2} {
    \pgfmathsetmacro{\xstart}{\i*\spacing}
    \pgfmathsetmacro{\xmid}{(\i+0.5)*\spacing}
    \pgfmathsetmacro{\xend}{(\i+1)*\spacing}
    \pgfmathsetmacro{\kval}{int(min(\i+1, \n-1-\i))}
    \pgfmathsetmacro{\dir}{int(mod(\i+1,2)==0?1:0)}
    \ifnum\dir=1
        \draw[very thick, <-] 
            (\xstart+0.2,-\externalarrow-0.2) to[out=90, in=180, looseness=1.5] (\xmid,-0.2) to[out=0, in=90, looseness=1.5] (\xend-0.2,-\externalarrow);
            \node at (\xmid, -\externalarrow-0.4) {$\kval$};
    \else
        \draw[very thick, ->] 
        (\xstart+0.2,-\externalarrow) to[out=90, in=180, looseness=1.5] (\xmid,-0.2) 
        to[out=0, in=90, looseness=1.5] (\xend-0.2,-\externalarrow-0.2);
        \node at (\xmid, -\externalarrow-0.4) {$\kval$};
    \fi
}

\end{tikzpicture}
\caption{Top row of $B_{10}$ with curved solid lines indicating possible bouncing motions and the corresponding maximal number of times each bounce can occur.}
\label{fig:bouncingbound}
\end{figure}

Each of the two corner cells of the top row can be used to perform only one type of bounce, which can occur at most once. Indeed, after a corner cell is used for a bounce, its color is updated, and the ant's next visit to that cell immediately leads to an out-boundary arc of $G_n$ in a single step.

In general, each non-corner cell in the top row can be used to perform two types of bouncing motions - left and right - depending on which side of the cell is located the horizontal arc that is involved in the bounce. Moreover, for each such cell, the counts of left and right bounces can differ by at most one, since consecutive visits to a given cell must alternate between the two directions due to the color-alternation property.

Combining these two observations, the number of possible bounces on the top row of $B_n$ is bounded by:
\[
\begin{cases}
1+2+...+(\frac{n}{2}-1) +\frac{n}{2}+ (\frac{n}{2} -1) +...+2+1= \frac{n^2}{4} & \text{if n is even}\\

1+2+...+(\frac{n-1}{2} -1) +\frac{n-1}{2}+\frac{n-1}{2}+ (\frac{n-1}{2} -1) +...+2+1= \frac{n^2-1}{4} & \text{if n is odd}
\end{cases}
\]
By symmetry across the four sides of $B_n$, we obtain  $X_n \le 4 \cdot \frac{n^2}{4}=n^2$. 
Substituting in $ S_n\le (X_n+1) \cdot S_{n-2} + 3X_n+6$, we obtain:
$ S_n \le \left(n^2+1\right) S_{n-2} + 3n^2 + 6.$
Let us define $$T_n = S_n + 3.$$ Then the previous inequality  can be written as
$ (S_n +3) \le \left(n^2+1\right)\left(S_{n-2}  +3 \right) + 6$ which yields
$ T_n \le \left(n^2+1\right)T_{n-2} + 6$.

We claim that $6<(n-1)T_{n-2}$ for all $n\ge 3$. Indeed, 
for $n = 3$,  we have, $S_1= 1$, and  thus $(n-1)T_{n-2} = 2\times (1+3) = 8 > 6$.
Since $(n-1)T_{n-2}$ is increasing, it follows that $6<(n-1)T_{n-2}$ for all $n\ge 3$. 
Therefore
$$ T_n \le \left(n^2+1\right)T_{n-2} + 6 \le \left(n^2+1\right) T_{n-2} + (n-1)T_{n-2}  = (n+1) \cdot n \cdot T_{n-2}$$ and by induction on the inequality $T_n\le (n+1) \cdot n \cdot T_{n-2}$, we have

\[
T_{n} \le 
\begin{cases}
(n+1) \cdot n \cdot (n-1)\cdot (n-2)\cdot  ...\cdot 6 \cdot T_4 & \text{if } n \text{ is even} \\
(n+1) \cdot n \cdot (n-1)\cdot (n-2) \cdot ... \cdot 5 \cdot T_3 & \text{if } n \text{ is odd} 
\end{cases}
\]
which in both cases yield $S_{n} \le T_{n} \le (n+1)!$ since $T_3= S_3+3=20 < 4!$ and $T_4=S_4+3 =49 < 5!$.
\end{proof}

\subsection{Linear bound of the Escaping Time of Some Rectangular Domains}

In this section, we focus on the escaping time of rectangular finite connected domains consisting of two or three rows and an arbitrary number of columns.
Our approach to determining the escape time of these domains consists in bounding the number of steps the ant can perform within each column of the domain before the termination of its motion within the entire domain. These columns —  horizontal or vertical $G_{2,1}$ or
$G_{3,1}$ — can be interpreted as gadgets consisting of entry arcs (the in-boundary arcs), exit arcs (the out-boundary arcs), and an internal state determined by the color configuration. Given its internal state and the specific entry arc through which the ant enters, the gadget updates its internal state and directs the ant toward one of its exit arcs.
There are four possible types of motion that the ant can perform on every such column:
\begin{itemize}
    \item \emph{Bouncing}: the ant enters and leaves through the same side of the column;
    \item \emph{Crossing}: the ant enters through one side of the column and exits through the opposite side;
    \item \emph{Initial}: the ant enters the grid through an in-boundary arc at the top or bottom of the column;
    \item \emph{Out}: the ant exits the grid through an out-boundary arc at the top or bottom of the column.
\end{itemize}
Initial and out motions are performed in one step,  bouncing and crossing motions in two steps.

In all figures of this section, black, white, and gray cells still represent cells whose states are $L$, $R$, or unspecified, respectively, and curved solid arrows depict the motion of the ant.

\subsubsection{Two-row Grids}\label{sec:tworow}

When viewed column by column, a two-row grid (either vertical or horizontal) can be seen as an alternating
succession of horizontal and vertical $G_{2,1}$ ( See Figure$~\ref{fig:2rowgrid}$).

\begin{figure}[H]
\centering

\begin{subfigure}{0.5\textwidth}
\centering
\begin{tikzpicture}[
scale=0.5,
vertex/.style={circle, draw, minimum size=2.5mm, inner sep=0pt},
fvertex/.style={circle, draw, fill=black, minimum size=2.5mm, inner sep=0pt},
gvertex/.style={circle, draw, fill=gray!60, minimum size=2.5mm, inner sep=0pt},
>=Stealth,
thick
]
\def\k{2} 
\def\n{7} 
\def\spacing{1.5} 
\def\externalarrow{1.2} 
\def\margin{0.5} 
\foreach \i in {0,...,\numexpr\n-1} {
 \foreach \j in {0,...,\numexpr\k-1} {
 \pgfmathsetmacro{\x}{\i*\spacing}
 \pgfmathsetmacro{\y}{\j*\spacing}
 \node[gvertex] (\i v\j) at (\x,\y) {};
 }
}
\foreach \i/\j in {} {
 \pgfmathsetmacro{\x}{\i*\spacing}
 \pgfmathsetmacro{\y}{\j*\spacing}
 \node[fvertex] (\i v\j) at (\x,\y) {};
}
\pgfmathsetmacro{\upbound}{(\k-1)*\spacing + \margin}
\pgfmathsetmacro{\bottombound}{-\margin}
\draw[gray] (-\spacing, \upbound) -- ({(\n-1)*\spacing + \spacing }, \upbound);
\draw[gray] (-\spacing, \bottombound) -- ({(\n-1)*\spacing + \spacing}, \bottombound);
\foreach \i in {0,...,\n} {
 \pgfmathsetmacro{\x}{\i*\spacing - \spacing/2}
 \draw[gray] (\x, \bottombound) -- (\x, \upbound);
}
\foreach \j in {0,...,\numexpr\k-1} {
 \pgfmathsetmacro{\y}{\j*\spacing}
 \pgfmathsetmacro{\dir}{int(mod(\j,2)==0?1:0)}
 \ifnum\dir=1
 \draw[->] (-\externalarrow,\y) -- (0v\j);
 \else
 \draw[<-] (-\externalarrow,\y) -- (0v\j);
 \fi
 \foreach \i in {0,...,\numexpr\n-2} {
 \pgfmathsetmacro{\arrowdir}{int(mod(\i+\j,2)==0?1:0)}
 \ifnum\arrowdir=1
 \draw[<-] (\i v\j) -- (\the\numexpr\i+1\relax v\j);
 \else
 \draw[->] (\i v\j) -- (\the\numexpr\i+1\relax v\j);
 \fi
 }
 \pgfmathsetmacro{\xright}{(\n-1)*\spacing + \externalarrow}
 \pgfmathsetmacro{\lastdir}{int(mod(\n-1+\j,2)==0?1:0)}
 \ifnum\lastdir=1
 \draw[<-] (\the\numexpr\n-1\relax v\j) -- (\xright,\y);
 \else
 \draw[->] (\the\numexpr\n-1\relax v\j) -- (\xright,\y);
 \fi
}
\foreach \i in {0,...,\numexpr\n-1} {
 \pgfmathsetmacro{\x}{\i*\spacing}
 \pgfmathsetmacro{\dir}{int(mod(\i,2)==0?1:0)}
 \ifnum\dir=1
 \draw[<-] (\x,-\externalarrow) -- (\i v0) node[at start, below] {H};
 \else
 \draw[->] (\x,-\externalarrow) -- (\i v0) node[at start, below] {V};
 \fi
 \foreach \j in {0,...,\numexpr\k-2} {
 \pgfmathsetmacro{\arrowdir}{int(mod(\i+\j,2)==0?1:0)}
 \ifnum\arrowdir=1
 \draw[->] (\i v\j) -- (\i v\the\numexpr\j+1\relax);
 \else
 \draw[<-] (\i v\j) -- (\i v\the\numexpr\j+1\relax);
 \fi
 }
 \pgfmathsetmacro{\yup}{(\k-1)*\spacing + \externalarrow}
 \pgfmathsetmacro{\lastdir}{int(mod(\i+\k-1,2)==0?1:0)}
 \ifnum\lastdir=1
 \draw[->] (\i v\the\numexpr\k-1\relax) -- (\x,\yup);
 \else
 \draw[<-] (\i v\the\numexpr\k-1\relax) -- (\x,\yup);
 \fi
}
\end{tikzpicture}
\caption{Portion of two-row grid}
\label{fig:grid27c}
\end{subfigure}
\hfill
\begin{subfigure}{0.25\textwidth}
\centering
\begin{tikzpicture}[
scale=0.5,
vertex/.style={circle, draw, minimum size=2.5mm, inner sep=0pt},
fvertex/.style={circle, draw, fill=black, minimum size=2.5mm, inner sep=0pt},
gvertex/.style={circle, draw, fill=gray!60, minimum size=2.5mm, inner sep=0pt},
>=Stealth,
thick
]
\def\spacing{1.5}
\def\externalarrow{1.2}
\def\margin{0.5}
\node[gvertex] (0v0) at (0,0) {};
\node[gvertex] (0v1) at (0,1.5) {};
\draw[gray] (-\spacing/2, -\margin) -- (\spacing/2, -\margin);
\draw[gray] (-\spacing/2, 2) -- (\spacing/2, 2);
\draw[gray] (-\spacing/2, -\margin) -- (-\spacing/2, 2);
\draw[gray] (\spacing/2, -\margin) -- (\spacing/2, 2);
\draw[->] (-\externalarrow, 0) -- (0v0);
\draw[<-] (0v0) -- (\externalarrow, 0);
\draw[<-] (-\externalarrow, 1.5) -- (0v1);
\draw[->] (0v1) -- (\externalarrow, 1.5);
\draw[<-] (0, -\externalarrow) -- (0v0)node[at start, below] {H};
\draw[->] (0v0) -- (0v1);
\draw[<-] (0v1) -- (0, 2.7);
\end{tikzpicture}
\caption{Horizontal $G_{2,1}$}
\label{fig:grid21h}
\end{subfigure}
\hfill
\begin{subfigure}{0.2\textwidth}
\centering
\begin{tikzpicture}[
scale=0.5,
vertex/.style={circle, draw, minimum size=2.5mm, inner sep=0pt},
fvertex/.style={circle, draw, fill=black, minimum size=2.5mm, inner sep=0pt},
gvertex/.style={circle, draw, fill=gray!60, minimum size=2.5mm, inner sep=0pt},
>=Stealth,
thick
]
\def\spacing{1.5}
\def\externalarrow{1.2}
\def\margin{0.5}
\node[gvertex] (0v0) at (0,0) {};
\node[gvertex] (0v1) at (0,1.5) {};
\draw[gray] (-\spacing/2, -\margin) -- (\spacing/2, -\margin);
\draw[gray] (-\spacing/2, 2) -- (\spacing/2, 2);
\draw[gray] (-\spacing/2, -\margin) -- (-\spacing/2, 2);
\draw[gray] (\spacing/2, -\margin) -- (\spacing/2, 2);
\draw[<-] (-\externalarrow, 0) -- (0v0);
\draw[->] (0v0) -- (\externalarrow, 0);
\draw[->] (-\externalarrow, 1.5) -- (0v1);
\draw[<-] (0v1) -- (\externalarrow, 1.5);
\draw[->] (0, -\externalarrow) -- (0v0)node[at start, below] {V};
\draw[<-] (0v0) -- (0v1);
\draw[->] (0v1) -- (0, 2.7);
\end{tikzpicture}
\caption{Vertical $G_{2,1}$}
\label{fig:grid21v}
\end{subfigure}

\caption{Two-row grid viewed by columns}
\label{fig:2rowgrid}
\end{figure}
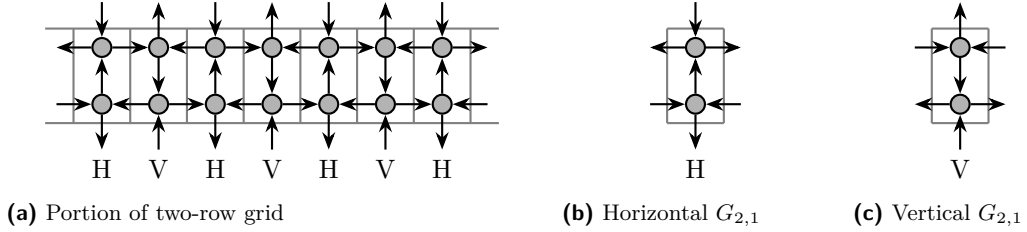

Since horizontal and vertical $G_{2,1}$ are symmetric by a half turn rotation, it suffices to restrict our analysis to the horizontal $G_{2,1}$. 
We therefore examine the possible behaviors arising when the ant enters
and exits a horizontal $G_{2,1}$ during its motion within the whole two-row
grid.

Let's label the three in-boundary arcs by $l_{in},r_{in}$ and $t_{in}$
and the three out-boundary arcs by $l_{out},r_{out}$ and $b_{out}$. The
letters $l,r,t$ and $b$ respectively denote the left, right, top,
and bottom sides of $G_{2,1}$, according to the side of the gadget
on which the corresponding arc is located, see Figure$~\ref{fig:grid21_io}$
for an illustration.

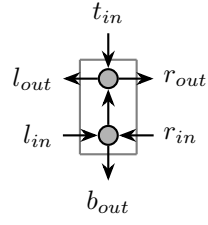
\begin{figure}[htbp]
\centering
\begin{tikzpicture}[
scale=0.5,
vertex/.style={circle, draw, minimum size=2.5mm, inner sep=0pt},
fvertex/.style={circle, draw, fill=black, minimum size=2.5mm, inner sep=0pt},
gvertex/.style={circle, draw, fill=gray!60, minimum size=2.5mm, inner sep=0pt},
>=Stealth,
thick
]
\def\spacing{1.5}
\def\externalarrow{1.2}
\def\margin{0.5}
\node[gvertex] (0v0) at (0,0) {};
\node[gvertex] (0v1) at (0,1.5) {};
\draw[gray] (-\spacing/2, -\margin) -- (\spacing/2, -\margin);
\draw[gray] (-\spacing/2, 2) -- (\spacing/2, 2);
\draw[gray] (-\spacing/2, -\margin) -- (-\spacing/2, 2);
\draw[gray] (\spacing/2, -\margin) -- (\spacing/2, 2);
\draw[->] (-\externalarrow, 0) -- (0v0) node[at start, left] {$l_{in}$};
\draw[<-] (0v0) -- (\externalarrow, 0) node[at end, right] {$r_{in}$};
\draw[<-] (-\externalarrow, 1.5) -- (0v1) node[at start, left] {$l_{out}$};
\draw[->] (0v1) -- (\externalarrow, 1.5) node[at end, right] {$r_{out}$};
\draw[<-] (0, -\externalarrow) -- (0v0) node[at start, below] {$b_{out}$};
\draw[->] (0v0) -- (0v1);
\draw[<-] (0v1) -- (0, 2.7) node[at end, above] {$t_{in}$};
\end{tikzpicture}
\caption{Horizontal $G_{2,1}$ with labeled in-boundary and out-boundary arcs}
\label{fig:grid21_io}
\end{figure}

When the ant enters through $l_{in}$ or $r_{in}$, it can exit through
any of the three exit arcs. In contrast, when the ant enters through
$t_{in}$, it can exit only through $l_{out}$ or $r_{out}$. We associate with each admissible entry--exit pair a label to describe the motion of the ant, as summarized in Table$~\ref{tab:motiontable2}$ where the letters $I,B,O$ and $C$ respectively denote Initial, Bouncing, Out and Crossing motions.

\begin{table}[H]
\centering
\begin{tabular}{|c||c|c|c|c|c|c|c|c|}
\hline
 Entry & $t_{in}$ & $t_{in}$ & $l_{in}$ & $l_{in}$ & $l_{in}$ & $r_{in}$ & $r_{in}$  & $r_{in}$  \\
\hline
 Exit & $l_{out}$ & $r_{out}$ & $l_{out}$ & $r_{out}$ & $b_{out}$ & $l_{out}$ & $r_{out}$& $b_{out}$ \\
 \hline
 Motion & $R$ & $L$ & $LL$   & $LR$ & $R$ & $RR$   & $RL$ & $L$  \\
\hline
 Label & $I_l$ & $I_r$ & $B_l$   & $C_l$ & $O_l$ & $B_r$   & $C_r$ & $O_r$  \\
\hline
\end{tabular}
\caption{Motions and labels associated to admissible entry-exit pairs. Initial and out motions are performed in one step,  bouncing and crossing motions in two steps.}
\label{tab:motiontable2}
\end{table}

Note that, during its motion within any two-row grid, the ant can use $t_{in}$ as an entry into a horizontal $G_{2,1}$ only once, namely when it enters the whole two-row grid. Moreover, if the ant exits through $b_{out}$, it escapes the two-row grid entirely, thereby terminating its motion. 
Also note that if the ant exits a horizontal $G_{2,1}$ through $l_{out}$ (respectively $r_{out}$), it can subsequently re-enter the column only through $l_{in}$ (respectively $r_{in}$).

The following result analyzes, for each type of motion (Initial, Bouncing, Out or Crossing) performed by the ant upon its first visit to a horizontal $G_{2,1}$, the maximum number of steps it can perform within this horizontal $G_{2,1}$  during its motion within the entire two-row grid. This is done by counting the steps of the initial motion together with those of any subsequent motions upon return. The next theorem is also true for vertical $G_{2,1}$.

\begin{lemma}\label{pr:2Rbound}
During its motion within a two-row grid, if the motion performed by the ant upon its first visit to a column that is a horizontal $G_{2,1}$ is: 
\begin{enumerate}
    \item an out motion, then the ant performs at most one step within this column before escaping the two-row grid; 
    \item a bouncing motion, then the ant performs at most three steps within this column before escaping the two-row grid. Moreover, these three steps consist of the bouncing motion and an out motion;
    \item an initial motion, then the ant performs at most four steps within this column before escaping the two-row grid. Moreover, these four steps consist of the initial motion, a bouncing motion and an out motion;
    \item a crossing motion, then the ant performs at most six steps within this column before escaping the two-row grid. Moreover, these six steps consist of three crossing motions, two of which occur in the same direction;
\end{enumerate}
\end{lemma}

\begin{proof}
In all figures cited in this proof, black cells are in the $L$ state and white cells are in the $R$ state. Gray cells represent cells whose state may be either $L$ or $R$ and is not yet fixed. Curved solid arrows depict the motion of the ant within the column.

Let $OC$ denote the horizontal $G_{2,1}$ under consideration ($OC$ for Original Column). We denote by ${C_{1} \brack C_{2} }$ the color configuration of $OC$, where $C_{1}, C_{2} \in \{ L, R \}$ correspond respectively to the states of the top and bottom cells.
If the motion performed by the ant upon its first visit to $OC$ is: 

\medskip

 \textbf{1. an out motion ( $O_{l}$ or $O_{r}$): one step.} The ant performs $O_{l}$ or $O_{r}$ in one step escaping the entire two-row grid.
\medskip

 \textbf{2. a bouncing motion ( $B_{l}$ or $B_{r}$): three steps at most.} For the ant to perform $B_{l}$ for example, the color configuration must necessarily be ${L \brack L}$. It becomes ${R \brack R}$ after $B_l$ and upon a subsequent return, the ant then performs $O_{l}$ reaching $b_{out}$ in one step. See Figure \ref{fig:grid21_bo}. The symmetric behavior occurs for $B_{r}$.

\begin{figure}[H]
\centering
\begin{tikzpicture}[
scale=0.6,
vertex/.style={circle, draw, minimum size=2.5mm, inner sep=0pt},
fvertex/.style={circle, draw, fill=black, minimum size=2.5mm, inner sep=0pt},
>=Stealth,
thick
]
\def\spacing{1.5}
\def\externalarrow{1.2}
\def\margin{0.5}

\begin{scope}
\node[fvertex] (0v0) at (0,0) {};
\node[fvertex] (0v1) at (0,1.5) {};
\draw[gray] (-\spacing/2, -\margin) -- (\spacing/2, -\margin);
\draw[gray] (-\spacing/2, 2) -- (\spacing/2, 2);
\draw[gray] (-\spacing/2, -\margin) -- (-\spacing/2, 2);
\draw[gray] (\spacing/2, -\margin) -- (\spacing/2, 2);
\draw[->, dashed] (-\externalarrow, 0) -- (0v0) ;
\draw[<-, dashed] (0v0) -- (\externalarrow, 0);
\draw[<-, dashed] (-\externalarrow, 1.5) -- (0v1) ;
\draw[->, dashed] (0v1) -- (\externalarrow, 1.5) ;
\draw[<-, dashed] (0, -\externalarrow) -- (0v0);
\draw[->, dashed] (0v0) -- (0v1);
\draw[<-, dashed] (0v1) -- (0, 2.7) ;
\draw[ultra thick, rounded corners=8pt, ->] 
    (-\externalarrow -0.2, 0.2) -- (-0.2, +0.2) -- (-0.2, 1.5 - 0.2) -- (-\externalarrow-0.2, 1.5 - 0.2);

\end{scope}

\draw[->, ultra thick, dotted] (2, 0.75) -- node[above] {$~$} (3.5, 0.75);

\begin{scope}[xshift=5.5cm]
\node[vertex] (0v0) at (0,0) {};
\node[vertex] (0v1) at (0,1.5) {};
\draw[gray] (-\spacing/2, -\margin) -- (\spacing/2, -\margin);
\draw[gray] (-\spacing/2, 2) -- (\spacing/2, 2);
\draw[gray] (-\spacing/2, -\margin) -- (-\spacing/2, 2);
\draw[gray] (\spacing/2, -\margin) -- (\spacing/2, 2);
\draw[->, dashed] (-\externalarrow, 0) -- (0v0) ;
\draw[<-, dashed] (0v0) -- (\externalarrow, 0) ;
\draw[<-, dashed] (-\externalarrow, 1.5) -- (0v1) ;
\draw[->, dashed] (0v1) -- (\externalarrow, 1.5) ;
\draw[<-, dashed] (0, -\externalarrow) -- (0v0);
\draw[->, dashed] (0v0) -- (0v1);
\draw[<-, dashed] (0v1) -- (0, 2.7)  ;
\draw[ultra thick, rounded corners=8pt, ->] 
    (-\externalarrow-0.2, -0.2) -- (-0.2, -0.2)  -- (-0.2, -\externalarrow-0.2);

\end{scope}

\end{tikzpicture}
\caption{$B_l$ in two steps then $O_l$ in one step.}
\label{fig:grid21_bo}
\end{figure}
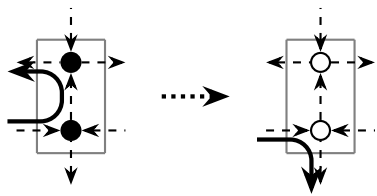

\textbf{3. an initial motion ($I_{l}$ or $I_{r}$): four steps at most.} For the ant to perform $I_{l}$ for example, the color configuration is
necessarily either \textbf{(A)} ${R \brack L}$ or \textbf{(B)} ${R \brack R}$. After the execution of $I_{l}$, the ant may subsequently re-enter the column only from the left.

If the initial configuration is ${R \brack L}$, then after performing
$I_{l}$ in one step it becomes ${L \brack L}$. From this it can perform additional three steps at most, since the ant can only eventually return from left to perform $B_l$ and then $O_l$. See \textbf{A} in Figure~\ref{fig:I}. 

If the initial configuration is ${R \brack R}$ , then after performing
$I_{l}$ in one step it becomes ${L \brack R}$. Upon the next return from the
left, the ant then performs $O_{l}$ reaching $b_{out}$ in one step. See \textbf{B} in Figure~\ref{fig:I}.

The symmetric behavior occurs for $I_{r}$.

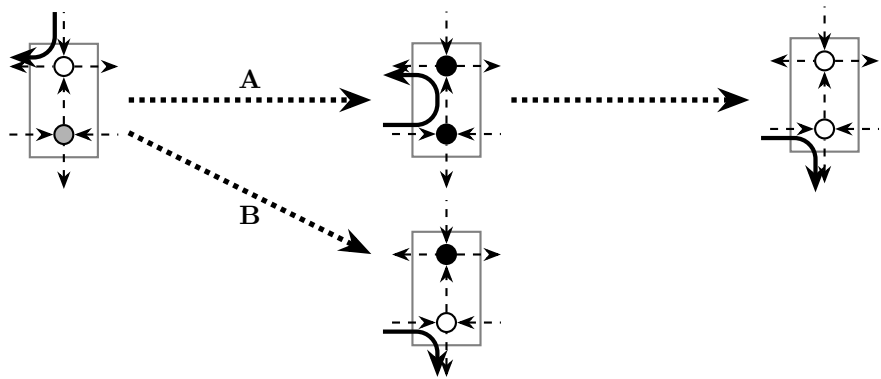
\begin{figure}[H]
\centering
\begin{tikzpicture}[
baseline, >=Stealth, thick,
]

\node (start) at (0,0) {
\begin{tikzpicture}[scale=0.6, baseline,
vertex/.style={circle, draw, minimum size=2.5mm, inner sep=0pt},
fvertex/.style={circle, draw, fill=black, minimum size=2.5mm, inner sep=0pt},
gvertex/.style={circle, draw, fill=gray!60, minimum size=2.5mm, inner sep=0pt},
>=Stealth, thick]

\def\spacing{1.5}
\def\externalarrow{1.2}
\def\margin{0.5}

\node[gvertex] (0v0) at (0,0) {};
\node[vertex] (0v1) at (0,1.5) {};

\draw[gray] (-\spacing/2, -\margin) rectangle (\spacing/2, 2);

\draw[->, dashed] (-\externalarrow, 0) -- (0v0);
\draw[<-, dashed] (0v0) -- (\externalarrow, 0);
\draw[<-, dashed] (-\externalarrow, 1.5) -- (0v1);
\draw[->, dashed] (0v1) -- (\externalarrow, 1.5);

\draw[<-, dashed] (0, -\externalarrow) -- (0v0);
\draw[->, dashed] (0v0) -- (0v1);
\draw[<-, dashed] (0v1) -- (0, 2.7);

\draw[ultra thick, rounded corners=8pt, ->]
(-0.2, 2.7) -- (-0.2, 1.7) -- (-\externalarrow, 1.7);

\end{tikzpicture}
};

\node (A) at (5,0) {
\begin{tikzpicture}[scale=0.6, baseline,
vertex/.style={circle, draw, minimum size=2.5mm, inner sep=0pt},
fvertex/.style={circle, draw, fill=black, minimum size=2.5mm, inner sep=0pt},
>=Stealth, thick]

\def\spacing{1.5}
\def\externalarrow{1.2}
\def\margin{0.5}

\node[fvertex] (0v0) at (0,0) {};
\node[fvertex] (0v1) at (0,1.5) {};

\draw[gray] (-\spacing/2, -\margin) rectangle (\spacing/2, 2);

\draw[->, dashed] (-\externalarrow, 0) -- (0v0);
\draw[<-, dashed] (0v0) -- (\externalarrow, 0);
\draw[<-, dashed] (-\externalarrow, 1.5) -- (0v1);
\draw[->, dashed] (0v1) -- (\externalarrow, 1.5);

\draw[<-, dashed] (0, -\externalarrow) -- (0v0) ;
\draw[->, dashed] (0v0) -- (0v1);
\draw[<-, dashed] (0v1) -- (0, 2.7);

\draw[ultra thick, rounded corners=8pt, ->]
(-\externalarrow -0.2, 0.2) -- (-0.2, 0.2)
-- (-0.2, 1.3) -- (-\externalarrow-0.2, 1.3);

\end{tikzpicture}
};

\node (A2) at (10,0) {
\begin{tikzpicture}[scale=0.6, baseline,
vertex/.style={circle, draw, minimum size=2.5mm, inner sep=0pt},
>=Stealth, thick]

\def\spacing{1.5}
\def\externalarrow{1.2}
\def\margin{0.5}

\node[vertex] (0v0) at (0,0) {};
\node[vertex] (0v1) at (0,1.5) {};

\draw[gray] (-\spacing/2, -\margin) rectangle (\spacing/2, 2);

\draw[->, dashed] (-\externalarrow, 0) -- (0v0);
\draw[<-, dashed] (0v0) -- (\externalarrow, 0);
\draw[<-, dashed] (-\externalarrow, 1.5) -- (0v1);
\draw[->, dashed] (0v1) -- (\externalarrow, 1.5);

\draw[<-, dashed] (0, -\externalarrow) -- (0v0);
\draw[->, dashed] (0v0) -- (0v1);
\draw[<-, dashed] (0v1) -- (0, 2.7);

\draw[ultra thick, rounded corners=8pt, ->]
(-\externalarrow-0.2, -0.2) -- (-0.2, -0.2)
-- (-0.2, -\externalarrow-0.2);

\end{tikzpicture}
};

\node (B) at (5,-2.5) {
\begin{tikzpicture}[scale=0.6, baseline,
vertex/.style={circle, draw, minimum size=2.5mm, inner sep=0pt},
fvertex/.style={circle, draw, fill=black, minimum size=2.5mm, inner sep=0pt},
>=Stealth, thick]

\def\spacing{1.5}
\def\externalarrow{1.2}
\def\margin{0.5}

\node[vertex] (0v0) at (0,0) {};
\node[fvertex] (0v1) at (0,1.5) {};

\draw[gray] (-\spacing/2, -\margin) rectangle (\spacing/2, 2);

\draw[->, dashed] (-\externalarrow, 0) -- (0v0);
\draw[<-, dashed] (0v0) -- (\externalarrow, 0);
\draw[<-, dashed] (-\externalarrow, 1.5) -- (0v1);
\draw[->, dashed] (0v1) -- (\externalarrow, 1.5);

\draw[<-, dashed] (0, -\externalarrow) -- (0v0) ;
\draw[->, dashed] (0v0) -- (0v1);
\draw[<-, dashed] (0v1) -- (0, 2.7);

\draw[ultra thick, rounded corners=8pt, ->]
(-\externalarrow -0.2, -0.2) -- (-0.2, -0.2)
-- (-0.2, -\externalarrow);

\end{tikzpicture}
};

\draw[line width=2pt, ->, dotted] (start) -- (A)
node[midway, above] {\textbf{A}};

\draw[line width=2pt, ->, dotted] (A) -- (A2);

\draw[line width=2pt, ->, dotted] (start) -- (B)
node[midway, ,left, below] {\textbf{B}};


\end{tikzpicture}

\caption{When $I_l$ is the first motion performed}
\label{fig:I}
\end{figure}

\medskip

\textbf{4. a crossing motion ($C_{l}$ or $C_{r}$): six steps at most.}  For the motion $C_{l}$ to occur for example, the color configuration must be
${R \brack L}$. In this case, the ant traverses $OC$ from left to right in two steps, thereby updating the configuration to ${L \brack R}$.
Upon a subsequent return from the right, the ant traverses $OC$ from right to left in two steps, restoring the configuration to ${R \brack L}$.
In principle, this alternating behavior could continue indefinitely, provided that the ant repeatedly crosses the column. See Figure~\ref{fig:2Cinfty}.

\begin{figure}[H]
\centering
\begin{tikzpicture}[
scale=0.6,
vertex/.style={circle, draw, minimum size=2.5mm, inner sep=0pt},
fvertex/.style={circle, draw, fill=black, minimum size=2.5mm, inner sep=0pt},
gvertex/.style={circle, draw, fill=gray!60, minimum size=2.5mm, inner sep=0pt},
>=Stealth,
thick
]
\def\spacing{1.5}
\def\externalarrow{1.2}
\def\margin{0.5}

\begin{scope}
\node[fvertex] (0v0) at (0,0) {};
\node[vertex] (0v1) at (0,1.5) {};
\draw[gray] (-\spacing/2, -\margin) -- (\spacing/2, -\margin);
\draw[gray] (-\spacing/2, 2) -- (\spacing/2, 2);
\draw[gray] (-\spacing/2, -\margin) -- (-\spacing/2, 2);
\draw[gray] (\spacing/2, -\margin) -- (\spacing/2, 2);
\draw[->, dashed] (-\externalarrow, 0) -- (0v0) ;
\draw[<-, dashed] (0v0) -- (\externalarrow, 0);
\draw[<-, dashed] (-\externalarrow, 1.5) -- (0v1) ;
\draw[->, dashed] (0v1) -- (\externalarrow, 1.5) ;
\draw[<-, dashed] (0, -\externalarrow) -- (0v0);
\draw[->, dashed] (0v0) -- (0v1);
\draw[<-, dashed] (0v1) -- (0, 2.7) ;
\draw[ultra thick, rounded corners=8pt, ->] 
    (-\externalarrow -0.2, 0.2) -- (-0.2, +0.2) -- (-0.2, 1.5 - 0.2) -- (\externalarrow, 1.5 - 0.2);
\end{scope}

\draw[<->, ultra thick, dotted] (2, 0.75) -- node[above] {$~$} (3.5, 0.75);

\begin{scope}[xshift=5.5cm]
\node[vertex] (0v0) at (0,0) {};
\node[fvertex] (0v1) at (0,1.5) {};
\draw[gray] (-\spacing/2, -\margin) -- (\spacing/2, -\margin);
\draw[gray] (-\spacing/2, 2) -- (\spacing/2, 2);
\draw[gray] (-\spacing/2, -\margin) -- (-\spacing/2, 2);
\draw[gray] (\spacing/2, -\margin) -- (\spacing/2, 2);
\draw[->, dashed] (-\externalarrow, 0) -- (0v0) ;
\draw[<-, dashed] (0v0) -- (\externalarrow, 0) ;
\draw[<-, dashed] (-\externalarrow, 1.5) -- (0v1) ;
\draw[->, dashed] (0v1) -- (\externalarrow, 1.5) ;
\draw[<-, dashed] (0, -\externalarrow) -- (0v0);
\draw[->, dashed] (0v0) -- (0v1);
\draw[<-, dashed] (0v1) -- (0, 2.7)  ;
\draw[ultra thick, rounded corners=8pt, ->] 
    (\externalarrow +0.2, 0.2) -- (+0.2, +0.2) -- (+0.2, 1.5 - 0.2) -- (-\externalarrow, 1.5 - 0.2);
\end{scope}

\end{tikzpicture}
\caption{$C_l \leftrightarrow C_r$}
\label{fig:2Cinfty}
\end{figure}
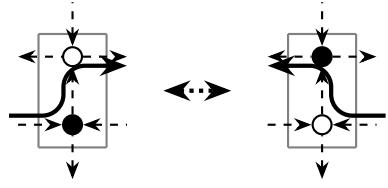

However, such indefinite repetition cannot occur in a two-row grid. After traversing $OC$ from left to right ($C_l$), the ant must reach a column $RC$ (as Right Column) to the right of $OC$ in order to return. To do so, it traverses all intermediate columns and performs a bouncing motion within $RC$. This sends the ant back from right to left across the same columns until it returns to $OC$, which it now traverses performing  $C_r$. Upon returning to $OC$ again, the ant performs a second $C_l$ and moves once more toward $RC$. This time, however, the visit to $RC$ results in an outgoing motion, as it already executed a bouncing motion there previously. Thereby the ant terminates its motion within the entire two-row grid. See Figure~\ref{fig:2C} for an illustration of all this.

The symmetric behavior occurs when the first crossing motion within $OC$ is $C_{r}$.

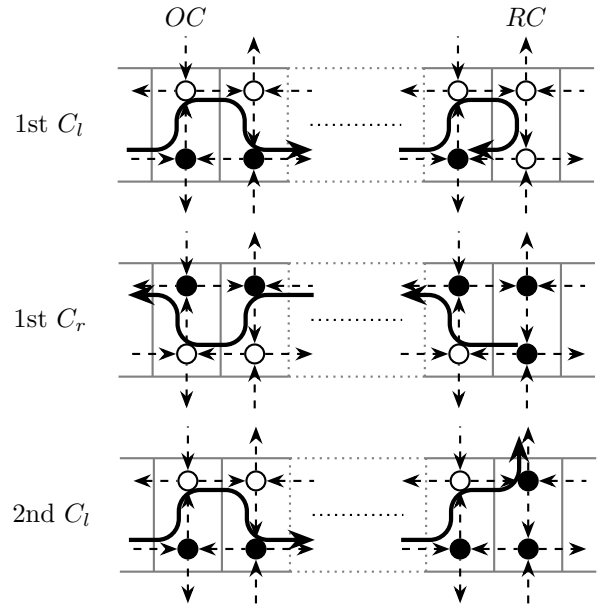
\begin{figure}[H]
\centering

\begin{tikzpicture}[
scale=0.6,
vertex/.style={circle, draw, minimum size=2.5mm, inner sep=0pt},
fvertex/.style={circle, draw, fill=black, minimum size=2.5mm, inner sep=0pt},
gvertex/.style={circle, draw, fill=gray!60, minimum size=2.5mm, inner sep=0pt},
>=Stealth,
thick
]

\node(t1) at (-3,0.75){1st $C_l$};
\node[fvertex] (0v0) at (0,0) {};
\node[fvertex] (1v0) at (1.5,0) {};
\node (2v0) at (3,0) {};
\node (3v0) at (4.5,0) {};
\node[fvertex] (4v0) at (6,0) {};
\node[vertex] (5v0) at (7.5,0) {};
\node(6v0) at (9,0) {};

\node[vertex] (0v1) at (0,1.5) {};
\node[vertex] (1v1) at (1.5,1.5) {};
\node (2v1) at (3,1.5) {};
\node (3v1) at (4.5,1.5) {};
\node[vertex] (4v1) at (6,1.5) {};
\node[vertex] (5v1) at (7.5,1.5) {};
\node(6v1) at (9,1.5) {};


\draw[gray] (-1.5, -0.5) -- (2.25, -0.5);
\draw[gray] (-1.5, 2) -- (2.25, 2);

\draw[gray, dotted](2.25, -0.5)--(5.25, -0.5);
\draw[dotted](2.75, 0.75)--(4.75, 0.75);
\draw[gray, dotted](2.25, 2)--(5.25, 2);

\draw[gray] (5.25, -0.5) -- (9, -0.5);
\draw[gray] (5.25, 2) -- (9, 2);

\draw[gray] (-0.75, -0.5) -- (-0.75, 2);
\draw[gray] (0.75, -0.5) -- (0.75, 2);
\draw[gray, dotted] (2.25, -0.5) -- (2.25, 2);
\draw[gray, dotted] (5.25, -0.5) -- (5.25, 2);
\draw[gray] (6.75, -0.5) -- (6.75, 2);
\draw[gray] (8.25, -0.5) -- (8.25, 2);


\draw[->,dashed] (-1.2,0) -- (0v0);
\draw[<-,dashed] (0v0) -- (1v0);
\draw[->,dashed] (1v0) -- (2v0);
\draw[->,dashed] (3v0) -- (4v0);
\draw[<-,dashed] (4v0) -- (5v0);
\draw[->,dashed] (5v0) -- (6v0);

\draw[<-,dashed] (-1.2,1.5) -- (0v1);
\draw[->,dashed] (0v1) -- (1v1);
\draw[<-,dashed] (1v1) -- (2v1);
\draw[<-,dashed] (3v1) -- (4v1);
\draw[->,dashed] (4v1) -- (5v1);
\draw[<-,dashed] (5v1) -- (6v1);


\draw[<-,dashed] (0,-1.2) -- (0v0);
\draw[->,dashed] (0v0) -- (0v1);
\draw[<-,dashed] (0v1) -- (0,2.7) node[at end, above] {$OC$};

\draw[->,dashed] (1.5,-1.2) -- (1v0);
\draw[<-,dashed] (1v0) -- (1v1);
\draw[->,dashed] (1v1) -- (1.5,2.7);

\draw[<-,dashed] (6,-1.2) -- (4v0);
\draw[->,dashed] (4v0) -- (4v1);
\draw[<-,dashed] (4v1) -- (6,2.7);

\draw[->,dashed] (7.5,-1.2) -- (5v0);
\draw[<-,dashed] (5v0) -- (5v1);
\draw[->,dashed] (5v1) -- (7.5,2.7) node[at end, above] {$RC$};

\draw[ultra thick, rounded corners=8pt, ->] 
    (-1.3, 0.2) -- (-0.2, 0.2) -- (-0.2, 1.3) -- 
    (1.3, 1.3) -- (1.3,0.2) -- (2.8,0.2) ;
    
\draw[ultra thick, rounded corners=8pt, ->] 
    (-1.3 + 6, 0.2) -- (-0.2 +6, 0.2) -- (-0.2+6, 1.3) -- 
    (1.3+6, 1.3) -- (1.3+6,0.2) -- ( 6+0.2, 0.2) ;    

\end{tikzpicture}

\vspace{0.5em}

\begin{tikzpicture}[
scale=0.6,
vertex/.style={circle, draw, minimum size=2.5mm, inner sep=0pt},
fvertex/.style={circle, draw, fill=black, minimum size=2.5mm, inner sep=0pt},
gvertex/.style={circle, draw, fill=gray!60, minimum size=2.5mm, inner sep=0pt},
>=Stealth,
thick
]

\def\spacing{1.5}
\def\externalarrow{1.2}
\def\margin{0.5}


\node(t2) at (-3,0.75){1st $C_r$};
\node[vertex] (0v0) at (0,0) {};
\node[vertex] (1v0) at (1.5,0) {};
\node (2v0) at (3,0) {};
\node (3v0) at (4.5,0) {};
\node[vertex] (4v0) at (6,0) {};
\node[fvertex] (5v0) at (7.5,0) {};
\node(6v0) at (9,0) {};

\node[fvertex] (0v1) at (0,1.5) {};
\node[fvertex] (1v1) at (1.5,1.5) {};
\node (2v1) at (3,1.5) {};
\node (3v1) at (4.5,1.5) {};
\node[fvertex] (4v1) at (6,1.5) {};
\node[fvertex] (5v1) at (7.5,1.5) {};
\node(6v1) at (9,1.5) {};


\draw[gray] (-1.5, -0.5) -- (2.25, -0.5);
\draw[gray] (-1.5, 2) -- (2.25, 2);

\draw[gray, dotted](2.25, -0.5)--(5.25, -0.5);
\draw[dotted](2.75, 0.75)--(4.75, 0.75);
\draw[gray, dotted](2.25, 2)--(5.25, 2);

\draw[gray] (5.25, -0.5) -- (9, -0.5);
\draw[gray] (5.25, 2) -- (9, 2);

\draw[gray] (-0.75, -0.5) -- (-0.75, 2);
\draw[gray] (0.75, -0.5) -- (0.75, 2);
\draw[gray, dotted] (2.25, -0.5) -- (2.25, 2);

\draw[gray, dotted] (5.25, -0.5) -- (5.25, 2);
\draw[gray] (6.75, -0.5) -- (6.75, 2);
\draw[gray] (8.25, -0.5) -- (8.25, 2);


\draw[->,dashed] (-1.2,0) -- (0v0);
\draw[<-,dashed] (0v0) -- (1v0);
\draw[->,dashed] (1v0) -- (2v0);
\draw[->,dashed] (3v0) -- (4v0);
\draw[<-,dashed] (4v0) -- (5v0);
\draw[->,dashed] (5v0) -- (6v0);

\draw[<-,dashed] (-1.2,1.5) -- (0v1);
\draw[->,dashed] (0v1) -- (1v1);
\draw[<-,dashed] (1v1) -- (2v1);
\draw[<-,dashed] (3v1) -- (4v1);
\draw[->,dashed] (4v1) -- (5v1);
\draw[<-,dashed] (5v1) -- (6v1);


\draw[<-,dashed] (0,-1.2) -- (0v0);
\draw[->,dashed] (0v0) -- (0v1);
\draw[<-,dashed] (0v1) -- (0,2.7);

\draw[->,dashed] (1.5,-1.2) -- (1v0);
\draw[<-,dashed] (1v0) -- (1v1);
\draw[->,dashed] (1v1) -- (1.5,2.7);

\draw[<-,dashed] (6,-1.2) -- (4v0);
\draw[->,dashed] (4v0) -- (4v1);
\draw[<-,dashed] (4v1) -- (6,2.7);

\draw[->,dashed] (7.5,-1.2) -- (5v0);
\draw[<-,dashed] (5v0) -- (5v1);
\draw[->,dashed] (5v1) -- (7.5,2.7);

\draw[ultra thick, rounded corners=8pt, <-] 
    (-1.3, 1.3) -- (-0.2, 1.3) -- (-0.2, 0.2) -- 
    (1.3, 0.2) -- (1.3,1.3) -- (2.8,1.3) ;
    
\draw[ultra thick, rounded corners=8pt, <-] 
    (-1.3 + 6, 1.3) -- (-0.2 +6, 1.3) -- (-0.2+6, 0.2) -- 
    (1.3+6, 0.2) ;    

\end{tikzpicture}

\vspace{0.5em}

\begin{tikzpicture}[
scale=0.6,
vertex/.style={circle, draw, minimum size=2.5mm, inner sep=0pt},
fvertex/.style={circle, draw, fill=black, minimum size=2.5mm, inner sep=0pt},
>=Stealth,
thick
]

\node(t3) at (-3,0.75){2nd $C_l$};
\node[fvertex] (0v0) at (0,0) {};
\node[fvertex] (1v0) at (1.5,0) {};
\node (2v0) at (3,0) {};
\node (3v0) at (4.5,0) {};
\node[fvertex] (4v0) at (6,0) {};
\node[fvertex] (5v0) at (7.5,0) {};
\node(6v0) at (9,0) {};

\node[vertex] (0v1) at (0,1.5) {};
\node[vertex] (1v1) at (1.5,1.5) {};
\node (2v1) at (3,1.5) {};
\node (3v1) at (4.5,1.5) {};
\node[vertex] (4v1) at (6,1.5) {};
\node[fvertex] (5v1) at (7.5,1.5) {};
\node(6v1) at (9,1.5) {};


\draw[gray] (-1.5, -0.5) -- (2.25, -0.5);
\draw[gray] (-1.5, 2) -- (2.25, 2);

\draw[gray, dotted](2.25, -0.5)--(5.25, -0.5);
\draw[dotted](2.75, 0.75)--(4.75, 0.75);
\draw[gray, dotted](2.25, 2)--(5.25, 2);

\draw[gray] (5.25, -0.5) -- (9, -0.5);
\draw[gray] (5.25, 2) -- (9, 2);

\draw[gray] (-0.75, -0.5) -- (-0.75, 2);
\draw[gray] (0.75, -0.5) -- (0.75, 2);
\draw[gray, dotted] (2.25, -0.5) -- (2.25, 2);
\draw[gray, dotted] (5.25, -0.5) -- (5.25, 2);
\draw[gray] (6.75, -0.5) -- (6.75, 2);
\draw[gray] (8.25, -0.5) -- (8.25, 2);


\draw[->,dashed] (-1.2,0) -- (0v0);
\draw[<-,dashed] (0v0) -- (1v0);
\draw[->,dashed] (1v0) -- (2v0);
\draw[->,dashed] (3v0) -- (4v0);
\draw[<-,dashed] (4v0) -- (5v0);
\draw[->,dashed] (5v0) -- (6v0);

\draw[<-,dashed] (-1.2,1.5) -- (0v1);
\draw[->,dashed] (0v1) -- (1v1);
\draw[<-,dashed] (1v1) -- (2v1);
\draw[<-,dashed] (3v1) -- (4v1);
\draw[->,dashed] (4v1) -- (5v1);
\draw[<-,dashed] (5v1) -- (6v1);


\draw[<-,dashed] (0,-1.2) -- (0v0);
\draw[->,dashed] (0v0) -- (0v1);
\draw[<-,dashed] (0v1) -- (0,2.7);

\draw[->,dashed] (1.5,-1.2) -- (1v0);
\draw[<-,dashed] (1v0) -- (1v1);
\draw[->,dashed] (1v1) -- (1.5,2.7);

\draw[<-,dashed] (6,-1.2) -- (4v0);
\draw[->,dashed] (4v0) -- (4v1);
\draw[<-,dashed] (4v1) -- (6,2.7);

\draw[->,dashed] (7.5,-1.2) -- (5v0);
\draw[<-,dashed] (5v0) -- (5v1);
\draw[->,dashed] (5v1) -- (7.5,2.7);

\draw[ultra thick, rounded corners=8pt, ->] 
    (-1.3, 0.2) -- (-0.2, 0.2) -- (-0.2, 1.3) -- 
    (1.3, 1.3) -- (1.3,0.2) -- (2.8,0.2) ;
    
\draw[ultra thick, rounded corners=8pt, ->] 
    (-1.3 + 6, 0.2) -- (-0.2 +6, 0.2) -- (-0.2+6, 1.3) -- 
    (1.3+6, 1.3) --  ( 1.3+6, 2.5) ;    

\end{tikzpicture}
\caption{When $C_l$ is the first motion performed.  At most three possible crossings on $OC$ yielding six steps.}
\label{fig:2C}
\end{figure}

\end{proof}

\begin{theorem}\label{pr:theotworow}
For every integer $n\ge2$, one has $S_{2,n}\le 6(n-1)$.
\end{theorem}

\begin{proof}
For $S_2=6$ by simulation. For $n>2$, the grid consists of two boundary columns and $n-2$ interior columns. 
To maximize the number of steps, the ant must perform a crossing motion within each interior column upon its first visit, and a bouncing or initial motion within the boundary columns. Since the initial motion can occur only once, the optimal strategy is to perform an initial motion within one boundary column, a bouncing motion within the other, and crossing motions on all interior columns. Thus by Lemma \ref{pr:2Rbound} $$S_{2,n}\le 4 +3 -1 + 6(n-2)  = 6(n-1). $$ The term $-1$ accounts for the fact that the out motion occurs only once but would otherwise be counted twice in the contribution of the boundary columns.
\end{proof}

This bound is tight. Figure~\ref{fig:2max} exhibits explicit configurations
achieving $S_{2,n} = 6(n-1)$ for all $n \ge 2$, with separate constructions
for even and odd $n$.

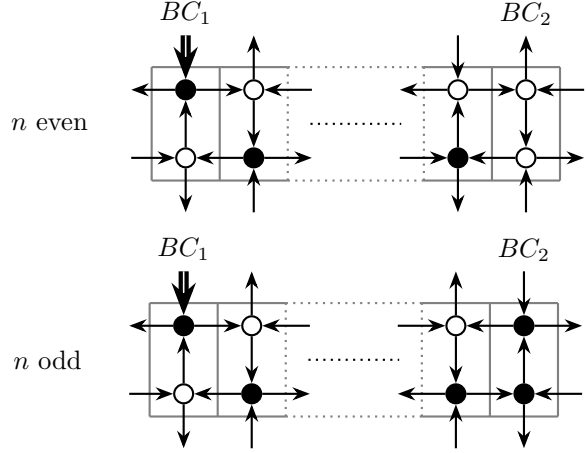
\begin{figure}[H]
\centering

\begin{tikzpicture}[
scale=0.6,
vertex/.style={circle, draw, minimum size=2.5mm, inner sep=0pt},
fvertex/.style={circle, draw, fill=black, minimum size=2.5mm, inner sep=0pt},
gvertex/.style={circle, draw, fill=gray!60, minimum size=2.5mm, inner sep=0pt},
>=Stealth,
thick
]

\node(t1) at (-3,0.75){$n$ even};
\node[vertex] (0v0) at (0,0) {};
\node[fvertex] (1v0) at (1.5,0) {};
\node (2v0) at (3,0) {};
\node (3v0) at (4.5,0) {};
\node[fvertex] (4v0) at (6,0) {};
\node[vertex] (5v0) at (7.5,0) {};
\node(6v0) at (9,0) {};

\node[fvertex] (0v1) at (0,1.5) {};
\node[vertex] (1v1) at (1.5,1.5) {};
\node (2v1) at (3,1.5) {};
\node (3v1) at (4.5,1.5) {};
\node[vertex] (4v1) at (6,1.5) {};
\node[vertex] (5v1) at (7.5,1.5) {};
\node(6v1) at (9,1.5) {};


\draw[gray] (-0.75, -0.5) -- (2.25, -0.5);
\draw[gray] (-0.75, 2) -- (2.25, 2);

\draw[gray, dotted](2.25, -0.5)--(5.25, -0.5);
\draw[dotted](2.75, 0.75)--(4.75, 0.75);
\draw[gray, dotted](2.25, 2)--(5.25, 2);

\draw[gray] (5.25, -0.5) -- (8.25, -0.5);
\draw[gray] (5.25, 2) -- (8.25, 2);

\draw[gray] (-0.75, -0.5) -- (-0.75, 2);
\draw[gray] (0.75, -0.5) -- (0.75, 2);
\draw[gray, dotted] (2.25, -0.5) -- (2.25, 2);
\draw[gray, dotted] (5.25, -0.5) -- (5.25, 2);
\draw[gray] (6.75, -0.5) -- (6.75, 2);
\draw[gray] (8.25, -0.5) -- (8.25, 2);


\draw[->] (-1.2,0) -- (0v0);
\draw[<-] (0v0) -- (1v0);
\draw[->] (1v0) -- (2v0);
\draw[->] (3v0) -- (4v0);
\draw[<-] (4v0) -- (5v0);
\draw[->] (5v0) -- (6v0);

\draw[<-] (-1.2,1.5) -- (0v1);
\draw[->] (0v1) -- (1v1);
\draw[<-] (1v1) -- (2v1);
\draw[<-] (3v1) -- (4v1);
\draw[->] (4v1) -- (5v1);
\draw[<-] (5v1) -- (6v1);


\draw[<-] (0,-1.2) -- (0v0);
\draw[->] (0v0) -- (0v1);
\draw[double, ultra thick, <-] (0v1) -- (0,2.7) node[at end, above] {$BC_1$};

\draw[->] (1.5,-1.2) -- (1v0);
\draw[<-] (1v0) -- (1v1);
\draw[->] (1v1) -- (1.5,2.7);

\draw[<-] (6,-1.2) -- (4v0);
\draw[->] (4v0) -- (4v1);
\draw[<-] (4v1) -- (6,2.7);

\draw[->] (7.5,-1.2) -- (5v0);
\draw[<-] (5v0) -- (5v1);
\draw[->] (5v1) -- (7.5,2.7) node[at end, above] {$BC_2$};
\end{tikzpicture}

\vspace{0.5em}

\begin{tikzpicture}[
scale=0.6,
vertex/.style={circle, draw, minimum size=2.5mm, inner sep=0pt},
fvertex/.style={circle, draw, fill=black, minimum size=2.5mm, inner sep=0pt},
gvertex/.style={circle, draw, fill=gray!60, minimum size=2.5mm, inner sep=0pt},
>=Stealth,
thick
]

\node(t1) at (-3,0.75){$n$ odd};
\node[vertex] (0v0) at (0,0) {};
\node[fvertex] (1v0) at (1.5,0) {};
\node (2v0) at (3,0) {};
\node (3v0) at (4.5,0) {};
\node[fvertex] (4v0) at (6,0) {};
\node[fvertex] (5v0) at (7.5,0) {};
\node(6v0) at (9,0) {};

\node[fvertex] (0v1) at (0,1.5) {};
\node[vertex] (1v1) at (1.5,1.5) {};
\node (2v1) at (3,1.5) {};
\node (3v1) at (4.5,1.5) {};
\node[vertex] (4v1) at (6,1.5) {};
\node[fvertex] (5v1) at (7.5,1.5) {};
\node(6v1) at (9,1.5) {};


\draw[gray] (-0.75, -0.5) -- (2.25, -0.5);
\draw[gray] (-0.75, 2) -- (2.25, 2);

\draw[gray, dotted](2.25, -0.5)--(5.25, -0.5);
\draw[dotted](2.75, 0.75)--(4.75, 0.75);
\draw[gray, dotted](2.25, 2)--(5.25, 2);

\draw[gray] (5.25, -0.5) -- (8.25, -0.5);
\draw[gray] (5.25, 2) -- (8.25, 2);

\draw[gray] (-0.75, -0.5) -- (-0.75, 2);
\draw[gray] (0.75, -0.5) -- (0.75, 2);
\draw[gray, dotted] (2.25, -0.5) -- (2.25, 2);
\draw[gray, dotted] (5.25, -0.5) -- (5.25, 2);
\draw[gray] (6.75, -0.5) -- (6.75, 2);
\draw[gray] (8.25, -0.5) -- (8.25, 2);


\draw[->] (-1.2,0) -- (0v0);
\draw[<-] (0v0) -- (1v0);
\draw[->] (1v0) -- (2v0);
\draw[<-] (3v0) -- (4v0);
\draw[->] (4v0) -- (5v0);
\draw[<-] (5v0) -- (6v0);

\draw[<-] (-1.2,1.5) -- (0v1);
\draw[->] (0v1) -- (1v1);
\draw[<-] (1v1) -- (2v1);
\draw[->] (3v1) -- (4v1);
\draw[<-] (4v1) -- (5v1);
\draw[->] (5v1) -- (6v1);


\draw[<-] (0,-1.2) -- (0v0);
\draw[->] (0v0) -- (0v1);
\draw[double, ultra thick, <-] (0v1) -- (0,2.7) node[at end, above] {$BC_1$};

\draw[->] (1.5,-1.2) -- (1v0);
\draw[<-] (1v0) -- (1v1);
\draw[->] (1v1) -- (1.5,2.7);

\draw[->] (6,-1.2) -- (4v0);
\draw[<-] (4v0) -- (4v1);
\draw[->] (4v1) -- (6,2.7);

\draw[<-] (7.5,-1.2) -- (5v0);
\draw[->] (5v0) -- (5v1);
\draw[<-] (5v1) -- (7.5,2.7) node[at end, above] {$BC_2$};

\end{tikzpicture}
\caption{Configurations achieving $S_{2,n}=6(n-1)$. The initial position of the ant is indicated by the double arrow. All intermediate columns have configuration ${R \brack L}$. The ant starts with an initial motion within $BC_1$, traverses the intermediate columns to $BC_2$ where it bounces, returns to $BC_1$ for a second bounce, and finally reaches $BC_2$ again to perform an out motion.}
\label{fig:2max}
\end{figure}

\subsubsection{Three-row Grids}\label{sec:threerow}

When viewed column by column, a three-row grid (either vertical or horizontal) can be seen as an alternating
succession of horizontal and vertical $G_{3,1}$ ( see Figure$~\ref{fig:3rowgrid}$).
There are four possible types of motion that the ant can perform on a column of a three-row grid: Initial, Bouncing, Crossing and Out motions. However, no Out motion can occur on a vertical $G_{3,1}$, and no Initial motion can occur on a horizontal $G_{3,1}$.

\begin{figure}[H]
\centering

\begin{subfigure}{0.45\textwidth}
\centering
\begin{tikzpicture}[
scale=0.5,
vertex/.style={circle, draw, minimum size=2.5mm, inner sep=0pt},
fvertex/.style={circle, draw, fill=black, minimum size=2.5mm, inner sep=0pt},
gvertex/.style={circle, draw, fill=gray!60, minimum size=2.5mm, inner sep=0pt},
>=Stealth,
thick
]
\def\k{3} 
\def\n{7} 
\def\spacing{1.5} 
\def\externalarrow{1.2} 
\def\margin{0.5} 
\foreach \i in {0,...,\numexpr\n-1} {
 \foreach \j in {0,...,\numexpr\k-1} {
 \pgfmathsetmacro{\x}{\i*\spacing}
 \pgfmathsetmacro{\y}{\j*\spacing}
 \node[gvertex] (\i v\j) at (\x,\y) {};
 }
}
\foreach \i/\j in {} {
 \pgfmathsetmacro{\x}{\i*\spacing}
 \pgfmathsetmacro{\y}{\j*\spacing}
 \node[fvertex] (\i v\j) at (\x,\y) {};
}
\pgfmathsetmacro{\upbound}{(\k-1)*\spacing + \margin}
\pgfmathsetmacro{\bottombound}{-\margin}
\draw[gray] (-\spacing, \upbound) -- ({(\n-1)*\spacing + \spacing}, \upbound);
\draw[gray] (-\spacing, \bottombound) -- ({(\n-1)*\spacing + \spacing}, \bottombound);
\foreach \i in {0,...,\n} {
 \pgfmathsetmacro{\x}{\i*\spacing - \spacing/2}
 \draw[gray] (\x, \bottombound) -- (\x, \upbound);
}
\foreach \j in {0,...,\numexpr\k-1} {
 \pgfmathsetmacro{\y}{\j*\spacing}
 \pgfmathsetmacro{\dir}{int(mod(\j,2)==0?1:0)}
 \ifnum\dir=1
 \draw[->] (-\externalarrow,\y) -- (0v\j);
 \else
 \draw[<-] (-\externalarrow,\y) -- (0v\j);
 \fi
 \foreach \i in {0,...,\numexpr\n-2} {
 \pgfmathsetmacro{\arrowdir}{int(mod(\i+\j,2)==0?1:0)}
 \ifnum\arrowdir=1
 \draw[<-] (\i v\j) -- (\the\numexpr\i+1\relax v\j);
 \else
 \draw[->] (\i v\j) -- (\the\numexpr\i+1\relax v\j);
 \fi
 }
 \pgfmathsetmacro{\xright}{(\n-1)*\spacing + \externalarrow}
 \pgfmathsetmacro{\lastdir}{int(mod(\n-1+\j,2)==0?1:0)}
 \ifnum\lastdir=1
 \draw[<-] (\the\numexpr\n-1\relax v\j) -- (\xright,\y);
 \else
 \draw[->] (\the\numexpr\n-1\relax v\j) -- (\xright,\y);
 \fi
}
\foreach \i in {0,...,\numexpr\n-1} {
 \pgfmathsetmacro{\x}{\i*\spacing}
 \pgfmathsetmacro{\dir}{int(mod(\i,2)==0?1:0)}
 \ifnum\dir=1
 \draw[<-] (\x,-\externalarrow) -- (\i v0)node[at start, below] {H};
 \else
 \draw[->] (\x,-\externalarrow) -- (\i v0)node[at start, below] {V};
 \fi
 \foreach \j in {0,...,\numexpr\k-2} {
 \pgfmathsetmacro{\arrowdir}{int(mod(\i+\j,2)==0?1:0)}
 \ifnum\arrowdir=1
 \draw[->] (\i v\j) -- (\i v\the\numexpr\j+1\relax);
 \else
 \draw[<-] (\i v\j) -- (\i v\the\numexpr\j+1\relax);
 \fi
 }
 \pgfmathsetmacro{\yup}{(\k-1)*\spacing + \externalarrow}
 \pgfmathsetmacro{\lastdir}{int(mod(\i+\k-1,2)==0?1:0)}
 \ifnum\lastdir=1
 \draw[->] (\i v\the\numexpr\k-1\relax) -- (\x,\yup);
 \else
 \draw[<-] (\i v\the\numexpr\k-1\relax) -- (\x,\yup);
 \fi
}

\end{tikzpicture}
\caption{Portion of three-row grid }
\label{fig:grid37c}
\end{subfigure}
\hfill
\begin{subfigure}{0.25\textwidth}
\centering
\begin{tikzpicture}[
scale=0.5,
vertex/.style={circle, draw, minimum size=2.5mm, inner sep=0pt},
fvertex/.style={circle, draw, fill=black, minimum size=2.5mm, inner sep=0pt},
gvertex/.style={circle, draw, fill=gray!60, minimum size=2.5mm, inner sep=0pt},
>=Stealth,
thick
]
\def\spacing{1.5}
\def\externalarrow{1.2}
\def\margin{0.5}
\node[gvertex] (0v0) at (0,0) {};
\node[gvertex] (0v1) at (0,1.5) {};
\node[gvertex] (0v2) at (0,3) {};
\draw[gray] (-\spacing/2, -\margin) -- (\spacing/2, -\margin);
\draw[gray] (-\spacing/2, 3+ \margin) -- (\spacing/2, 3 + \margin);
\draw[gray] (-\spacing/2, -\margin) -- (-\spacing/2, 3+ \margin);
\draw[gray] (\spacing/2, -\margin) -- (\spacing/2, 3+ \margin);

\draw[->] (-\externalarrow, 0) -- (0v0);
\draw[<-] (0v0) -- (\externalarrow, 0);
\draw[<-] (-\externalarrow, 1.5) -- (0v1);
\draw[->] (0v1) -- (\externalarrow, 1.5);
\draw[->] (-\externalarrow, 3) -- (0v2);
\draw[<-] (0v2) -- (\externalarrow, 3);

\draw[<-] (0, -\externalarrow) -- (0v0)node[at start, below] {H};
\draw[->] (0v0) -- (0v1);
\draw[<-] (0v1) -- (0v2);
\draw[<-] (0,4.2) -- (0v2);
\end{tikzpicture}
\caption{Horizontal $G_{3,1}$}
\label{fig:grid31h}
\end{subfigure}
\hfill
\begin{subfigure}{0.2\textwidth}
\centering
\begin{tikzpicture}[
scale=0.5,
vertex/.style={circle, draw, minimum size=2.5mm, inner sep=0pt},
fvertex/.style={circle, draw, fill=black, minimum size=2.5mm, inner sep=0pt},
gvertex/.style={circle, draw, fill=gray!60, minimum size=2.5mm, inner sep=0pt},
>=Stealth,
thick
]
\def\spacing{1.5}
\def\externalarrow{1.2}
\def\margin{0.5}
\node[gvertex] (0v0) at (0,0) {};
\node[gvertex] (0v1) at (0,1.5) {};
\node[gvertex] (0v2) at (0,3) {};
\draw[gray] (-\spacing/2, -\margin) -- (\spacing/2, -\margin);
\draw[gray] (-\spacing/2, 3+ \margin) -- (\spacing/2, 3 + \margin);
\draw[gray] (-\spacing/2, -\margin) -- (-\spacing/2, 3+ \margin);
\draw[gray] (\spacing/2, -\margin) -- (\spacing/2, 3+ \margin);
\draw[<-] (-\externalarrow, 0) -- (0v0);
\draw[->] (0v0) -- (\externalarrow, 0);
\draw[->] (-\externalarrow, 1.5) -- (0v1);
\draw[<-] (0v1) -- (\externalarrow, 1.5);
\draw[<-] (-\externalarrow, 3) -- (0v2);
\draw[->] (0v2) -- (\externalarrow, 3);
\draw[->] (0, -\externalarrow) -- (0v0)node[at start, below] {V};
\draw[<-] (0v0) -- (0v1);
\draw[->] (0v1) -- (0v2);
\draw[->] (0,4.2) -- (0v2);
\end{tikzpicture}
\caption{Vertical $G_{3,1}$}
\label{fig:grid31v}
\end{subfigure}

\caption{Three-row grid viewed by columns.}
\label{fig:3rowgrid}
\end{figure}
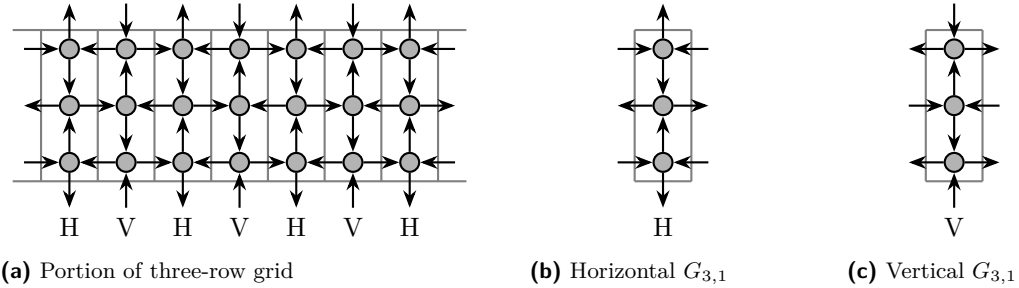

The main result of this section is:
\begin{theorem}\label{pr:theothreerow}
For every integer $n\ge2$, one has $S_{3,n} \le 10n-4$. 
\end{theorem}
The proof of Theorem \ref{pr:theothreerow} relies on Corollary~\ref{pr:Vcolbound}, which bounds the 
number of steps the ant can perform within each vertical column 
depending on the type of its first motion.  This corollary itself 
follows from two lemmas - Lemma~\ref{pr:Hcrossing} and 
Lemma~\ref{pr:Hbouncing} - bounding the additional steps after a 
crossing and a bouncing motion respectively, both of which rely on 
the blocking behavior of horizontal columns established in 
Lemma~\ref{pr:Hblocking}. We now state and prove these results.

We start with a result concerning horizontal $G_{3,1}$.
Let's label its four in-boundary arcs  by $lt_{in},lb_{in},rt_{in}$
and $rb_{in}$ and its four out-boundary arcs by  $l_{out},r_{out},t_{out}$ and
$b_{out}$ as illustrated on Figure$~\ref{fig:Hgrid31_io}$. 

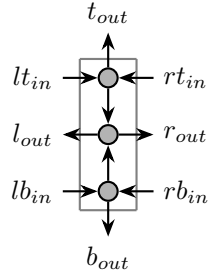
\begin{figure}[H]
\centering

\centering
\begin{tikzpicture}[
scale=0.5,
vertex/.style={circle, draw, minimum size=2.5mm, inner sep=0pt},
fvertex/.style={circle, draw, fill=black, minimum size=2.5mm, inner sep=0pt},
gvertex/.style={circle, draw, fill=gray!60, minimum size=2.5mm, inner sep=0pt},
>=Stealth,
thick
]

\def\spacing{1.5}
\def\externalarrow{1.2}
\def\margin{0.5}
\node[gvertex] (0v0) at (0,0) {};
\node[gvertex] (0v1) at (0,1.5) {};
\node[gvertex] (0v2) at (0,3) {};

\draw[gray] (-\spacing/2, -\margin) -- (\spacing/2, -\margin);
\draw[gray] (-\spacing/2, 3+ \margin) -- (\spacing/2, 3 + \margin);
\draw[gray] (-\spacing/2, -\margin) -- (-\spacing/2, 3+ \margin);
\draw[gray] (\spacing/2, -\margin) -- (\spacing/2, 3+ \margin);
\draw[->] (-\externalarrow, 0) -- (0v0) node[at start, left] {$lb_{in}$};
\draw[<-] (0v0) -- (\externalarrow, 0) node[at end, right] {$rb_{in}$};
\draw[<-] (-\externalarrow, 1.5) -- (0v1) node[at start, left] {$l_{out}$};
\draw[->] (0v1) -- (\externalarrow, 1.5) node[at end, right] {$r_{out}$};
\draw[->] (-\externalarrow, 3) -- (0v2)node[at start, left] {$lt_{in}$};
\draw[<-] (0v2) -- (\externalarrow, 3)node[at end, right] {$rt_{in}$};

\draw[<-] (0, -\externalarrow) -- (0v0) node[at start, below] {$b_{out}$};
\draw[->] (0v0) -- (0v1);
\draw[<-] (0v1) -- (0v2) ;
\draw[<-] (0,4.2) -- (0v2)node[at start, above] {$t_{out}$};
\end{tikzpicture}
\caption{Horizontal $G_{3,1}$ with labeled in-boundary and out-boundary arcs.}
\label{fig:Hgrid31_io}
\end{figure}

\begin{lemma}\label{pr:Hblocking}
During the motion of the ant on a three-row grid, suppose the ant performs 
a crossing motion within a column that is a horizontal $G_{3,1}$ and subsequently 
re-enters it via an arc incident to a cell not exited during that 
crossing. Then the ant cannot perform any further crossing motion through 
this column before escaping the three-row grid.
\end{lemma}

\begin{proof}
Let $OC$ denote the horizontal $G_{3,1}$ under consideration. We denote by 
\scalebox{0.6}{$\left[\begin{array}{c}C_{1}\\C_{2}\\C_{3}\end{array}\right]$} the color configuration of $OC$, where $C_{1},C_{2},C_{3}\in\{L,R\}$ correspond respectively to the state of the top, middle, and bottom cells.
Assume without loss of generality that the ant crosses $OC$ from left to right using cells of the top and the middle rows (enter through $lt_{in}$ and exit through $r_{out}$). The argument for the other types of crossings follows by symmetry. Figure~\ref{fig:3Blocking} illustrates the arguments that follow.

Before the crossing, the color configuration of the column is either  \scalebox{0.6}{$\left[\begin{array}{c}R\\L\\L\end{array}\right]$} or  \scalebox{0.6}{$\left[\begin{array}{c}R\\L\\R\end{array}\right]$}. 
After the crossing, the ant can subsequently return to the column only from the right and in order to re-enter it through an arc incident to a cell  different from the top and the middle cells, it must enter through $rb_{in}$. 

\textbf{Case A.} Suppose the initial color configuration  is \scalebox{0.6}{$\left[\begin{array}{c}R\\L\\L\end{array}\right]$}. 
 After the crossing, it becomes \scalebox{0.6}{$\left[\begin{array}{c}L\\R\\L\end{array}\right]$}. When the ant returns  through $rb_{in}$, it must perform an out motion reaching $b_{out}$, thereby escaping the three-row grid.
 
 \textbf{Case B.} Suppose the initial color configuration is
 \scalebox{0.6}{$\left[\begin{array}{c}R\\L\\R\end{array}\right]$}.
 After the crossing, it becomes \scalebox{0.6}{$\left[\begin{array}{c}L\\R\\R\end{array}\right]$}. When the ant returns through $rb_{in}$, it performs a bouncing motion updating the color configuration to 
 \scalebox{0.6}{$\left[\begin{array}{c}L\\L\\L\end{array}\right]$}. The ant can subsequently return to the column only from the right through $rb_{in}$ (\textbf{B.A}) or $rt_{in}$ (\textbf{B.B}). 
 
 \textbf{Case B.A.} If the ant again returns through $rb_{in}$, it must perform an out motion reaching $b_{out}$.
 
 \textbf{Case B.B.} If instead, the ant returns through $rt_{in}$, it performs a bouncing motion, yielding the color configuration \scalebox{0.6}{$\left[\begin{array}{c}R\\R\\L\end{array}\right]$}. 
 With this configuration, any subsequent return through either $rb_{in}$ (\textbf{B.B.A}) or $rt_{in}$ (\textbf{B.B.B}) necessarily leads to $b_{out}$ or $t_{out}$, respectively.
 
In all cases, the ant escapes the three-row grid without performing any further crossing motion through $OC$.

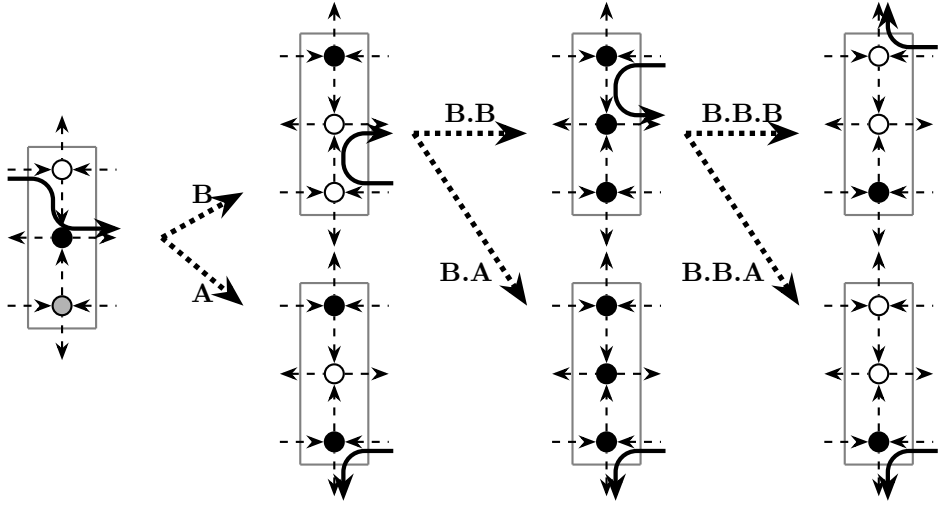
\begin{figure}[H]
\centering

\begin{tikzpicture}[
scale=0.6,
vertex/.style={circle, draw, minimum size=2.5mm, inner sep=0pt},
fvertex/.style={circle, draw, fill=black, minimum size=2.5mm, inner sep=0pt},
gvertex/.style={circle, draw, fill=gray!60, minimum size=2.5mm, inner sep=0pt},
>=Stealth,
thick
]

\begin{scope}[shift={(0, 0)}]
\node[gvertex] (start_v0) at (0,0) {};
\node[fvertex] (start_v1) at (0,1.5) {};
\node[vertex] (start_v2) at (0,3) {};

\draw[gray] (-0.75, -0.5) -- (0.75, -0.5);
\draw[gray] (-0.75, 3.5) -- (0.75, 3.5);
\draw[gray] (-0.75, -0.5) -- (-0.75, 3.5);
\draw[gray] (0.75, -0.5) -- (0.75, 3.5);

\draw[->, dashed] (-1.2, 0) -- (start_v0);
\draw[<-, dashed] (start_v0) -- (1.2, 0);
\draw[<-, dashed] (-1.2, 1.5) -- (start_v1);
\draw[->, dashed] (start_v1) -- (1.2, 1.5);
\draw[->, dashed] (-1.2, 3) -- (start_v2);
\draw[<-, dashed] (start_v2) -- (1.2, 3);

\draw[<-, dashed] (0, -1.2) -- (start_v0);
\draw[->, dashed] (start_v0) -- (start_v1);
\draw[<-, dashed] (start_v1) -- (start_v2);
\draw[<-, dashed] (0, 4.2) -- (start_v2);

\draw[ultra thick, rounded corners=8pt, ->] 
    (-1.2, 2.8) -- (-0.2, 2.8) -- (-0.2, 1.7) -- (1.3, 1.7);

\end{scope}

\begin{scope}[shift={(6, -3)}]
\node[fvertex] (caseA_v0) at (0,0) {};
\node[vertex] (caseA_v1) at (0,1.5) {};
\node[fvertex] (caseA_v2) at (0,3) {};

\draw[gray] (-0.75, -0.5) -- (0.75, -0.5);
\draw[gray] (-0.75, 3.5) -- (0.75, 3.5);
\draw[gray] (-0.75, -0.5) -- (-0.75, 3.5);
\draw[gray] (0.75, -0.5) -- (0.75, 3.5);

\draw[->, dashed] (-1.2, 0) -- (caseA_v0);
\draw[<-, dashed] (caseA_v0) -- (1.2, 0);
\draw[<-, dashed] (-1.2, 1.5) -- (caseA_v1);
\draw[->, dashed] (caseA_v1) -- (1.2, 1.5);
\draw[->, dashed] (-1.2, 3) -- (caseA_v2);
\draw[<-, dashed] (caseA_v2) -- (1.2, 3);

\draw[<-, dashed] (0, -1.2) -- (caseA_v0);
\draw[->, dashed] (caseA_v0) -- (caseA_v1);
\draw[<-, dashed] (caseA_v1) -- (caseA_v2);
\draw[<-, dashed] (0, 4.2) -- (caseA_v2);

\draw[ultra thick, rounded corners=8pt, ->] 
    (1.3, -0.2) -- (0.2, -0.2) -- (0.2, -1.3);

\end{scope}

\begin{scope}[shift={(6, 2.5)}]
\node[vertex] (caseB_v0) at (0,0) {};
\node[vertex] (caseB_v1) at (0,1.5) {};
\node[fvertex] (caseB_v2) at (0,3) {};

\draw[gray] (-0.75, -0.5) -- (0.75, -0.5);
\draw[gray] (-0.75, 3.5) -- (0.75, 3.5);
\draw[gray] (-0.75, -0.5) -- (-0.75, 3.5);
\draw[gray] (0.75, -0.5) -- (0.75, 3.5);

\draw[->, dashed] (-1.2, 0) -- (caseB_v0);
\draw[<-, dashed] (caseB_v0) -- (1.2, 0);
\draw[<-, dashed] (-1.2, 1.5) -- (caseB_v1);
\draw[->, dashed] (caseB_v1) -- (1.2, 1.5);
\draw[->, dashed] (-1.2, 3) -- (caseB_v2);
\draw[<-, dashed] (caseB_v2) -- (1.2, 3);

\draw[<-, dashed] (0, -1.2) -- (caseB_v0);
\draw[->, dashed] (caseB_v0) -- (caseB_v1);
\draw[<-, dashed] (caseB_v1) -- (caseB_v2);
\draw[<-, dashed] (0, 4.2) -- (caseB_v2);

\draw[ultra thick, rounded corners=8pt, ->] 
    (1.3, 0.2) -- (0.2, 0.2) -- (0.2, 1.3) -- (1.3, 1.3);

\end{scope}

\begin{scope}[shift={(12, -3)}]
\node[fvertex] (caseBA_v0) at (0,0) {};
\node[fvertex] (caseBA_v1) at (0,1.5) {};
\node[fvertex] (caseBA_v2) at (0,3) {};

\draw[gray] (-0.75, -0.5) -- (0.75, -0.5);
\draw[gray] (-0.75, 3.5) -- (0.75, 3.5);
\draw[gray] (-0.75, -0.5) -- (-0.75, 3.5);
\draw[gray] (0.75, -0.5) -- (0.75, 3.5);

\draw[->, dashed] (-1.2, 0) -- (caseBA_v0);
\draw[<-, dashed] (caseBA_v0) -- (1.2, 0);
\draw[<-, dashed] (-1.2, 1.5) -- (caseBA_v1);
\draw[->, dashed] (caseBA_v1) -- (1.2, 1.5);
\draw[->, dashed] (-1.2, 3) -- (caseBA_v2);
\draw[<-, dashed] (caseBA_v2) -- (1.2, 3);

\draw[<-, dashed] (0, -1.2) -- (caseBA_v0);
\draw[->, dashed] (caseBA_v0) -- (caseBA_v1);
\draw[<-, dashed] (caseBA_v1) -- (caseBA_v2);
\draw[<-, dashed] (0, 4.2) -- (caseBA_v2);

\draw[ultra thick, rounded corners=8pt, ->] 
    (1.3, -0.2) -- (0.2, -0.2) -- (0.2, -1.3);

\end{scope}

\begin{scope}[shift={(12, 2.5)}]
\node[fvertex] (caseBB_v0) at (0,0) {};
\node[fvertex] (caseBB_v1) at (0,1.5) {};
\node[fvertex] (caseBB_v2) at (0,3) {};

\draw[gray] (-0.75, -0.5) -- (0.75, -0.5);
\draw[gray] (-0.75, 3.5) -- (0.75, 3.5);
\draw[gray] (-0.75, -0.5) -- (-0.75, 3.5);
\draw[gray] (0.75, -0.5) -- (0.75, 3.5);

\draw[->, dashed] (-1.2, 0) -- (caseBB_v0);
\draw[<-, dashed] (caseBB_v0) -- (1.2, 0);
\draw[<-, dashed] (-1.2, 1.5) -- (caseBB_v1);
\draw[->, dashed] (caseBB_v1) -- (1.2, 1.5);
\draw[->, dashed] (-1.2, 3) -- (caseBB_v2);
\draw[<-, dashed] (caseBB_v2) -- (1.2, 3);

\draw[<-, dashed] (0, -1.2) -- (caseBB_v0);
\draw[->, dashed] (caseBB_v0) -- (caseBB_v1);
\draw[<-, dashed] (caseBB_v1) -- (caseBB_v2);
\draw[<-, dashed] (0, 4.2) -- (caseBB_v2);

\draw[ultra thick, rounded corners=8pt, ->] 
    (1.3, 2.8) -- (0.2, 2.8) -- (0.2, 1.7) -- (1.3, 1.7);

\end{scope}

\begin{scope}[shift={(18, -3)}]
\node[fvertex] (caseBBA_v0) at (0,0) {};
\node[vertex] (caseBBA_v1) at (0,1.5) {};
\node[vertex] (caseBBA_v2) at (0,3) {};

\draw[gray] (-0.75, -0.5) -- (0.75, -0.5);
\draw[gray] (-0.75, 3.5) -- (0.75, 3.5);
\draw[gray] (-0.75, -0.5) -- (-0.75, 3.5);
\draw[gray] (0.75, -0.5) -- (0.75, 3.5);

\draw[->, dashed] (-1.2, 0) -- (caseBBA_v0);
\draw[<-, dashed] (caseBBA_v0) -- (1.2, 0);
\draw[<-, dashed] (-1.2, 1.5) -- (caseBBA_v1);
\draw[->, dashed] (caseBBA_v1) -- (1.2, 1.5);
\draw[->, dashed] (-1.2, 3) -- (caseBBA_v2);
\draw[<-, dashed] (caseBBA_v2) -- (1.2, 3);

\draw[<-, dashed] (0, -1.2) -- (caseBBA_v0);
\draw[->, dashed] (caseBBA_v0) -- (caseBBA_v1);
\draw[<-, dashed] (caseBBA_v1) -- (caseBBA_v2);
\draw[<-, dashed] (0, 4.2) -- (caseBBA_v2);

\draw[ultra thick, rounded corners=8pt, ->] 
    (1.3, -0.2) -- (0.2, -0.2) -- (0.2, -1.3);

\end{scope}

\begin{scope}[shift={(18, 2.5)}]
\node[fvertex] (caseBBB_v0) at (0,0) {};
\node[vertex] (caseBBB_v1) at (0,1.5) {};
\node[vertex] (caseBBB_v2) at (0,3) {};

\draw[gray] (-0.75, -0.5) -- (0.75, -0.5);
\draw[gray] (-0.75, 3.5) -- (0.75, 3.5);
\draw[gray] (-0.75, -0.5) -- (-0.75, 3.5);
\draw[gray] (0.75, -0.5) -- (0.75, 3.5);

\draw[->, dashed] (-1.2, 0) -- (caseBBB_v0);
\draw[<-, dashed] (caseBBB_v0) -- (1.2, 0);
\draw[<-, dashed] (-1.2, 1.5) -- (caseBBB_v1);
\draw[->, dashed] (caseBBB_v1) -- (1.2, 1.5);
\draw[->, dashed] (-1.2, 3) -- (caseBBB_v2);
\draw[<-, dashed] (caseBBB_v2) -- (1.2, 3);

\draw[<-, dashed] (0, -1.2) -- (caseBBB_v0);
\draw[->, dashed] (caseBBB_v0) -- (caseBBB_v1);
\draw[<-, dashed] (caseBBB_v1) -- (caseBBB_v2);
\draw[<-, dashed] (0, 4.2) -- (caseBBB_v2);

\draw[ultra thick, rounded corners=8pt, ->] 
    (1.3, 3.2) -- (0.2, 3.2) -- (0.2, 4.2);

\end{scope}

\draw[line width=2pt, ->, dotted] (2.2, 1.5) -- (4, 0)
node[midway, below] {\textbf{A}};

\draw[line width=2pt, ->, dotted] (2.2, 1.5) -- (4, 2.5)
node[midway, above] {\textbf{B}};

\draw[line width=2pt, ->, dotted] (7.75, 3.8) -- (10.25, 0)
node[pos=0.8, left]  {\textbf{B.A}};

\draw[line width=2pt, ->, dotted] (7.75, 3.8) -- (10.25, 3.8)
node[midway, above]{\textbf{B.B}};

\draw[line width=2pt, ->, dotted] (13.75, 3.8) -- (16.25, 0)
node[pos=0.8, left] {\textbf{B.B.A}};

\draw[line width=2pt, ->, dotted] (13.75, 3.8) -- (16.25, 3.8)
node[midway, above] {\textbf{B.B.B}};

\end{tikzpicture}

\caption{Illustration of proof of Lemma \ref{pr:Hblocking}. }
\label{fig:3Blocking}
\end{figure}

\end{proof}

For simplicity, throughout the remainder of this section, we say that a horizontal column is \emph{blocking} if it is in a color configuration such that the ant cannot perform any further crossing motion through it before escaping the entire three-row grid.

We now turn to vertical columns. Label the in-boundary and out-boundary 
arcs of a vertical $G_{3,1}$ by $l_{in},r_{in},t_{in}$ and
$b_{in}$ and its four out-boundary arcs by $lt_{out},lb_{out},rt_{out}$
and $rb_{out}$ as illustrated on Figure$~\ref{fig:Vgrid31_io}$. 

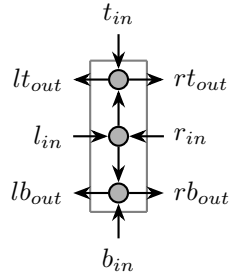
\begin{figure}[H]

\centering
\begin{tikzpicture}[
scale=0.5,
vertex/.style={circle, draw, minimum size=2.5mm, inner sep=0pt},
fvertex/.style={circle, draw, fill=black, minimum size=2.5mm, inner sep=0pt},
gvertex/.style={circle, draw, fill=gray!60, minimum size=2.5mm, inner sep=0pt},
>=Stealth,
thick
]

\def\spacing{1.5}
\def\externalarrow{1.2}
\def\margin{0.5}
\node[gvertex] (0v0) at (0,0) {};
\node[gvertex] (0v1) at (0,1.5) {};
\node[gvertex] (0v2) at (0,3) {};

\draw[gray] (-\spacing/2, -\margin) -- (\spacing/2, -\margin);
\draw[gray] (-\spacing/2, 3+ \margin) -- (\spacing/2, 3 + \margin);
\draw[gray] (-\spacing/2, -\margin) -- (-\spacing/2, 3+ \margin);
\draw[gray] (\spacing/2, -\margin) -- (\spacing/2, 3+ \margin);
\draw[<-] (-\externalarrow, 0) -- (0v0) node[at start, left] {$lb_{out}$};
\draw[->] (0v0) -- (\externalarrow, 0) node[at end, right] {$rb_{out}$};
\draw[->] (-\externalarrow, 1.5) -- (0v1) node[at start, left] {$l_{in}$};
\draw[<-] (0v1) -- (\externalarrow, 1.5) node[at end, right] {$r_{in}$};
\draw[<-] (-\externalarrow, 3) -- (0v2)node[at start, left] {$lt_{out}$};
\draw[->] (0v2) -- (\externalarrow, 3)node[at end, right] {$rt_{out}$};

\draw[->] (0, -\externalarrow) -- (0v0) node[at start, below] {$b_{in}$};
\draw[<-] (0v0) -- (0v1);
\draw[->] (0v1) -- (0v2) ;
\draw[->] (0,4.2) -- (0v2)node[at start, above] {$t_{in}$};
\end{tikzpicture}

\caption{Vertical $G_{3,1}$ with labeled in-boundary and out-boundary arcs}
\label{fig:Vgrid31_io}
\end{figure}

When the ant enters through $l_{in}$ or $r_{in}$, it can exit through
any of the four exit arcs. In contrast, when the ant enters through
$t_{in}$ it can exit only through $lt_{out}$ or $rt_{out}$ and
when the ant enters through $b_{in}$ it can exit only through $lb_{out}$
or $rb_{out}$. We associate with each admissible entry--exit pair
a label to describe the motion of the ant, as summarized in Table$~\ref{tab:motiontable3}$.

\begin{table}[htbp]
\centering
\scalebox{0.9}{
\begin{tabular}{|c||c|c|c|c|c|c|c|c|c|c|c|c|}
\hline
 Entry & $t_{in}$ & $t_{in}$ & $b_{in}$ & $b_{in}$ & 
$l_{in}$ & $l_{in}$ & $l_{in}$  & $l_{in}$ & 
$r_{in}$ & $r_{in}$ & $r_{in}$  & $r_{in}$  \\
\hline
 Exit & $lt_{out}$ & $rt_{out}$ & $lb_{out}$ & $rb_{out}$ &
$lt_{out}$ & $lb_{out}$ & $rt_{out}$ & $rb_{out}$ & 
$rt_{out}$ & $rb_{out}$ & $lt_{out}$ & $lb_{out}$  \\
\hline
 Motion & $R$ & $L$ & $L$ & $R$ &
$LL$   & $RR$ & $LR$ & $RL$  & 
$RR$   & $LL$ & $RL$ & $LR$ \\
\hline
 Label & $I_{lt}$ & $I_{rt}$ & $I_{lb}$ & $I_{rb}$ &
$B_{lt}$   & $B_{lb}$ & $C_{lt}$ & $C_{lb}$  & 
$B_{rt}$   & $B_{rb}$ & $C_{rt}$ & $C_{rb}$ \\
\hline
\end{tabular}
}
\caption{Motions and labels associated to admissible entry-exit pairs. Initial motions are performed in one step, bouncing and crossing motions in two steps.}
\label{tab:motiontable3}
\end{table}

\begin{lemma}\label{pr:Hcrossing}
During its motion within a three-row grid, after performing a crossing motion within a column that is a vertical $G_{3,1}$, the ant can subsequently perform at most four additional steps on that column before escaping the three-row grid. Moreover, these additional four steps consist of two crossing motions in opposite directions.
\end{lemma}

\begin{proof}
Let $OC$ denote the vertical $G_{3,1}$ under consideration.  Without loss of generality, we may suppose that the ant performs $C_{lt}$ within $OC$. The argument for the other types of crossings follows by symmetry.

After performing $C_{lt}$, the ant may subsequently re-enter $OC$ only from the right, through $r_{in}$ and in order for the ant to return to $OC$, it must first perform a bouncing motion within some column $RC$ located to the right of $OC$. 
The forward motion, which includes the crossing of $OC$, the traversal of the eventual intermediate columns between $OC$ and $RC$, and the bounce on $RC$, may be performed in one of the following two ways: either \textbf{(A)} using only cells of the top and middle rows of the three-row grid or \textbf{(B)} using cells of all three rows.

\textbf{Case A.}
In the case the forward motion is performed exclusively using cells of the top and middle rows, the ant behaves as if it were moving on a two-row grid. Consequently, after bouncing on $RC$, it may return toward $OC$, perform $C_{rt}$ and eventually another $C_{lt}$ on $OC$, and return to $RC$.
The column $RC$ may be either \textbf{(A.A)} horizontal or \textbf{(A.B)} vertical. 

\textbf{Case A.A.} If $RC$ is horizontal, then upon returning to it the ant reaches $t_{out}$ in a single step. Hence, it escapes the entire grid and does not return to $OC$. See Figure~\ref{fig:A.A} for an illustration of the motion of the ant.

\begin{figure}[H]
\centering
\begin{tikzpicture}[
scale=0.6,
vertex/.style={circle, draw, minimum size=2.5mm, inner sep=0pt},
fvertex/.style={circle, draw, fill=black, minimum size=2.5mm, inner sep=0pt},
gvertex/.style={circle, draw, fill=gray!60, minimum size=2.5mm, inner sep=0pt},
>=Stealth,
thick
]


\node at (-2, 1.5){1};
\node[gvertex] (0v0) at (0,0) {};
\node[gvertex] (1v0) at (1.5,0) {};
\node (2v0) at (3,0) {};
\node (3v0) at (4.5,0) {};
\node[gvertex] (4v0) at (6,0) {};
\node[gvertex] (5v0) at (7.5,0) {};
\node[gvertex] (6v0) at (9,0) {};

\node[fvertex] (0v1) at (0,1.5) {};
\node[fvertex] (1v1) at (1.5,1.5) {};
\node (2v1) at (3,1.5) {};
\node (3v1) at (4.5,1.5) {};
\node[fvertex] (4v1) at (6,1.5) {};
\node[vertex] (5v1) at (7.5,1.5) {};
\node[gvertex] (6v1) at (9,1.5) {};

\node[vertex] (0v2) at (0,3) {};
\node[vertex] (1v2) at (1.5,3) {};
\node (2v2) at (3,3) {};
\node (3v2) at (4.5,3) {};
\node[vertex] (4v2) at (6,3) {};
\node[vertex] (5v2) at (7.5,3) {};
\node[gvertex] (6v2) at (9,3) {};


\draw[gray] (-1.5, -0.5) -- (2.25, -0.5);
\draw[gray] (-1.5, 3.5) -- (2.25, 3.5);

\draw[gray, dotted](2.25, -0.5)--(5.25, -0.5);
\draw[ultra thick, dotted](2.75, 2.25)--(4.75, 2.25);
\draw[gray, dotted](2.25, 3.5)--(5.25, 3.5);

\draw[gray] (5.25, -0.5) -- (10.25, -0.5);
\draw[gray] (5.25, 3.5) -- (10.25, 3.5);

\draw[gray] (-0.75,-0.5)--(-0.75,3.5);
\draw[gray] (0.75,-0.5)--(0.75,3.5);
\draw[gray, dotted] (2.25,-0.5)--(2.25,3.5);
\draw[gray, dotted] (5.25,-0.5)--(5.25,3.5);
\draw[gray] (6.75,-0.5)--(6.75,3.5);
\draw[gray] (8.25,-0.5)--(8.25,3.5);
\draw[gray] (9.75,-0.5)--(9.75,3.5);


\draw[<-,dashed] (-1.2,0)--(0v0);
\draw[->,dashed] (0v0)--(1v0);
\draw[<-,dashed] (1v0)--(2v0);

\draw[<-,dashed] (3v0)--(4v0);
\draw[->,dashed] (4v0)--(5v0);
\draw[<-,dashed] (5v0)--(6v0);
\draw[->,dashed] (6v0)--(10.2,0);

\draw[->,dashed] (-1.2,1.5)--(0v1);
\draw[<-,dashed] (0v1)--(1v1);
\draw[->,dashed] (1v1)--(2v1);

\draw[->,dashed] (3v1)--(4v1);
\draw[<-,dashed] (4v1)--(5v1);
\draw[->,dashed] (5v1)--(6v1);
\draw[<-,dashed] (6v1)--(10.2,1.5);

\draw[<-,dashed] (-1.2,3)--(0v2);
\draw[->,dashed] (0v2)--(1v2);
\draw[<-,dashed] (1v2)--(2v2);

\draw[<-,dashed] (3v2)--(4v2);
\draw[->,dashed] (4v2)--(5v2);
\draw[<-,dashed] (5v2)--(6v2);
\draw[->,dashed] (6v2)--(10.2,3);


\draw[->,dashed] (0,-1.2)--(0v0);
\draw[<-,dashed] (0v0)--(0v1);
\draw[->,dashed] (0v1)--(0v2);
\draw[<-,dashed] (0v2)--(0,4.2) node[at end, above] {$OC$} ;

\draw[<-,dashed] (1.5,-1.2)--(1v0);
\draw[->,dashed] (1v0)--(1v1);
\draw[<-,dashed] (1v1)--(1v2);
\draw[->,dashed] (1v2)--(1.5,4.2);

\draw[->,dashed] (6,-1.2)--(4v0);
\draw[<-,dashed] (4v0)--(4v1);
\draw[->,dashed] (4v1)--(4v2);
\draw[<-,dashed] (4v2)--(6,4.2);

\draw[<-,dashed] (7.5,-1.2)--(5v0);
\draw[->,dashed] (5v0)--(5v1);
\draw[<-,dashed] (5v1)--(5v2);
\draw[->,dashed] (5v2)--(7.5,4.2)node[at end, above] {$RC$};

\draw[->,dashed] (9,-1.2)--(6v0);
\draw[<-,dashed] (6v0)--(6v1);
\draw[->,dashed] (6v1)--(6v2);
\draw[<-,dashed] (6v2)--(9,4.2);

\draw[ultra thick, rounded corners=8pt, ->] 
    (-1.3, 0.2+1.5) -- (-0.2, 0.2+1.5) -- (-0.2, 1.3+1.5) -- 
    (1.3, 1.3+1.5) -- (1.3,0.2+1.5) -- (2.8,0.2+1.5) ;
    
\draw[ultra thick, rounded corners=8pt, ->] 
    (-1.3 + 6, 0.2+1.5) -- (-0.2 +6, 0.2+1.5) -- (-0.2+6, 1.3+1.5) -- 
    (1.3+6, 1.3+1.5) -- (1.3+6,0.2+1.5) -- ( 6+0.2, 0.2+1.5) ;    

\end{tikzpicture}

\vspace{0.5em}

\begin{tikzpicture}[
scale=0.6,
vertex/.style={circle, draw, minimum size=2.5mm, inner sep=0pt},
fvertex/.style={circle, draw, fill=black, minimum size=2.5mm, inner sep=0pt},
gvertex/.style={circle, draw, fill=gray!60, minimum size=2.5mm, inner sep=0pt},
>=Stealth,
thick
]

\node at (-2, 1.5){2};
\node[gvertex] (0v0) at (0,0) {};
\node[gvertex] (1v0) at (1.5,0) {};
\node (2v0) at (3,0) {};
\node (3v0) at (4.5,0) {};
\node[gvertex] (4v0) at (6,0) {};
\node[gvertex] (5v0) at (7.5,0) {};
\node[gvertex] (6v0) at (9,0) {};

\node[vertex] (0v1) at (0,1.5) {};
\node[vertex] (1v1) at (1.5,1.5) {};
\node (2v1) at (3,1.5) {};
\node (3v1) at (4.5,1.5) {};
\node[vertex] (4v1) at (6,1.5) {};
\node[fvertex] (5v1) at (7.5,1.5) {};
\node[gvertex] (6v1) at (9,1.5) {};

\node[fvertex] (0v2) at (0,3) {};
\node[fvertex] (1v2) at (1.5,3) {};
\node (2v2) at (3,3) {};
\node (3v2) at (4.5,3) {};
\node[fvertex] (4v2) at (6,3) {};
\node[fvertex] (5v2) at (7.5,3) {};
\node[gvertex] (6v2) at (9,3) {};


\draw[gray] (-1.5, -0.5) -- (2.25, -0.5);
\draw[gray] (-1.5, 3.5) -- (2.25, 3.5);

\draw[gray, dotted](2.25, -0.5)--(5.25, -0.5);
\draw[ultra thick, dotted](2.75, 2.25)--(4.75, 2.25);
\draw[gray, dotted](2.25, 3.5)--(5.25, 3.5);

\draw[gray] (5.25, -0.5) -- (10.25, -0.5);
\draw[gray] (5.25, 3.5) -- (10.25, 3.5);

\draw[gray] (-0.75,-0.5)--(-0.75,3.5);
\draw[gray] (0.75,-0.5)--(0.75,3.5);
\draw[gray, dotted] (2.25,-0.5)--(2.25,3.5);
\draw[gray, dotted] (5.25,-0.5)--(5.25,3.5);
\draw[gray] (6.75,-0.5)--(6.75,3.5);
\draw[gray] (8.25,-0.5)--(8.25,3.5);
\draw[gray] (9.75,-0.5)--(9.75,3.5);


\draw[<-,dashed] (-1.2,0)--(0v0);
\draw[->,dashed] (0v0)--(1v0);
\draw[<-,dashed] (1v0)--(2v0);
\draw[<-,dashed] (3v0)--(4v0);
\draw[->,dashed] (4v0)--(5v0);
\draw[<-,dashed] (5v0)--(6v0);
\draw[->,dashed] (6v0)--(10.2,0);

\draw[->,dashed] (-1.2,1.5)--(0v1);
\draw[<-,dashed] (0v1)--(1v1);
\draw[->,dashed] (1v1)--(2v1);
\draw[->,dashed] (3v1)--(4v1);
\draw[<-,dashed] (4v1)--(5v1);
\draw[->,dashed] (5v1)--(6v1);
\draw[<-,dashed] (6v1)--(10.2,1.5);

\draw[<-,dashed] (-1.2,3)--(0v2);
\draw[->,dashed] (0v2)--(1v2);
\draw[<-,dashed] (1v2)--(2v2);
\draw[<-,dashed] (3v2)--(4v2);
\draw[->,dashed] (4v2)--(5v2);
\draw[<-,dashed] (5v2)--(6v2);
\draw[->,dashed] (6v2)--(10.2,3);


\draw[->,dashed] (0,-1.2)--(0v0);
\draw[<-,dashed] (0v0)--(0v1);
\draw[->,dashed] (0v1)--(0v2);
\draw[<-,dashed] (0v2)--(0,4.2);

\draw[<-,dashed] (1.5,-1.2)--(1v0);
\draw[->,dashed] (1v0)--(1v1);
\draw[<-,dashed] (1v1)--(1v2);
\draw[->,dashed] (1v2)--(1.5,4.2);

\draw[->,dashed] (6,-1.2)--(4v0);
\draw[<-,dashed] (4v0)--(4v1);
\draw[->,dashed] (4v1)--(4v2);
\draw[<-,dashed] (4v2)--(6,4.2);

\draw[<-,dashed] (7.5,-1.2)--(5v0);
\draw[->,dashed] (5v0)--(5v1);
\draw[<-,dashed] (5v1)--(5v2);
\draw[->,dashed] (5v2)--(7.5,4.2);

\draw[->,dashed] (9,-1.2)--(6v0);
\draw[<-,dashed] (6v0)--(6v1);
\draw[->,dashed] (6v1)--(6v2);
\draw[<-,dashed] (6v2)--(9,4.2);

\draw[ultra thick, rounded corners=8pt, <-] 
    (-1.3, 1.3+1.5) -- (-0.2, 1.3+1.5) -- (-0.2, 0.2+1.5) -- 
    (1.3, 0.2+1.5) -- (1.3,1.3+1.5) -- (2.8,1.3+1.5) ;
    
\draw[ultra thick, rounded corners=8pt, <-] 
    (-1.3 + 6, 1.3+1.5) -- (-0.2 +6, 1.3+1.5) -- (-0.2+6, 0.2+1.5) -- 
    (1.3+6, 0.2+1.5) ;    

\end{tikzpicture}

\vspace{0.5em}

\begin{tikzpicture}[
scale=0.6,
vertex/.style={circle, draw, minimum size=2.5mm, inner sep=0pt},
fvertex/.style={circle, draw, fill=black, minimum size=2.5mm, inner sep=0pt},
gvertex/.style={circle, draw, fill=gray!60, minimum size=2.5mm, inner sep=0pt},
>=Stealth,
thick
]

\node at (-2, 1.5){3};
\node[gvertex] (0v0) at (0,0) {};
\node[gvertex] (1v0) at (1.5,0) {};
\node (2v0) at (3,0) {};
\node (3v0) at (4.5,0) {};
\node[gvertex] (4v0) at (6,0) {};
\node[gvertex] (5v0) at (7.5,0) {};
\node[gvertex] (6v0) at (9,0) {};

\node[fvertex] (0v1) at (0,1.5) {};
\node[fvertex] (1v1) at (1.5,1.5) {};
\node (2v1) at (3,1.5) {};
\node (3v1) at (4.5,1.5) {};
\node[fvertex] (4v1) at (6,1.5) {};
\node[fvertex] (5v1) at (7.5,1.5) {};
\node[gvertex] (6v1) at (9,1.5) {};

\node[vertex] (0v2) at (0,3) {};
\node[vertex] (1v2) at (1.5,3) {};
\node (2v2) at (3,3) {};
\node (3v2) at (4.5,3) {};
\node[vertex] (4v2) at (6,3) {};
\node[fvertex] (5v2) at (7.5,3) {};
\node[gvertex] (6v2) at (9,3) {};


\draw[gray] (-1.5, -0.5) -- (2.25, -0.5);
\draw[gray] (-1.5, 3.5) -- (2.25, 3.5);

\draw[gray, dotted](2.25, -0.5)--(5.25, -0.5);
\draw[ultra thick, dotted](2.75, 2.25)--(4.75, 2.25);
\draw[gray, dotted](2.25, 3.5)--(5.25, 3.5);

\draw[gray] (5.25, -0.5) -- (10.25, -0.5);
\draw[gray] (5.25, 3.5) -- (10.25, 3.5);

\draw[gray] (-0.75,-0.5)--(-0.75,3.5);
\draw[gray] (0.75,-0.5)--(0.75,3.5);
\draw[gray, dotted] (2.25,-0.5)--(2.25,3.5);
\draw[gray, dotted] (5.25,-0.5)--(5.25,3.5);
\draw[gray] (6.75,-0.5)--(6.75,3.5);
\draw[gray] (8.25,-0.5)--(8.25,3.5);
\draw[gray] (9.75,-0.5)--(9.75,3.5);


\draw[<-,dashed] (-1.2,0)--(0v0);
\draw[->,dashed] (0v0)--(1v0);
\draw[<-,dashed] (1v0)--(2v0);

\draw[<-,dashed] (3v0)--(4v0);
\draw[->,dashed] (4v0)--(5v0);
\draw[<-,dashed] (5v0)--(6v0);
\draw[->,dashed] (6v0)--(10.2,0);

\draw[->,dashed] (-1.2,1.5)--(0v1);
\draw[<-,dashed] (0v1)--(1v1);
\draw[->,dashed] (1v1)--(2v1);

\draw[->,dashed] (3v1)--(4v1);
\draw[<-,dashed] (4v1)--(5v1);
\draw[->,dashed] (5v1)--(6v1);
\draw[<-,dashed] (6v1)--(10.2,1.5);

\draw[<-,dashed] (-1.2,3)--(0v2);
\draw[->,dashed] (0v2)--(1v2);
\draw[<-,dashed] (1v2)--(2v2);

\draw[<-,dashed] (3v2)--(4v2);
\draw[->,dashed] (4v2)--(5v2);
\draw[<-,dashed] (5v2)--(6v2);
\draw[->,dashed] (6v2)--(10.2,3);


\draw[->,dashed] (0,-1.2)--(0v0);
\draw[<-,dashed] (0v0)--(0v1);
\draw[->,dashed] (0v1)--(0v2);
\draw[<-,dashed] (0v2)--(0,4.2);

\draw[<-,dashed] (1.5,-1.2)--(1v0);
\draw[->,dashed] (1v0)--(1v1);
\draw[<-,dashed] (1v1)--(1v2);
\draw[->,dashed] (1v2)--(1.5,4.2);

\draw[->,dashed] (6,-1.2)--(4v0);
\draw[<-,dashed] (4v0)--(4v1);
\draw[->,dashed] (4v1)--(4v2);
\draw[<-,dashed] (4v2)--(6,4.2);

\draw[<-,dashed] (7.5,-1.2)--(5v0);
\draw[->,dashed] (5v0)--(5v1);
\draw[<-,dashed] (5v1)--(5v2);
\draw[->,dashed] (5v2)--(7.5,4.2);

\draw[->,dashed] (9,-1.2)--(6v0);
\draw[<-,dashed] (6v0)--(6v1);
\draw[->,dashed] (6v1)--(6v2);
\draw[<-,dashed] (6v2)--(9,4.2);

\draw[ultra thick, rounded corners=8pt, ->] 
    (-1.3, 0.2+1.5) -- (-0.2, 0.2+1.5) -- (-0.2, 1.3+1.5) -- 
    (1.3, 1.3+1.5) -- (1.3,0.2+1.5) -- (2.8,0.2+1.5) ;
    
\draw[ultra thick, rounded corners=8pt, ->] 
    (-1.3 + 6, 0.2+1.5) -- (-0.2 +6, 0.2+1.5) -- (-0.2+6, 1.3+1.5) -- 
    (1.3+6, 1.3+1.5) -- (1.3+6,4.2) ;    

\end{tikzpicture}

\caption{Case A.A:  The forward motion is performed exclusively using cells of the top and middle rows and $RC$ is horizontal.}
\label{fig:A.A}
\end{figure}
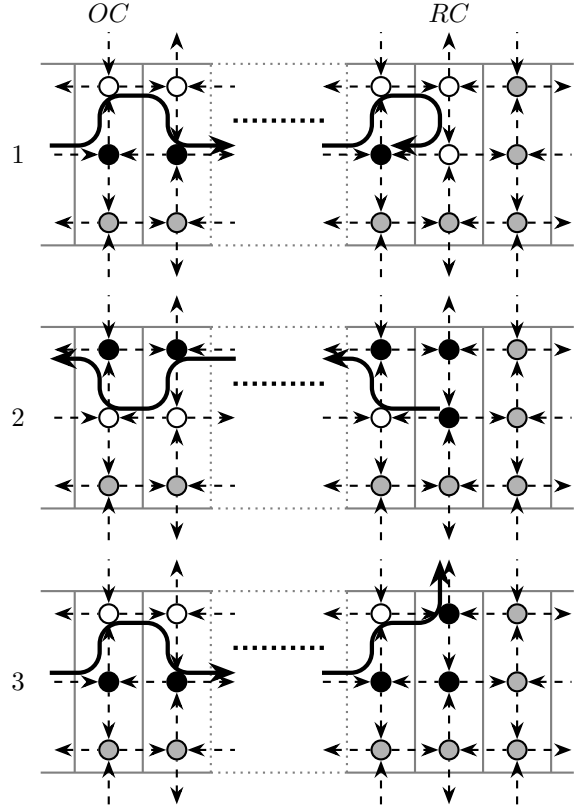

\textbf{Case A.B.} If $RC$ is a vertical column, then there exists at least one horizontal intermediate column between $OC$ and $RC$, since $OC$ is also a vertical column. Let $LHC$ denote the horizontal column immediately to the left of $RC$.  
When the ant return to $RC$, it can only perform either \textbf{(A.B.A)} $B_{lb}$ or  \textbf{(A.B.B)} $C_{lb}$. Figure~\ref{fig:A.B} illustrates the arguments that follow.

\textbf{Case A.B.A.} If it performs $B_{lb}$ on $RC$, it immediately returns to $LHC$ through $rb_{in}$, thereby turning $LHC$ into a blocking column by Lemma~\ref{pr:Hblocking}, since the last motion performed on $LHC$ was a left-to-right crossing using cells of the top and the middle rows. 

\textbf{Case A.B.B.} If it performs $C_{lb}$ on $RC$, then upon to a return from the right, the ant performs $C_{rb}$ on $RC$ and enters $LHC$ through $rb_{in}$, again turning $LHC$ into a blocking column by Lemma~\ref{pr:Hblocking}.

After the first execution of $C_{lt}$ on $OC$, in \textbf{Case A}, the ant performs at most four additional steps on $OC$ consisting of $C_{rt}$ followed by another $C_{lt}$.

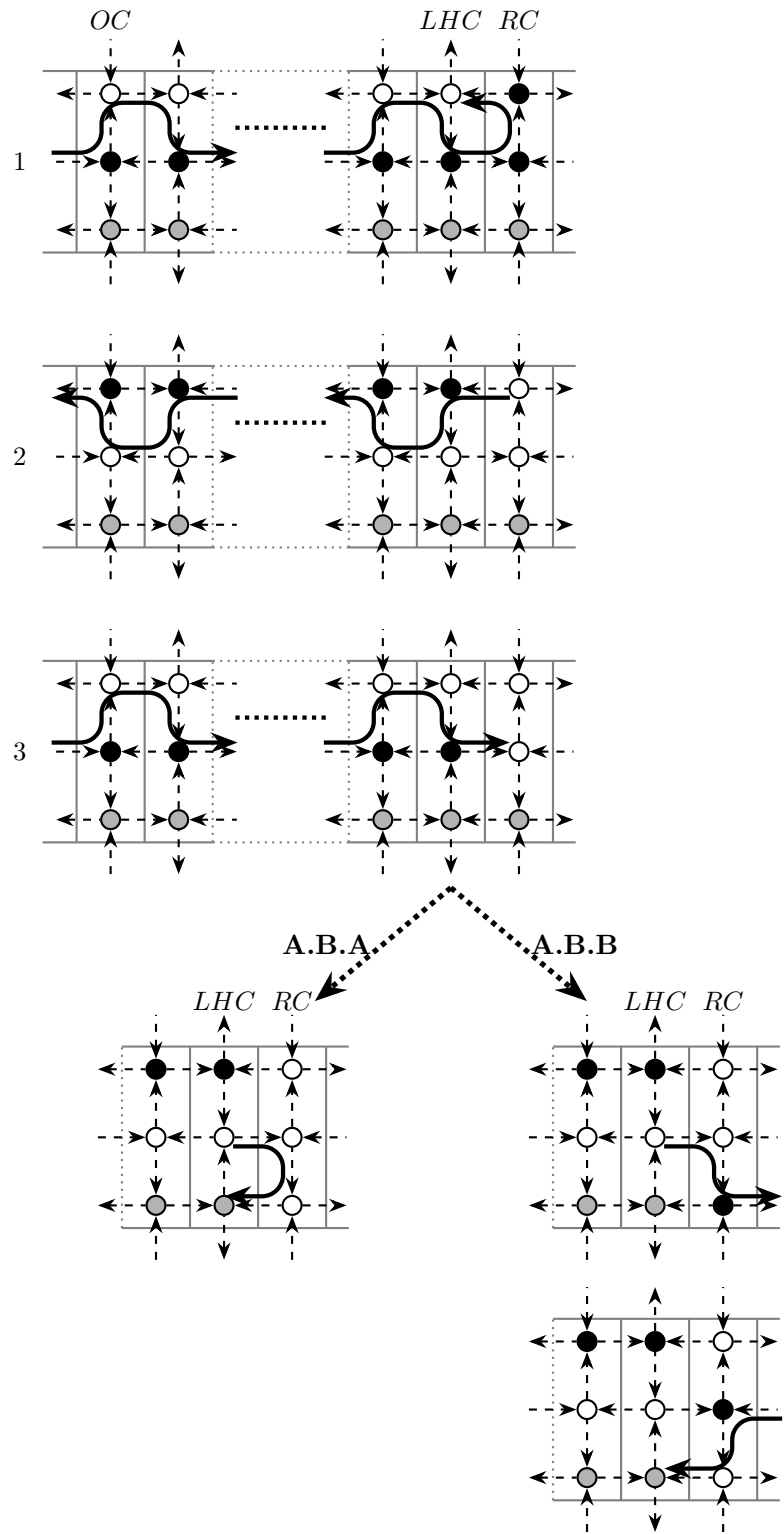
\begin{figure}[htbp]
\centering

\begin{tikzpicture}[
scale=0.6,
vertex/.style={circle, draw, minimum size=2.5mm, inner sep=0pt},
fvertex/.style={circle, draw, fill=black, minimum size=2.5mm, inner sep=0pt},
gvertex/.style={circle, draw, fill=gray!60, minimum size=2.5mm, inner sep=0pt},
>=Stealth,
thick
]

\def\drawgrid{
\draw[gray] (-1.5, -0.5) -- (2.25, -0.5);
\draw[gray] (-1.5, 3.5) -- (2.25, 3.5);
\draw[gray, dotted](2.25, -0.5)--(5.25, -0.5);
\draw[ultra thick, dotted](2.75, 2.25)--(4.75, 2.25);
\draw[gray, dotted](2.25, 3.5)--(5.25, 3.5);
\draw[gray] (5.25, -0.5) -- (10.25, -0.5);
\draw[gray] (5.25, 3.5) -- (10.25, 3.5);
\draw[gray] (-0.75,-0.5)--(-0.75,3.5);
\draw[gray] (0.75,-0.5)--(0.75,3.5);
\draw[gray, dotted] (2.25,-0.5)--(2.25,3.5);
\draw[gray, dotted] (5.25,-0.5)--(5.25,3.5);
\draw[gray] (6.75,-0.5)--(6.75,3.5);
\draw[gray] (8.25,-0.5)--(8.25,3.5);
\draw[gray] (9.75,-0.5)--(9.75,3.5);
}

\def\drawgridpartial{

\draw[gray] (5.25, -0.5) -- (10.25, -0.5);
\draw[gray] (5.25, 3.5) -- (10.25, 3.5);

\draw[gray, dotted] (5.25,-0.5)--(5.25,3.5);
\draw[gray] (6.75,-0.5)--(6.75,3.5);
\draw[gray] (8.25,-0.5)--(8.25,3.5);
\draw[gray] (9.75,-0.5)--(9.75,3.5);
}

\def\drawhorizonzarrows{
\draw[<-,dashed] (-1.2,0)--(0v0);
\draw[->,dashed] (0v0)--(1v0);
\draw[<-,dashed] (1v0)--(2v0);
\draw[<-,dashed] (3v0)--(4v0);
\draw[->,dashed] (4v0)--(5v0);
\draw[<-,dashed] (5v0)--(6v0);
\draw[->,dashed] (6v0)--(10.2,0);

\draw[->,dashed] (-1.2,1.5)--(0v1);
\draw[<-,dashed] (0v1)--(1v1);
\draw[->,dashed] (1v1)--(2v1);
\draw[->,dashed] (3v1)--(4v1);
\draw[<-,dashed] (4v1)--(5v1);
\draw[->,dashed] (5v1)--(6v1);
\draw[<-,dashed] (6v1)--(10.2,1.5);

\draw[<-,dashed] (-1.2,3)--(0v2);
\draw[->,dashed] (0v2)--(1v2);
\draw[<-,dashed] (1v2)--(2v2);
\draw[<-,dashed] (3v2)--(4v2);
\draw[->,dashed] (4v2)--(5v2);
\draw[<-,dashed] (5v2)--(6v2);
\draw[->,dashed] (6v2)--(10.2,3);
}

\def\drawhorizonzarrowspartial{

\draw[<-,dashed] (3v0)--(4v0);
\draw[->,dashed] (4v0)--(5v0);
\draw[<-,dashed] (5v0)--(6v0);
\draw[->,dashed] (6v0)--(10.2,0);

\draw[->,dashed] (3v1)--(4v1);
\draw[<-,dashed] (4v1)--(5v1);
\draw[->,dashed] (5v1)--(6v1);
\draw[<-,dashed] (6v1)--(10.2,1.5);

\draw[<-,dashed] (3v2)--(4v2);
\draw[->,dashed] (4v2)--(5v2);
\draw[<-,dashed] (5v2)--(6v2);
\draw[->,dashed] (6v2)--(10.2,3);
}

\def\drawvertarrows{
\draw[->,dashed] (0,-1.2)--(0v0);
\draw[<-,dashed] (0v0)--(0v1);
\draw[->,dashed] (0v1)--(0v2);
\draw[<-,dashed] (0v2)--(0,4.2);

\draw[<-,dashed] (1.5,-1.2)--(1v0);
\draw[->,dashed] (1v0)--(1v1);
\draw[<-,dashed] (1v1)--(1v2);
\draw[->,dashed] (1v2)--(1.5,4.2);

\draw[->,dashed] (6,-1.2)--(4v0);
\draw[<-,dashed] (4v0)--(4v1);
\draw[->,dashed] (4v1)--(4v2);
\draw[<-,dashed] (4v2)--(6,4.2);

\draw[<-,dashed] (7.5,-1.2)--(5v0);
\draw[->,dashed] (5v0)--(5v1);
\draw[<-,dashed] (5v1)--(5v2);
\draw[->,dashed] (5v2)--(7.5,4.2);

\draw[->,dashed] (9,-1.2)--(6v0);
\draw[<-,dashed] (6v0)--(6v1);
\draw[->,dashed] (6v1)--(6v2);
\draw[<-,dashed] (6v2)--(9,4.2);
}

\def\drawvertarrowspartial{
\draw[->,dashed] (6,-1.2)--(4v0);
\draw[<-,dashed] (4v0)--(4v1);
\draw[->,dashed] (4v1)--(4v2);
\draw[<-,dashed] (4v2)--(6,4.2);

\draw[<-,dashed] (7.5,-1.2)--(5v0);
\draw[->,dashed] (5v0)--(5v1);
\draw[<-,dashed] (5v1)--(5v2);
\draw[->,dashed] (5v2)--(7.5,4.2);

\draw[->,dashed] (9,-1.2)--(6v0);
\draw[<-,dashed] (6v0)--(6v1);
\draw[->,dashed] (6v1)--(6v2);
\draw[<-,dashed] (6v2)--(9,4.2);
}

\begin{scope}[shift={(0, 24)}]
\node at (-2, 1.5) {1};
\node[gvertex] (0v0) at (0,0) {};
\node[gvertex] (1v0) at (1.5,0) {};
\node (2v0) at (3,0) {};
\node (3v0) at (4.5,0) {};
\node[gvertex] (4v0) at (6,0) {};
\node[gvertex] (5v0) at (7.5,0) {};
\node[gvertex] (6v0) at (9,0) {};
\node[fvertex] (0v1) at (0,1.5) {};
\node[fvertex] (1v1) at (1.5,1.5) {};
\node (2v1) at (3,1.5) {};
\node (3v1) at (4.5,1.5) {};
\node[fvertex] (4v1) at (6,1.5) {};
\node[fvertex] (5v1) at (7.5,1.5) {};
\node[fvertex] (6v1) at (9,1.5) {};
\node[vertex] (0v2) at (0,3) {};
\node[vertex] (1v2) at (1.5,3) {};
\node (2v2) at (3,3) {};
\node (3v2) at (4.5,3) {};
\node[vertex] (4v2) at (6,3) {};
\node[vertex] (5v2) at (7.5,3) {};
\node[fvertex] (6v2) at (9,3) {};
\drawgrid
\drawhorizonzarrows

\draw[->,dashed] (0,-1.2)--(0v0);
\draw[<-,dashed] (0v0)--(0v1);
\draw[->,dashed] (0v1)--(0v2);
\draw[<-,dashed] (0v2)--(0,4.2) node [at end, above]{$OC$};

\draw[<-,dashed] (1.5,-1.2)--(1v0);
\draw[->,dashed] (1v0)--(1v1);
\draw[<-,dashed] (1v1)--(1v2);
\draw[->,dashed] (1v2)--(1.5,4.2);

\draw[->,dashed] (6,-1.2)--(4v0);
\draw[<-,dashed] (4v0)--(4v1);
\draw[->,dashed] (4v1)--(4v2);
\draw[<-,dashed] (4v2)--(6,4.2);

\draw[<-,dashed] (7.5,-1.2)--(5v0);
\draw[->,dashed] (5v0)--(5v1);
\draw[<-,dashed] (5v1)--(5v2);
\draw[->,dashed] (5v2)--(7.5,4.2) node [at end, above]{$LHC$};

\draw[->,dashed] (9,-1.2)--(6v0);
\draw[<-,dashed] (6v0)--(6v1);
\draw[->,dashed] (6v1)--(6v2);
\draw[<-,dashed] (6v2)--(9,4.2) node [at end, above]{$RC$};

\draw[ultra thick, rounded corners=8pt, ->] 
    (-1.3, 0.2+1.5) -- (-0.2, 0.2+1.5) -- (-0.2, 1.3+1.5) -- 
    (1.3, 1.3+1.5) -- (1.3,0.2+1.5) -- (2.8,0.2+1.5) ;
\draw[ultra thick, rounded corners=8pt, ->] 
    (-1.3 + 6, 0.2+1.5) -- (-0.2 +6, 0.2+1.5) -- (-0.2+6, 1.3+1.5) -- 
    (1.3+6, 1.3+1.5) -- (1.3+6,0.2+1.5) -- ( 2.8+6, 0.2+1.5) --
    ( 2.8+6, 1.3+1.5) -- (1.7+6, 1.3+1.5);
\end{scope}

\begin{scope}[shift={(0, 17.5)}]
\node at (-2, 1.5) {2};
\node[gvertex] (0v0) at (0,0) {};
\node[gvertex] (1v0) at (1.5,0) {};
\node (2v0) at (3,0) {};
\node (3v0) at (4.5,0) {};
\node[gvertex] (4v0) at (6,0) {};
\node[gvertex] (5v0) at (7.5,0) {};
\node[gvertex] (6v0) at (9,0) {};
\node[vertex] (0v1) at (0,1.5) {};
\node[vertex] (1v1) at (1.5,1.5) {};
\node (2v1) at (3,1.5) {};
\node (3v1) at (4.5,1.5) {};
\node[vertex] (4v1) at (6,1.5) {};
\node[vertex] (5v1) at (7.5,1.5) {};
\node[vertex] (6v1) at (9,1.5) {};
\node[fvertex] (0v2) at (0,3) {};
\node[fvertex] (1v2) at (1.5,3) {};
\node (2v2) at (3,3) {};
\node (3v2) at (4.5,3) {};
\node[fvertex] (4v2) at (6,3) {};
\node[fvertex] (5v2) at (7.5,3) {};
\node[vertex] (6v2) at (9,3) {};
\drawgrid
\drawhorizonzarrows
\drawvertarrows
\draw[ultra thick, rounded corners=8pt, <-] 
    (-1.3, 1.3+1.5) -- (-0.2, 1.3+1.5) -- (-0.2, 0.2+1.5) -- 
    (1.3, 0.2+1.5) -- (1.3,1.3+1.5) -- (2.8,1.3+1.5) ;
\draw[ultra thick, rounded corners=8pt, <-] 
    (-1.3 + 6, 1.3+1.5) -- (-0.2 +6, 1.3+1.5) -- (-0.2+6, 0.2+1.5) -- 
    (1.3+6, 0.2+1.5) -- (1.3+6, 2.8) -- ( 2.8+6, 1.3+1.5) ;
\end{scope}

\begin{scope}[shift={(0, 11)}]
\node at (-2, 1.5) {3};
\node[gvertex] (0v0) at (0,0) {};
\node[gvertex] (1v0) at (1.5,0) {};
\node (2v0) at (3,0) {};
\node (3v0) at (4.5,0) {};
\node[gvertex] (4v0) at (6,0) {};
\node[gvertex] (5v0) at (7.5,0) {};
\node[gvertex] (6v0) at (9,0) {};
\node[fvertex] (0v1) at (0,1.5) {};
\node[fvertex] (1v1) at (1.5,1.5) {};
\node (2v1) at (3,1.5) {};
\node (3v1) at (4.5,1.5) {};
\node[fvertex] (4v1) at (6,1.5) {};
\node[fvertex] (5v1) at (7.5,1.5) {};
\node[vertex] (6v1) at (9,1.5) {};
\node[vertex] (0v2) at (0,3) {};
\node[vertex] (1v2) at (1.5,3) {};
\node (2v2) at (3,3) {};
\node (3v2) at (4.5,3) {};
\node[vertex] (4v2) at (6,3) {};
\node[vertex] (5v2) at (7.5,3) {};
\node[vertex] (6v2) at (9,3) {};
\drawgrid
\drawhorizonzarrows
\drawvertarrows
\draw[ultra thick, rounded corners=8pt, ->] 
    (-1.3, 0.2+1.5) -- (-0.2, 0.2+1.5) -- (-0.2, 1.3+1.5) -- 
    (1.3, 1.3+1.5) -- (1.3,0.2+1.5) -- (2.8,0.2+1.5) ;
\draw[ultra thick, rounded corners=8pt, ->] 
    (-1.3 + 6, 0.2+1.5) -- (-0.2 +6, 0.2+1.5) -- (-0.2+6, 1.3+1.5) -- 
    (1.3+6, 1.3+1.5) -- (1.3+6,0.2+1.5) -- ( 2.8+6, 0.2+1.5);
\end{scope}

\begin{scope}[shift={(-5, 2.5)}]

\node (3v0) at (4.5,0) {};
\node[gvertex] (4v0) at (6,0) {};
\node[gvertex] (5v0) at (7.5,0) {};
\node[vertex] (6v0) at (9,0) {};

\node (3v1) at (4.5,1.5) {};
\node[vertex] (4v1) at (6,1.5) {};
\node[vertex] (5v1) at (7.5,1.5) {};
\node[vertex] (6v1) at (9,1.5) {};

\node (3v2) at (4.5,3) {};
\node[fvertex] (4v2) at (6,3) {};
\node[fvertex] (5v2) at (7.5,3) {};
\node[vertex] (6v2) at (9,3) {};
\node at (7.5,4.5){$LHC$};
\node at (9,4.5){$RC$};

\drawgridpartial
\drawhorizonzarrowspartial
\drawvertarrowspartial

\draw[ultra thick, rounded corners=8pt, ->] 
(1.7+6,1.3) -- ( 2.8+6, 1.3)--
    ( 2.8+6, 0.2) -- ( 1.5+6, 0.2) ;
\end{scope}

\begin{scope}[shift={(4.5, 2.5)}]

\node (3v0) at (4.5,0) {};
\node[gvertex] (4v0) at (6,0) {};
\node[gvertex] (5v0) at (7.5,0) {};
\node[fvertex] (6v0) at (9,0) {};

\node (3v1) at (4.5,1.5) {};
\node[vertex] (4v1) at (6,1.5) {};
\node[vertex] (5v1) at (7.5,1.5) {};
\node[vertex] (6v1) at (9,1.5) {};

\node (3v2) at (4.5,3) {};
\node[fvertex] (4v2) at (6,3) {};
\node[fvertex] (5v2) at (7.5,3){};
\node[vertex] (6v2) at (9,3){};

\node at (7.5,4.5){$LHC$};
\node at (9,4.5){$RC$};
\drawgridpartial
\drawhorizonzarrowspartial
\drawvertarrowspartial

\draw[ultra thick, rounded corners=8pt, ->] (1.7+6 ,1.3) -- ( 2.8+6, 1.3)--
    ( 2.8+6, 0.2) -- ( 2.8+1.5+6, 0.2) ;
\end{scope}

\begin{scope}[shift={(4.5, -3.5)}]

\node (3v0) at (4.5,0) {};
\node[gvertex] (4v0) at (6,0) {};
\node[gvertex] (5v0) at (7.5,0) {};
\node[vertex] (6v0) at (9,0) {};

\node (3v1) at (4.5,1.5) {};
\node[vertex] (4v1) at (6,1.5) {};
\node[vertex] (5v1) at (7.5,1.5) {};
\node[fvertex] (6v1) at (9,1.5) {};

\node (3v2) at (4.5,3) {};
\node[fvertex] (4v2) at (6,3) {};
\node[fvertex] (5v2) at (7.5,3) {};
\node[vertex] (6v2) at (9,3) {};

\drawgridpartial
\drawhorizonzarrowspartial
\drawvertarrowspartial
\draw[ultra thick, rounded corners=8pt, -> ] 
    ( 2.8+1.5+6, 1.3) -- ( 3.2+6, 1.3) --( 3.2+6, 0.2) -- (1.7+6 ,0.2) ;
\end{scope}

\draw[line width=2pt, ->, dotted] (7.5, 9.5) -- (4.5, 7) 
node[midway, left]  {\textbf{\large A.B.A}};

\draw[line width=2pt, ->, dotted] (7.5, 9.5) -- (10.5, 7)
node[midway, right] {\textbf{\large A.B.B}};

\end{tikzpicture}

\caption{Case A.B:  The forward motion is performed exclusively using cells of the top and middle rows and $RC$ is vertical.}

\label{fig:A.B}
\end{figure}

\textbf{Case B.}
In the case where the forward motion involves all three rows, let $LHC$ denote the last column in which the ant motion uses a cell belonging to a row different from those used to bounce on $RC$ during this forward motion. Figure~\ref{fig:B} illustrates the arguments that follow.

The column $LHC$ must necessarily be horizontal.
Indeed, for the ant to cross a vertical column, it must use either the cells of the middle and top rows or those of the middle and bottom rows.
If the traversal is performed using the middle and top cells, then in order to cross or bounce in the next column (which is horizontal), the ant must again use the top and middle cells. Similarly, if the crossing is performed using the middle and bottom cells, then in order to cross or bounce in the next column, the ant must use the bottom and middle cells.
Therefore, the motion performed by the ant on these two consecutive columns necessarily involves cells from the same two rows. Consequently, $LHC$ cannot be vertical and must be horizontal.

Let $LVC$ denote the column immediately reached by the ant after $LHC$ during the forward motion toward $RC$. The motion from $LVC$ to $RC$ is performed using exactly two rows, namely either the top and middle rows or the middle and bottom rows.
Assume that these two rows are the top and middle rows. Then the ant must have traversed $LHC$ using the bottom and middle rows and have performed $C_{lt}$ on $LVC$ during the forward motion. After bouncing on $RC$, the ant returns to $LVC$, performs $C_{rt}$, and subsequently enters $LHC$ through $rt_{in}$, thereby turning $LHC$ into a blocking column by Lemma \ref{pr:Hblocking}.
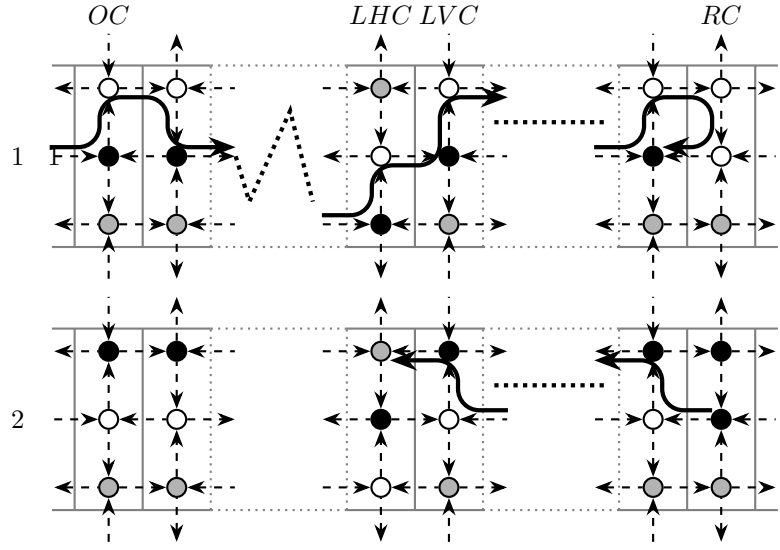
\begin{figure}[htbp]
\centering
\begin{tikzpicture}[
scale=0.6,
vertex/.style={circle, draw, minimum size=2.5mm, inner sep=0pt},
fvertex/.style={circle, draw, fill=black, minimum size=2.5mm, inner sep=0pt},
gvertex/.style={circle, draw, fill=gray!60, minimum size=2.5mm, inner sep=0pt},
>=Stealth,
thick
]
\node at (-2, 1.5){1};
\node[gvertex] (0v0) at (0,0) {};
\node[gvertex] (1v0) at (1.5,0) {};
\node (2v0) at (3,0) {};
\node (3v0) at (4.5,0) {};
\node[fvertex] (4v0) at (6,0) {};
\node[gvertex] (5v0) at (7.5,0) {};
\node (6v0) at (9,0) {};
\node (7v0) at (10.5,0) {};
\node[gvertex] (8v0) at (12,0) {};
\node[gvertex] (9v0) at (13.5,0) {};

\node[fvertex] (0v1) at (0,1.5) {};
\node[fvertex] (1v1) at (1.5,1.5) {};
\node (2v1) at (3,1.5) {};
\node (3v1) at (4.5,1.5) {};
\node[vertex] (4v1) at (6,1.5) {};
\node[fvertex] (5v1) at (7.5,1.5) {};
\node (6v1) at (9,1.5) {};
\node (7v1) at (10.5,1.5) {};
\node[fvertex] (8v1) at (12,1.5) {};
\node[vertex] (9v1) at (13.5,1.5) {};

\node[vertex] (0v2) at (0,3) {};
\node[vertex] (1v2) at (1.5,3) {};
\node (2v2) at (3,3) {};
\node (3v2) at (4.5,3) {};
\node[gvertex] (4v2) at (6,3) {};
\node[vertex] (5v2) at (7.5,3) {};
\node (6v2) at (9,3) {};
\node (7v2) at (10.5,3) {};
\node[vertex] (8v2) at (12,3) {};
\node[vertex] (9v2) at (13.5,3) {};

\draw[gray] (-0.75,-0.5)--(-0.75,3.5);
\draw[gray] (0.75,-0.5)--(0.75,3.5);
\draw[gray, dotted] (2.25,-0.5)--(2.25,3.5);


\draw[gray, dotted] (5.25,-0.5)--(5.25,3.5);
\draw[gray] (6.75,-0.5)--(6.75,3.5);
\draw[gray, dotted] (8.25,-0.5)--(8.25,3.5);


\draw[gray, dotted] (11.25,-0.5)--(11.25,3.5);
\draw[gray] (12.75,-0.5)--(12.75,3.5);
\draw[gray] (14.25,-0.5)--(14.25,3.5);

\draw[gray] (-1.25, -0.5) -- (2.25, -0.5);
\draw[gray] (-1.25, 3.5) -- (2.25, 3.5);

\draw[gray, dotted](2.25, -0.5)--(5.25, -0.5);
\draw[ultra thick, dotted](2.8, 1.5)--(3.1, 0.5)--(4,2.5)-- (4.5, 0.5);
\draw[gray, dotted](2.25, 3.5)--(5.25, 3.5);

\draw[gray] (5.25, -0.5) -- (8.25, -0.5);
\draw[gray] (5.25, 3.5) -- (8.25, 3.5);

\draw[gray, dotted](8.25, -0.5)--(11.25, -0.5);
\draw[ultra thick, dotted](8.5, 2.25)--(11, 2.25);
\draw[gray, dotted](8.25, 3.5)--(11.25, 3.5);

\draw[gray] (11.25, -0.5) -- (14.75, -0.5);
\draw[gray] (11.25, 3.5) -- (14.75, 3.5);

\draw[<-,dashed] (-1.2,0)--(0v0);
\draw[->,dashed] (0v0)--(1v0);
\draw[<-,dashed] (1v0)--(2v0);

\draw[->,dashed] (3v0)--(4v0);
\draw[<-,dashed] (4v0)--(5v0);
\draw[->,dashed] (5v0)--(6v0);

\draw[<-,dashed] (7v0)--(8v0);
\draw[->,dashed] (8v0)--(9v0);
\draw[->,dashed] (9v0)--(14.7,0);

\draw[->,dashed] (-1.2,1.5)--(0v1)node[at start]{1};
\draw[<-,dashed] (0v1)--(1v1);
\draw[->,dashed] (1v1)--(2v1);

\draw[<-,dashed] (3v1)--(4v1);
\draw[->,dashed] (4v1)--(5v1);
\draw[<-,dashed] (5v1)--(6v1);

\draw[->,dashed] (7v1)--(8v1);
\draw[<-,dashed] (8v1)--(9v1);
\draw[<-,dashed] (9v1)--(14.7,1.5);

\draw[<-,dashed] (-1.2,3)--(0v2);
\draw[->,dashed] (0v2)--(1v2);
\draw[<-,dashed] (1v2)--(2v2);

\draw[->,dashed] (3v2)--(4v2);
\draw[<-,dashed] (4v2)--(5v2);
\draw[->,dashed] (5v2)--(6v2);

\draw[<-,dashed] (7v2)--(8v2);
\draw[->,dashed] (8v2)--(9v2);
\draw[->,dashed] (9v2)--(14.7,3);

\draw[->,dashed] (0,-1.2)--(0v0);
\draw[<-,dashed] (0v0)--(0v1);
\draw[->,dashed] (0v1)--(0v2);
\draw[<-,dashed] (0v2)--(0,4.2) node[at end, above] {$OC$} ;

\draw[<-,dashed] (1.5,-1.2)--(1v0);
\draw[->,dashed] (1v0)--(1v1);
\draw[<-,dashed] (1v1)--(1v2);
\draw[->,dashed] (1v2)--(1.5,4.2);

\draw[<-,dashed] (6,-1.2)--(4v0);
\draw[->,dashed] (4v0)--(4v1);
\draw[<-,dashed] (4v1)--(4v2);
\draw[->,dashed] (4v2)--(6,4.2)node[at end, above] {$LHC$};

\draw[->,dashed] (7.5,-1.2)--(5v0);
\draw[<-,dashed] (5v0)--(5v1);
\draw[->,dashed] (5v1)--(5v2);
\draw[<-,dashed] (5v2)--(7.5,4.2) node[at end, above] {$LVC$};

\draw[->,dashed] (12,-1.2)--(8v0);
\draw[<-,dashed] (8v0)--(8v1);
\draw[->,dashed] (8v1)--(8v2);
\draw[<-,dashed] (8v2)--(12,4.2);

\draw[<-,dashed] (13.5,-1.2)--(9v0);
\draw[->,dashed] (9v0)--(9v1);
\draw[<-,dashed] (9v1)--(9v2);
\draw[->,dashed] (9v2)--(13.5,4.2) node[at end, above] {$RC$};

\draw[ultra thick, rounded corners=8pt, ->] 
    (-1.3, 0.2+1.5) -- (-0.2, 0.2+1.5) -- (-0.2, 1.3+1.5) -- 
    (1.3, 1.3+1.5) -- (1.3,0.2+1.5) -- (2.8,0.2+1.5) ;
    
\draw[ultra thick, rounded corners=8pt, ->] 
    (-1.3 + 6, 0.2) -- (-0.2 +6, 0.2) -- (-0.2+6, 1.3) -- 
    (1.3+6, 1.3) -- (1.3+6, 2.8) -- ( 2.8+ 6, 2.8) ;    
    
\draw[ultra thick, rounded corners=8pt, ->] 
    (-1.3 + 12, 0.2+1.5) -- (-0.2 +12, 0.2+1.5) -- (-0.2+12, 1.3+1.5) -- 
    (1.3+12, 1.3+1.5) -- (1.3+12,0.2+1.5) -- ( 12+0.2, 0.2+1.5) ;   

\end{tikzpicture}

\vspace{0.5em}
\begin{tikzpicture}[
scale=0.6,
vertex/.style={circle, draw, minimum size=2.5mm, inner sep=0pt},
fvertex/.style={circle, draw, fill=black, minimum size=2.5mm, inner sep=0pt},
gvertex/.style={circle, draw, fill=gray!60, minimum size=2.5mm, inner sep=0pt},
>=Stealth,
thick
]
\node at (-2, 1.5){2};
\node[gvertex] (0v0) at (0,0) {};
\node[gvertex] (1v0) at (1.5,0) {};
\node (2v0) at (3,0) {};
\node (3v0) at (4.5,0) {};
\node[vertex] (4v0) at (6,0) {};
\node[gvertex] (5v0) at (7.5,0) {};
\node (6v0) at (9,0) {};
\node (7v0) at (10.5,0) {};
\node[gvertex] (8v0) at (12,0) {};
\node[gvertex] (9v0) at (13.5,0) {};

\node[vertex] (0v1) at (0,1.5) {};
\node[vertex] (1v1) at (1.5,1.5) {};
\node (2v1) at (3,1.5) {};
\node (3v1) at (4.5,1.5) {};
\node[fvertex] (4v1) at (6,1.5) {};
\node[vertex] (5v1) at (7.5,1.5) {};
\node (6v1) at (9,1.5) {};
\node (7v1) at (10.5,1.5) {};
\node[vertex] (8v1) at (12,1.5) {};
\node[fvertex] (9v1) at (13.5,1.5) {};

\node[fvertex] (0v2) at (0,3) {};
\node[fvertex] (1v2) at (1.5,3) {};
\node (2v2) at (3,3) {};
\node (3v2) at (4.5,3) {};
\node[gvertex] (4v2) at (6,3) {};
\node[fvertex] (5v2) at (7.5,3) {};
\node (6v2) at (9,3) {};
\node (7v2) at (10.5,3) {};
\node[fvertex] (8v2) at (12,3) {};
\node[fvertex] (9v2) at (13.5,3) {};

\draw[gray] (-0.75,-0.5)--(-0.75,3.5);
\draw[gray] (0.75,-0.5)--(0.75,3.5);
\draw[gray, dotted] (2.25,-0.5)--(2.25,3.5);


\draw[gray, dotted] (5.25,-0.5)--(5.25,3.5);
\draw[gray] (6.75,-0.5)--(6.75,3.5);
\draw[gray, dotted] (8.25,-0.5)--(8.25,3.5);


\draw[gray, dotted] (11.25,-0.5)--(11.25,3.5);
\draw[gray] (12.75,-0.5)--(12.75,3.5);
\draw[gray] (14.25,-0.5)--(14.25,3.5);

\draw[gray] (-1.25, -0.5) -- (2.25, -0.5);
\draw[gray] (-1.25, 3.5) -- (2.25, 3.5);

\draw[gray, dotted](2.25, -0.5)--(5.25, -0.5);
\draw[gray, dotted](2.25, 3.5)--(5.25, 3.5);

\draw[gray] (5.25, -0.5) -- (8.25, -0.5);
\draw[gray] (5.25, 3.5) -- (8.25, 3.5);

\draw[gray, dotted](8.25, -0.5)--(11.25, -0.5);
\draw[ultra thick, dotted](8.5, 2.25)--(11, 2.25);
\draw[gray, dotted](8.25, 3.5)--(11.25, 3.5);

\draw[gray] (11.25, -0.5) -- (14.75, -0.5);
\draw[gray] (11.25, 3.5) -- (14.75, 3.5);

\draw[<-,dashed] (-1.2,0)--(0v0);
\draw[->,dashed] (0v0)--(1v0);
\draw[<-,dashed] (1v0)--(2v0);

\draw[->,dashed] (3v0)--(4v0);
\draw[<-,dashed] (4v0)--(5v0);
\draw[->,dashed] (5v0)--(6v0);

\draw[<-,dashed] (7v0)--(8v0);
\draw[->,dashed] (8v0)--(9v0);
\draw[->,dashed] (9v0)--(14.7,0);

\draw[->,dashed] (-1.2,1.5)--(0v1);
\draw[<-,dashed] (0v1)--(1v1);
\draw[->,dashed] (1v1)--(2v1);

\draw[<-,dashed] (3v1)--(4v1);
\draw[->,dashed] (4v1)--(5v1);
\draw[<-,dashed] (5v1)--(6v1);

\draw[->,dashed] (7v1)--(8v1);
\draw[<-,dashed] (8v1)--(9v1);
\draw[<-,dashed] (9v1)--(14.7,1.5);

\draw[<-,dashed] (-1.2,3)--(0v2);
\draw[->,dashed] (0v2)--(1v2);
\draw[<-,dashed] (1v2)--(2v2);

\draw[->,dashed] (3v2)--(4v2);
\draw[<-,dashed] (4v2)--(5v2);
\draw[->,dashed] (5v2)--(6v2);

\draw[<-,dashed] (7v2)--(8v2);
\draw[->,dashed] (8v2)--(9v2);
\draw[->,dashed] (9v2)--(14.7,3);

\draw[->,dashed] (0,-1.2)--(0v0);
\draw[<-,dashed] (0v0)--(0v1);
\draw[->,dashed] (0v1)--(0v2);
\draw[<-,dashed] (0v2)--(0,4.2) ;

\draw[<-,dashed] (1.5,-1.2)--(1v0);
\draw[->,dashed] (1v0)--(1v1);
\draw[<-,dashed] (1v1)--(1v2);
\draw[->,dashed] (1v2)--(1.5,4.2);

\draw[<-,dashed] (6,-1.2)--(4v0);
\draw[->,dashed] (4v0)--(4v1);
\draw[<-,dashed] (4v1)--(4v2);
\draw[->,dashed] (4v2)--(6,4.2);

\draw[->,dashed] (7.5,-1.2)--(5v0);
\draw[<-,dashed] (5v0)--(5v1);
\draw[->,dashed] (5v1)--(5v2);
\draw[<-,dashed] (5v2)--(7.5,4.2);

\draw[->,dashed] (12,-1.2)--(8v0);
\draw[<-,dashed] (8v0)--(8v1);
\draw[->,dashed] (8v1)--(8v2);
\draw[<-,dashed] (8v2)--(12,4.2);

\draw[<-,dashed] (13.5,-1.2)--(9v0);
\draw[->,dashed] (9v0)--(9v1);
\draw[<-,dashed] (9v1)--(9v2);
\draw[->,dashed] (9v2)--(13.5,4.2);

\draw[ultra thick, rounded corners=8pt, ->] 
    (1.3+7.5 ,0.2+1.5) -- ( 7.5+0.2, 0.2+1.5) -- ( 7.5+0.2, 2.8) -- (-1.3+7.5, 2.8)  ;   
    
\draw[ultra thick, rounded corners=8pt, ->] 
    (1.3+12,0.2+1.5) -- ( 12+0.2, 0.2+1.5) -- ( 12+0.2, 2.8) -- (-1.3+12, 2.8)  ;   

\end{tikzpicture}

\caption{Case B: The forward motion is performed  using cells of three rows. }
\label{fig:B}
\end{figure}

If, on the other hand, the motion from $LVC$ to $RC$ involves the middle and bottom rows, a symmetric argument shows that upon returning from $RC$, the ant re-enters $LHC$ in such a way that $LHC$ again becomes a blocking column.

After the first execution of $C_{lt}$ on $OC$, in \textbf{Case B} also, the ant performs at most four additional steps on $OC$ consisting of $C_{rt}$ followed by another $C_{lt}$.

\end{proof}

\begin{lemma}\label{pr:Hbouncing}
During its motion within a three-row grid, after performing a bouncing motion within a column that is a vertical $G_{3,1}$, the ant can subsequently perform at most eight additional steps on that column before escaping the three-row grid. Moreover, these steps consist of one bouncing motion within the same side of the initial bouncing and three crossing motions, two of which start from the side of the initial bouncing.
\end{lemma}

\begin{proof}
Let $OC$ denote the vertical $G_{3,1}$ under consideration.
Assume the ant performs a bouncing motion within $OC$. Without loss of generality, we may suppose that its a left top bouncing, $B_{lt}$, since the argument presented below applies symmetrically to all other types of bouncing. Figure \ref{fig:3B} illustrates the arguments that follow.

After performing $B_{lt}$, the ant may re-enter $OC$ from the left through $l_{in}$, and it can perform either $C_{lb}$ \textbf{(A)} or $B_{lb}$ \textbf{(B)}   on $OC$. 

\textbf{Case A. } If the ant performs $C_{lb}$, then by Lemma \ref{pr:Hcrossing}, it can subsequently perform at most four additional steps on $OC$. Together with the crossing $C_{lb}$, this yields a total of six additional steps after the first bouncing.

\textbf{Case B. } If the ant instead performs $B_{lb}$, then upon one further return from the left to $OC$ through $l_{in}$ and this time can only perform $C_{lt}$ \textbf{(B.A)}. After this, by Lemma \ref{pr:Hcrossing}, it can perform at most four additional steps on $OC$ before escaping the three-row grid. Hence, in this case, the total number of additional steps after the first bouncing reaches eight.

\begin{figure}[H]
\centering

\begin{tikzpicture}[
scale=0.6,
vertex/.style={circle, draw, minimum size=2.5mm, inner sep=0pt},
fvertex/.style={circle, draw, fill=black, minimum size=2.5mm, inner sep=0pt},
gvertex/.style={circle, draw, fill=gray!60, minimum size=2.5mm, inner sep=0pt},
>=Stealth,
thick
]

\begin{scope}[shift={(0, 0)}]
\def\spacing{1.5}
\def\externalarrow{1.2}
\def\margin{0.5}

\node[gvertex] (start_v0) at (0,0) {};
\node[fvertex] (start_v1) at (0,1.5) {};
\node[fvertex] (start_v2) at (0,3) {};

\draw[gray] (-\spacing/2, -\margin) -- (\spacing/2, -\margin);
\draw[gray] (-\spacing/2, 3+ \margin) -- (\spacing/2, 3 + \margin);
\draw[gray] (-\spacing/2, -\margin) -- (-\spacing/2, 3+ \margin);
\draw[gray] (\spacing/2, -\margin) -- (\spacing/2, 3+ \margin);

\draw[<-, dashed] (-\externalarrow, 0) -- (start_v0);
\draw[->, dashed] (start_v0) -- (\externalarrow, 0);
\draw[->, dashed] (-\externalarrow, 1.5) -- (start_v1);
\draw[<-, dashed] (start_v1) -- (\externalarrow, 1.5);
\draw[<-, dashed] (-\externalarrow, 3) -- (start_v2);
\draw[->, dashed] (start_v2) -- (\externalarrow, 3);

\draw[->, dashed] (0, -\externalarrow) -- (start_v0);
\draw[<-, dashed] (start_v0) -- (start_v1);
\draw[->, dashed] (start_v1) -- (start_v2);
\draw[->, dashed] (0,4.2) -- (start_v2);

\draw[ultra thick, rounded corners=8pt, ->] 
    (-1.3, 1.7) -- (-0.2, 1.7) -- (-0.2, 2.8) -- (-1.3, 2.8);

\end{scope}

\begin{scope}[shift={(6, 2.5)}]
\def\spacing{1.5}
\def\externalarrow{1.2}
\def\margin{0.5}

\node[vertex] (caseB_v0) at (0,0) {};
\node[vertex] (caseB_v1) at (0,1.5) {};
\node[vertex] (caseB_v2) at (0,3) {};

\draw[gray] (-\spacing/2, -\margin) -- (\spacing/2, -\margin);
\draw[gray] (-\spacing/2, 3+ \margin) -- (\spacing/2, 3 + \margin);
\draw[gray] (-\spacing/2, -\margin) -- (-\spacing/2, 3+ \margin);
\draw[gray] (\spacing/2, -\margin) -- (\spacing/2, 3+ \margin);

\draw[<-, dashed] (-\externalarrow, 0) -- (caseB_v0);
\draw[->, dashed] (caseB_v0) -- (\externalarrow, 0);
\draw[->, dashed] (-\externalarrow, 1.5) -- (caseB_v1);
\draw[<-, dashed] (caseB_v1) -- (\externalarrow, 1.5);
\draw[<-, dashed] (-\externalarrow, 3) -- (caseB_v2);
\draw[->, dashed] (caseB_v2) -- (\externalarrow, 3);

\draw[->, dashed] (0, -\externalarrow) -- (caseB_v0);
\draw[<-, dashed] (caseB_v0) -- (caseB_v1);
\draw[->, dashed] (caseB_v1) -- (caseB_v2);
\draw[->, dashed] (0,4.2) -- (caseB_v2);

\draw[ultra thick, rounded corners=8pt, ->] 
    (-1.3, 1.3) -- (-0.2, 1.3) -- (-0.2, 0.2) -- (-1.3, 0.2);

\end{scope}

\begin{scope}[shift={(6, -3)}]
\def\spacing{1.5}
\def\externalarrow{1.2}
\def\margin{0.5}

\node[fvertex] (caseA_v0) at (0,0) {};
\node[vertex] (caseA_v1) at (0,1.5) {};
\node[vertex] (caseA_v2) at (0,3) {};

\draw[gray] (-\spacing/2, -\margin) -- (\spacing/2, -\margin);
\draw[gray] (-\spacing/2, 3+ \margin) -- (\spacing/2, 3 + \margin);
\draw[gray] (-\spacing/2, -\margin) -- (-\spacing/2, 3+ \margin);
\draw[gray] (\spacing/2, -\margin) -- (\spacing/2, 3+ \margin);

\draw[<-, dashed] (-\externalarrow, 0) -- (caseA_v0);
\draw[->, dashed] (caseA_v0) -- (\externalarrow, 0);
\draw[->, dashed] (-\externalarrow, 1.5) -- (caseA_v1);
\draw[<-, dashed] (caseA_v1) -- (\externalarrow, 1.5);
\draw[<-, dashed] (-\externalarrow, 3) -- (caseA_v2);
\draw[->, dashed] (caseA_v2) -- (\externalarrow, 3);

\draw[->, dashed] (0, -\externalarrow) -- (caseA_v0);
\draw[<-, dashed] (caseA_v0) -- (caseA_v1);
\draw[->, dashed] (caseA_v1) -- (caseA_v2);
\draw[->, dashed] (0,4.2) -- (caseA_v2);

\draw[ultra thick, rounded corners=8pt, ->] 
    (-1.3, 1.3) -- (-0.2, 1.3) -- (-0.2, 0.2) -- (1.3, 0.2);

\end{scope}

\begin{scope}[shift={(12, 2.5)}]
\def\spacing{1.5}
\def\externalarrow{1.2}
\def\margin{0.5}

\node[fvertex] (caseB2_v0) at (0,0) {};
\node[fvertex] (caseB2_v1) at (0,1.5) {};
\node[vertex] (caseB2_v2) at (0,3) {};

\draw[gray] (-\spacing/2, -\margin) -- (\spacing/2, -\margin);
\draw[gray] (-\spacing/2, 3+ \margin) -- (\spacing/2, 3 + \margin);
\draw[gray] (-\spacing/2, -\margin) -- (-\spacing/2, 3+ \margin);
\draw[gray] (\spacing/2, -\margin) -- (\spacing/2, 3+ \margin);

\draw[<-, dashed] (-\externalarrow, 0) -- (caseB2_v0);
\draw[->, dashed] (caseB2_v0) -- (\externalarrow, 0);
\draw[->, dashed] (-\externalarrow, 1.5) -- (caseB2_v1);
\draw[<-, dashed] (caseB2_v1) -- (\externalarrow, 1.5);
\draw[<-, dashed] (-\externalarrow, 3) -- (caseB2_v2);
\draw[->, dashed] (caseB2_v2) -- (\externalarrow, 3);

\draw[->, dashed] (0, -\externalarrow) -- (caseB2_v0);
\draw[<-, dashed] (caseB2_v0) -- (caseB2_v1);
\draw[->, dashed] (caseB2_v1) -- (caseB2_v2);
\draw[->, dashed] (0,4.2) -- (caseB2_v2);

\draw[ultra thick, rounded corners=8pt, ->] 
    (-1.3, 1.7) -- (-0.2, 1.7) -- (-0.2, 2.8) -- (1.3, 2.8);

\end{scope}

\draw[line width=2pt, ->, dotted] (2, 1.5) -- (4.25, 3.2) 
node[midway, above] {\textbf{B}};

\draw[line width=2pt, ->, dotted] (2, 1.5) -- (4.25, -0.2)
node[midway, below]  {\textbf{A}};

\draw[line width=2pt, ->, dotted] (8, 3.4) -- (10.25, 3.4)
node[midway, above] {\textbf{B.A}};

\end{tikzpicture}

\caption{After performing a Bouncing motion within a vertical $G_{3,1}$.}
\label{fig:3B}
\end{figure}
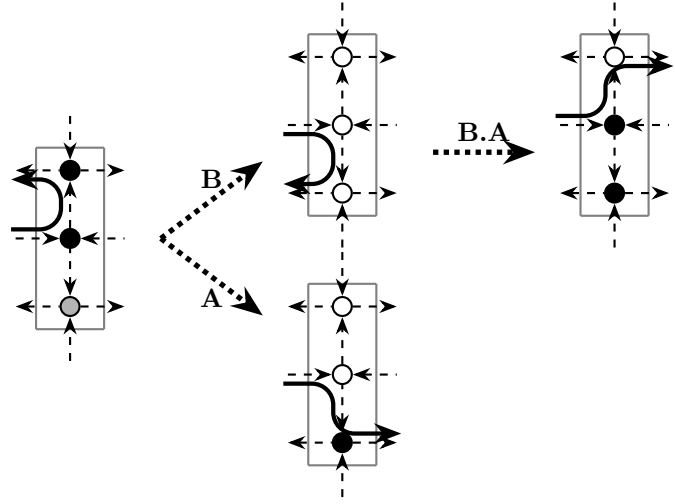
\end{proof}

\begin{corollary}\label{pr:Vcolbound}

During its motion within a three-row grid, if the motion performed by the ant upon its first visit to a column that is a vertical $G_{3,1}$ is: 
\begin{enumerate}
    \item  a crossing motion, then the ant performs at most six steps within this column before escaping the entire three-row grid. Moreover, these six steps consist of three crossing motions, two of which occur in the same direction;
    \item a bouncing motion, then the ant performs at most ten steps within this column before escaping the entire three-row grid. Moreover, these ten steps consist of two bouncing motion within the same side and three crossing motions, two of which start from the side of the bouncing;
    \item an initial motion, then the ant performs at most eleven steps within this column before escaping the entire three-row grid.
\end{enumerate}

\end{corollary}

\begin{proof}
Statements (1) and (2) follow directly from Lemma \ref{pr:Hcrossing} and Lemma \ref{pr:Hbouncing}.
For (3), it suffices to observe that any initial motion is performed in a single step and can occur only once. Upon a subsequent return of the ant to the column, the only possible motions are either a crossing or a bouncing. The conclusion then follows immediately from (1) and (2).
\end{proof}

We now give the proof of the main result.

\begin{proof}[Proof of Theorem \ref{pr:theothreerow}]
We decompose  $S_{3,n}= SH_{3,n} + SV_{3,n}$, where $SH_{3,n}$ and $SV_{3,n}$ denote respectively the maximum number of steps the ant can perform within all the horizontal columns and all the vertical columns of the three-row grid with $n$ columns.
Except for the initial motion and the final out motion, each of which is performed in a single step and occurs at most once, the motion of the ant can be decomposed into a succession of two-step segments within individual columns. Moreover, after performing two steps within a horizontal column, the ant must necessarily perform its next two steps within a vertical column. Conversely, after visiting a vertical column, the ant must necessarily move to a horizontal column. Therefore, the total contribution of horizontal columns cannot exceed that of vertical columns by more than two steps, and we obtain $SH_{3,n} \le SV_{3,n}+2$.
It remains to bound $SV_{3,n}$.

\textbf{Case A. $n$ is even.} 
In this case, there is exactly one vertical column that is a boundary column. On such a boundary column, if the first motion of the ant is a crossing, the ant can perform at most four steps on it, consisting of two crossings in opposite directions, the first making the ant enter the whole grid and the other making the ant leave it.
If the first motion is a bouncing, then the ant can perform at most six steps, consisting of two bouncing and one crossing, after which it exits the grid.
For non-boundary vertical columns, Corollary~\ref{pr:Vcolbound} provides a bound of $10$ steps per column, except for at most one column where the bound may increase to $11$ if the ant starts its motion there. Hence
$$SV_{3,n} \le 10\left(\frac{n}{2} - 1\right) + 6 + 1 = 5n - 3,$$
where the three terms account, respectively, for: the non-boundary vertical columns; 
the boundary vertical column; and the possible initial motion within a vertical column.

\textbf{Case B. $n$ is odd.}
The two boundary columns of the three-row grid are either both vertical \textbf{(B.A)} or both horizontal \textbf{(B.B)}. 

\textbf{Case B.A.} Both boundary columns are vertical. On one boundary column the ant may perform at most six steps (two bouncing and a crossing, after which it exits the grid), while on the other it may perform at most four steps (two bouncing motions), since the ant exits the grid only once. Hence,
$$SV_{3,n} \le 10\left(\frac{n+1}{2} - 2\right) + (6 + 4) + 1 = 5n - 4,$$
where the three terms account, respectively, for: the non-boundary vertical columns; 
the two boundary vertical columns; and the possible initial motion within a vertical column.

\textbf{Case B.B.} Both boundary columns are horizontal.
In this situation,
$$SV_{3,n} \le 10\left(\frac{n-1}{2}\right) + 1 = 5n - 4,$$
where the two terms account, respectively, for: the non-boundary vertical columns (all vertical columns are non-boundary in this case); and the possible initial motion within a vertical column.

In all cases, we obtain $SV_{3,n}\le 5n-3$.
Using $SH_{3,n} \le SV_{3,n}+2$, we have 
$$SH_{3,n}\le 5n-1 \text{  and } S_{3,n}=SH_{3,n}+SV_{3,n} \le 10n-4.$$
\end{proof}

\paragraph*{Lower bound for $S_{3,n}$}

The upper bound is nearly tight. Explicit constructions show that for every  $n \ge 3$:

$$
 S_{3,n} \ge 
\begin{cases}
10n - 12 & \text{ if } n \text{ is even.   (see Figure \ref{fig:3maxeven})} , \\
10n - 13 & \text{ if } n \text{ is odd.~    (see Figure \ref{fig:3maxodd}).}
\end{cases}
$$
Moreover, our simulations for $n\le20$ yield exact escaping times consistent with these constructions.

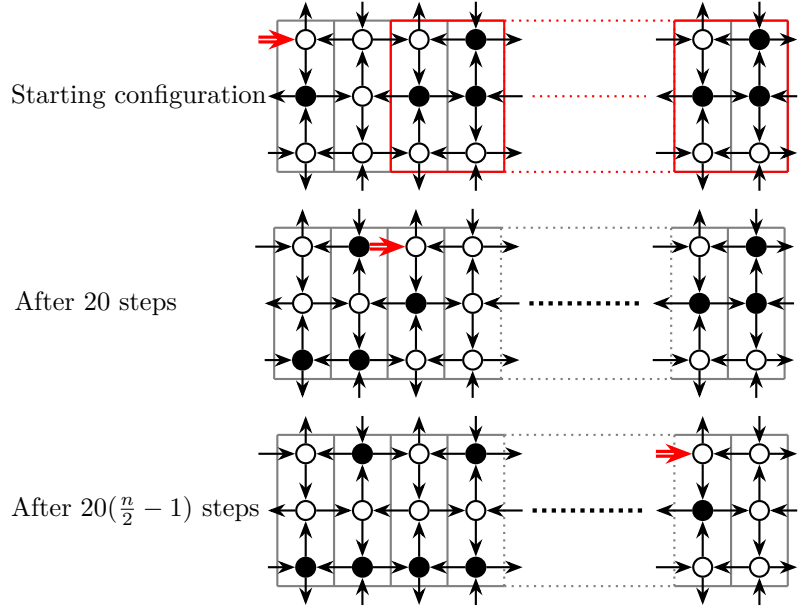
\begin{figure}[H]
\centering

\begin{tikzpicture}[
scale=0.5,
vertex/.style={circle, draw, minimum size=2.5mm, inner sep=0pt},
fvertex/.style={circle, draw, fill=black, minimum size=2.5mm, inner sep=0pt},
gvertex/.style={circle, draw, fill=gray!60, minimum size=2.5mm, inner sep=0pt},
>=Stealth,
thick
]
\node(t1) at (-4,1.5){Starting configuration $\,~~$};
\node[vertex] (0v0) at (0,0) {};
\node[vertex] (1v0) at (1.5,0) {};
\node[vertex] (2v0) at (3,0) {};
\node[vertex] (3v0) at (4.5,0) {};
\node (4v0) at (6,0) {};
\node (5v0) at (7.5,0) {};
\node (6v0) at (9,0) {};
\node[vertex] (7v0) at (10.5,0) {};
\node[vertex] (8v0) at (12,0) {};

\node[fvertex] (0v1) at (0,1.5) {};
\node[vertex] (1v1) at (1.5,1.5) {};
\node[fvertex] (2v1) at (3,1.5) {};
\node[fvertex] (3v1) at (4.5,1.5) {};
\node (4v1) at (6,1.5) {};
\node(5v1) at (7.5,1.5) {};
\node (6v1) at (9,1.5) {};
\node[fvertex] (7v1) at (10.5,1.5) {};
\node[fvertex] (8v1) at (12,1.5) {};

\node[vertex] (0v2) at (0,3) {};
\node[vertex] (1v2) at (1.5,3) {};
\node[vertex] (2v2) at (3,3) {};
\node[fvertex] (3v2) at (4.5,3) {};
\node (4v2) at (6,3) {};
\node (5v2) at (7.5,3) {};
\node (6v2) at (9,3) {};
\node[vertex] (7v2) at (10.5,3) {};
\node[fvertex] (8v2) at (12,3) {};

\draw[gray] (-0.75, -0.5) -- (5.25, -0.5);
\draw[gray] (-0.75, 3.5) -- (5.25, 3.5);
\draw[red, dotted] (5.25, -0.5) -- (9.75, -0.5);
\draw[red, dotted] (6, 1.5) -- (9, 1.5);
\draw[red, dotted] (5.25, 3.5) -- (9.75, 3.5);
\draw[gray] (9.75, -0.5) -- (12.75, -0.5);
\draw[gray] (9.75, 3.5) -- (12.75, 3.5);

\draw[gray] (-0.75,-0.5)--(-0.75,3.5);
\draw[gray] (0.75,-0.5)--(0.75,3.5);
\draw[gray] (2.25,-0.5)--(2.25,3.5);
\draw[gray] (3.75,-0.5)--(3.75,3.5);
\draw[gray, dotted] (5.25,-0.5)--(5.25,3.5);
\draw[gray, dotted] (9.75,-0.5)--(9.75,3.5);
\draw[gray] (11.25,-0.5)--(11.25,3.5);
\draw[gray] (12.75,-0.5)--(12.75,3.5);


\draw[red] (2.25,-0.5)--(2.25,3.5);
\draw[red] (5.25,-0.5)--(5.25,3.5);

\draw[red] (2.25,-0.5)--(5.25,-0.5);
\draw[red] (2.25,3.5)--(5.25,3.5);

\draw[red] (9.75,-0.5)--(9.75,3.5);
\draw[red] (12.75,-0.5)--(12.75,3.5);

\draw[red] (9.75,-0.5)--(12.75,-0.5);
\draw[red] (9.75,3.5)--(12.75,3.5);

\draw[->] (-1.0,0)--(0v0);
\draw[<-] (0v0)--(1v0);
\draw[->] (1v0)--(2v0);
\draw[<-] (2v0)--(3v0);
\draw[->] (3v0)--(4v0);

\draw[->] (6v0)--(7v0);
\draw[<-] (7v0)--(8v0);
\draw[->] (8v0)--(13.0,0);

\draw[<-] (-1.0,1.5)--(0v1);
\draw[->] (0v1)--(1v1);
\draw[<-] (1v1)--(2v1);
\draw[->] (2v1)--(3v1);
\draw[<-] (3v1)--(4v1);

\draw[<-] (6v1)--(7v1);
\draw[->] (7v1)--(8v1);
\draw[<-] (8v1)--(13.0,1.5);

\draw[red, ->, double, very thick] (-1.25,3)--(0v2);
\draw[<-] (0v2)--(1v2);
\draw[->] (1v2)--(2v2);
\draw[<-] (2v2)--(3v2);
\draw[->] (3v2)--(4v2);

\draw[->] (6v2)--(7v2);
\draw[<-] (7v2)--(8v2);
\draw[->] (8v2)--(13.0,3);

\draw[<-] (0,-1.0)--(0v0);
\draw[->] (0v0)--(0v1);
\draw[<-] (0v1)--(0v2);
\draw[->] (0v2)--(0,4.0);

\draw[->] (1.5,-1.0)--(1v0);
\draw[<-] (1v0)--(1v1);
\draw[->] (1v1)--(1v2);
\draw[<-] (1v2)--(1.5,4.0);

\draw[<-] (3,-1.0)--(2v0);
\draw[->] (2v0)--(2v1);
\draw[<-] (2v1)--(2v2);
\draw[->] (2v2)--(3,4.0);

\draw[->] (4.5,-1.0)--(3v0);
\draw[<-] (3v0)--(3v1);
\draw[->] (3v1)--(3v2);
\draw[<-] (3v2)--(4.5,4.0);

\draw[<-] (10.5,-1.0)--(7v0);
\draw[->] (7v0)--(7v1);
\draw[<-] (7v1)--(7v2);
\draw[->] (7v2)--(10.5,4.0);

\draw[->] (12,-1.0)--(8v0);
\draw[<-] (8v0)--(8v1);
\draw[->] (8v1)--(8v2);
\draw[<-] (8v2)--(12,4.0);

\end{tikzpicture}

\vspace{0.5em}

\begin{tikzpicture}[
scale=0.5,
vertex/.style={circle, draw, minimum size=2.5mm, inner sep=0pt},
fvertex/.style={circle, draw, fill=black, minimum size=2.5mm, inner sep=0pt},
gvertex/.style={circle, draw, fill=gray!60, minimum size=2.5mm, inner sep=0pt},
>=Stealth,
thick
]
\node(t1) at (-5,1.5){After $20$ steps $~~~$};
\node[fvertex] (0v0) at (0,0) {};
\node[fvertex] (1v0) at (1.5,0) {};
\node[vertex] (2v0) at (3,0) {};
\node[vertex] (3v0) at (4.5,0) {};
\node (4v0) at (6,0) {};
\node (5v0) at (7.5,0) {};
\node (6v0) at (9,0) {};
\node[vertex] (7v0) at (10.5,0) {};
\node[vertex] (8v0) at (12,0) {};

\node[vertex] (0v1) at (0,1.5) {};
\node[vertex] (1v1) at (1.5,1.5) {};
\node[fvertex] (2v1) at (3,1.5) {};
\node[vertex] (3v1) at (4.5,1.5) {};
\node (4v1) at (6,1.5) {};
\node(5v1) at (7.5,1.5) {};
\node (6v1) at (9,1.5) {};
\node[fvertex] (7v1) at (10.5,1.5) {};
\node[fvertex] (8v1) at (12,1.5) {};

\node[vertex] (0v2) at (0,3) {};
\node[fvertex] (1v2) at (1.5,3) {};
\node[vertex] (2v2) at (3,3) {};
\node[vertex] (3v2) at (4.5,3) {};
\node (4v2) at (6,3) {};
\node (5v2) at (7.5,3) {};
\node (6v2) at (9,3) {};
\node[vertex] (7v2) at (10.5,3) {};
\node[fvertex] (8v2) at (12,3) {};

\draw[gray] (-0.75, -0.5) -- (5.25, -0.5);
\draw[gray] (-0.75, 3.5) -- (5.25, 3.5);
\draw[gray, dotted] (5.25, -0.5) -- (9.75, -0.5);
\draw[gray, dotted] (5.25, 3.5) -- (9.75, 3.5);
\draw[gray] (9.75, -0.5) -- (12.75, -0.5);
\draw[gray] (9.75, 3.5) -- (12.75, 3.5);

\draw[gray] (-0.75,-0.5)--(-0.75,3.5);
\draw[gray] (0.75,-0.5)--(0.75,3.5);
\draw[gray] (2.25,-0.5)--(2.25,3.5);
\draw[gray] (3.75,-0.5)--(3.75,3.5);
\draw[gray, dotted] (5.25,-0.5)--(5.25,3.5);
\draw[gray, dotted] (9.75,-0.5)--(9.75,3.5);
\draw[gray] (11.25,-0.5)--(11.25,3.5);
\draw[gray] (12.75,-0.5)--(12.75,3.5);

\draw[ultra thick, dotted] (6, 1.5) -- (9, 1.5);

\draw[->] (-1.0,0)--(0v0);
\draw[<-] (0v0)--(1v0);
\draw[->] (1v0)--(2v0);
\draw[<-] (2v0)--(3v0);
\draw[->] (3v0)--(4v0);

\draw[->] (6v0)--(7v0);
\draw[<-] (7v0)--(8v0);
\draw[->] (8v0)--(13.0,0);

\draw[<-] (-1.0,1.5)--(0v1);
\draw[->] (0v1)--(1v1);
\draw[<-] (1v1)--(2v1);
\draw[->] (2v1)--(3v1);
\draw[<-] (3v1)--(4v1);

\draw[<-] (6v1)--(7v1);
\draw[->] (7v1)--(8v1);
\draw[<-] (8v1)--(13.0,1.5);

\draw[->] (-1.25,3)--(0v2);
\draw[<-] (0v2)--(1v2);
\draw[->,red, double, very thick] (1v2)--(2v2);
\draw[<-] (2v2)--(3v2);
\draw[->] (3v2)--(4v2);

\draw[-> ] (6v2)--(7v2);
\draw[<-] (7v2)--(8v2);
\draw[->] (8v2)--(13.0,3);

\draw[<-] (0,-1.0)--(0v0);
\draw[->] (0v0)--(0v1);
\draw[<-] (0v1)--(0v2);
\draw[->] (0v2)--(0,4.0);

\draw[->] (1.5,-1.0)--(1v0);
\draw[<-] (1v0)--(1v1);
\draw[->] (1v1)--(1v2);
\draw[<-] (1v2)--(1.5,4.0);

\draw[<-] (3,-1.0)--(2v0);
\draw[->] (2v0)--(2v1);
\draw[<-] (2v1)--(2v2);
\draw[->] (2v2)--(3,4.0);

\draw[->] (4.5,-1.0)--(3v0);
\draw[<-] (3v0)--(3v1);
\draw[->] (3v1)--(3v2);
\draw[<-] (3v2)--(4.5,4.0);

\draw[<-] (10.5,-1.0)--(7v0);
\draw[->] (7v0)--(7v1);
\draw[<-] (7v1)--(7v2);
\draw[->] (7v2)--(10.5,4.0);

\draw[->] (12,-1.0)--(8v0);
\draw[<-] (8v0)--(8v1);
\draw[->] (8v1)--(8v2);
\draw[<-] (8v2)--(12,4.0);

\end{tikzpicture}

\vspace{0.5em}

\begin{tikzpicture}[
scale=0.5,
vertex/.style={circle, draw, minimum size=2.5mm, inner sep=0pt},
fvertex/.style={circle, draw, fill=black, minimum size=2.5mm, inner sep=0pt},
gvertex/.style={circle, draw, fill=gray!60, minimum size=2.5mm, inner sep=0pt},
>=Stealth,
thick
]
\node(t1) at (-4.5,1.5){After $20(\frac{n}{2}-1)$ steps};
\node[fvertex] (0v0) at (0,0) {};
\node[fvertex] (1v0) at (1.5,0) {};
\node[fvertex] (2v0) at (3,0) {};
\node[fvertex] (3v0) at (4.5,0) {};
\node (4v0) at (6,0) {};
\node (5v0) at (7.5,0) {};
\node (6v0) at (9,0) {};
\node[vertex] (7v0) at (10.5,0) {};
\node[vertex] (8v0) at (12,0) {};

\node[vertex] (0v1) at (0,1.5) {};
\node[vertex] (1v1) at (1.5,1.5) {};
\node[vertex] (2v1) at (3,1.5) {};
\node[vertex] (3v1) at (4.5,1.5) {};
\node (4v1) at (6,1.5) {};
\node(5v1) at (7.5,1.5) {};
\node (6v1) at (9,1.5) {};
\node[fvertex] (7v1) at (10.5,1.5) {};
\node[vertex] (8v1) at (12,1.5) {};

\node[vertex] (0v2) at (0,3) {};
\node[fvertex] (1v2) at (1.5,3) {};
\node[vertex] (2v2) at (3,3) {};
\node[fvertex] (3v2) at (4.5,3) {};
\node (4v2) at (6,3) {};
\node (5v2) at (7.5,3) {};
\node (6v2) at (9,3) {};
\node[vertex] (7v2) at (10.5,3) {};
\node[vertex] (8v2) at (12,3) {};

\draw[gray] (-0.75, -0.5) -- (5.25, -0.5);
\draw[gray] (-0.75, 3.5) -- (5.25, 3.5);
\draw[gray, dotted] (5.25, -0.5) -- (9.75, -0.5);
\draw[gray, dotted] (5.25, 3.5) -- (9.75, 3.5);
\draw[gray] (9.75, -0.5) -- (12.75, -0.5);
\draw[gray] (9.75, 3.5) -- (12.75, 3.5);

\draw[gray] (-0.75,-0.5)--(-0.75,3.5);
\draw[gray] (0.75,-0.5)--(0.75,3.5);
\draw[gray] (2.25,-0.5)--(2.25,3.5);
\draw[gray] (3.75,-0.5)--(3.75,3.5);
\draw[gray, dotted] (5.25,-0.5)--(5.25,3.5);
\draw[gray, dotted] (9.75,-0.5)--(9.75,3.5);
\draw[gray] (11.25,-0.5)--(11.25,3.5);
\draw[gray] (12.75,-0.5)--(12.75,3.5);

\draw[ultra thick, dotted] (6, 1.5) -- (9, 1.5);

\draw[->] (-1.0,0)--(0v0);
\draw[<-] (0v0)--(1v0);
\draw[->] (1v0)--(2v0);
\draw[<-] (2v0)--(3v0);
\draw[->] (3v0)--(4v0);

\draw[->] (6v0)--(7v0);
\draw[<-] (7v0)--(8v0);
\draw[->] (8v0)--(13.0,0);

\draw[<-] (-1.0,1.5)--(0v1);
\draw[->] (0v1)--(1v1);
\draw[<-] (1v1)--(2v1);
\draw[->] (2v1)--(3v1);
\draw[<-] (3v1)--(4v1);

\draw[<-] (6v1)--(7v1);
\draw[->] (7v1)--(8v1);
\draw[<-] (8v1)--(13.0,1.5);

\draw[->] (-1.25,3)--(0v2);
\draw[<-] (0v2)--(1v2);
\draw[->] (1v2)--(2v2);
\draw[<-] (2v2)--(3v2);
\draw[->] (3v2)--(4v2);

\draw[-> , red, double, very thick] (6v2)--(7v2);
\draw[<-] (7v2)--(8v2);
\draw[->] (8v2)--(13.0,3);

\draw[<-] (0,-1.0)--(0v0);
\draw[->] (0v0)--(0v1);
\draw[<-] (0v1)--(0v2);
\draw[->] (0v2)--(0,4.0);

\draw[->] (1.5,-1.0)--(1v0);
\draw[<-] (1v0)--(1v1);
\draw[->] (1v1)--(1v2);
\draw[<-] (1v2)--(1.5,4.0);

\draw[<-] (3,-1.0)--(2v0);
\draw[->] (2v0)--(2v1);
\draw[<-] (2v1)--(2v2);
\draw[->] (2v2)--(3,4.0);

\draw[->] (4.5,-1.0)--(3v0);
\draw[<-] (3v0)--(3v1);
\draw[->] (3v1)--(3v2);
\draw[<-] (3v2)--(4.5,4.0);

\draw[<-] (10.5,-1.0)--(7v0);
\draw[->] (7v0)--(7v1);
\draw[<-] (7v1)--(7v2);
\draw[->] (7v2)--(10.5,4.0);

\draw[->] (12,-1.0)--(8v0);
\draw[<-] (8v0)--(8v1);
\draw[->] (8v1)--(8v2);
\draw[<-] (8v2)--(12,4.0);

\end{tikzpicture}

\caption{For even $n$ . The current position of the ant is indicated by the red double arrow. From the last configuration, the ant will escape the grid after $8$ additional steps.}
\label{fig:3maxeven}

\end{figure}

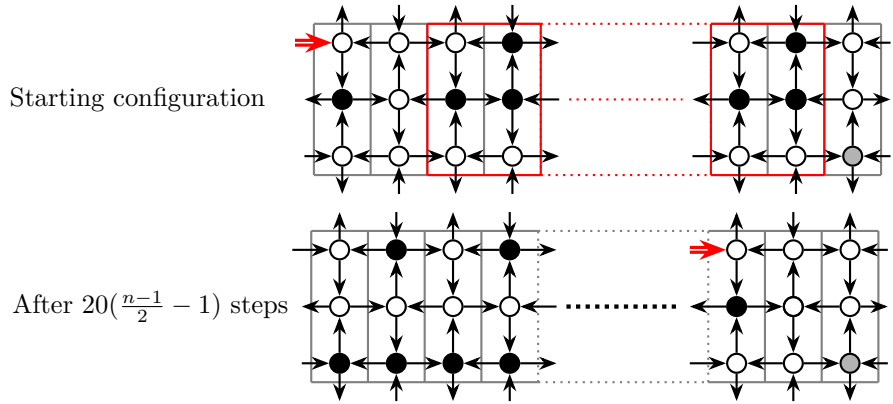
\begin{figure}[H]
\centering

\begin{tikzpicture}[
scale=0.5,
vertex/.style={circle, draw, minimum size=2.5mm, inner sep=0pt},
fvertex/.style={circle, draw, fill=black, minimum size=2.5mm, inner sep=0pt},
gvertex/.style={circle, draw, fill=gray!60, minimum size=2.5mm, inner sep=0pt},
>=Stealth,
thick
]
\node(t1) at (-5,1.5){Starting configuration $\,~~$};
\node[vertex] (0v0) at (0,0) {};
\node[vertex] (1v0) at (1.5,0) {};
\node[vertex] (2v0) at (3,0) {};
\node[vertex] (3v0) at (4.5,0) {};
\node (4v0) at (6,0) {};
\node (5v0) at (7.5,0) {};
\node (6v0) at (9,0) {};
\node[vertex] (7v0) at (10.5,0) {};
\node[vertex] (8v0) at (12,0) {};
\node[gvertex] (9v0) at (13.5,0) {};

\node[fvertex] (0v1) at (0,1.5) {};
\node[vertex] (1v1) at (1.5,1.5) {};
\node[fvertex] (2v1) at (3,1.5) {};
\node[fvertex] (3v1) at (4.5,1.5) {};
\node (4v1) at (6,1.5) {};
\node(5v1) at (7.5,1.5) {};
\node (6v1) at (9,1.5) {};
\node[fvertex] (7v1) at (10.5,1.5) {};
\node[fvertex] (8v1) at (12,1.5) {};
\node[fvertex] (8v1) at (12,1.5) {};
\node[vertex] (9v1) at (13.5,1.5) {};

\node[vertex] (0v2) at (0,3) {};
\node[vertex] (1v2) at (1.5,3) {};
\node[vertex] (2v2) at (3,3) {};
\node[fvertex] (3v2) at (4.5,3) {};
\node (4v2) at (6,3) {};
\node (5v2) at (7.5,3) {};
\node (6v2) at (9,3) {};
\node[vertex] (7v2) at (10.5,3) {};
\node[fvertex] (8v2) at (12,3) {};
\node[vertex] (9v2) at (13.5,3) {};
\draw[gray] (-0.75, -0.5) -- (5.25, -0.5);
\draw[gray] (-0.75, 3.5) -- (5.25, 3.5);
\draw[red, dotted] (5.25, -0.5) -- (9.75, -0.5);
\draw[red, dotted] (6, 1.5) -- (9, 1.5);
\draw[red, dotted] (5.25, 3.5) -- (9.75, 3.5);
\draw[gray] (9.75, -0.5) -- (14.25, -0.5);
\draw[gray] (9.75, 3.5) -- (14.25, 3.5);

\draw[gray] (-0.75,-0.5)--(-0.75,3.5);
\draw[gray] (0.75,-0.5)--(0.75,3.5);
\draw[gray] (2.25,-0.5)--(2.25,3.5);
\draw[gray] (3.75,-0.5)--(3.75,3.5);
\draw[gray, dotted] (5.25,-0.5)--(5.25,3.5);
\draw[gray, dotted] (9.75,-0.5)--(9.75,3.5);
\draw[gray] (11.25,-0.5)--(11.25,3.5);
\draw[gray] (12.75,-0.5)--(12.75,3.5);
\draw[gray] (14.25,-0.5)--(14.25,3.5);

\draw[red] (2.25,-0.5)--(2.25,3.5);
\draw[red] (5.25,-0.5)--(5.25,3.5);

\draw[red] (2.25,-0.5)--(5.25,-0.5);
\draw[red] (2.25,3.5)--(5.25,3.5);

\draw[red] (9.75,-0.5)--(9.75,3.5);
\draw[red] (12.75,-0.5)--(12.75,3.5);

\draw[red] (9.75,-0.5)--(12.75,-0.5);
\draw[red] (9.75,3.5)--(12.75,3.5);

\draw[->] (-1.0,0)--(0v0);
\draw[<-] (0v0)--(1v0);
\draw[->] (1v0)--(2v0);
\draw[<-] (2v0)--(3v0);
\draw[->] (3v0)--(4v0);

\draw[->] (6v0)--(7v0);
\draw[<-] (7v0)--(8v0);
\draw[->] (8v0)--(9v0);
\draw[<-] (9v0)--(14.5, 0);
\draw[<-] (-1.0,1.5)--(0v1);
\draw[->] (0v1)--(1v1);
\draw[<-] (1v1)--(2v1);
\draw[->] (2v1)--(3v1);
\draw[<-] (3v1)--(4v1);

\draw[<-] (6v1)--(7v1);
\draw[->] (7v1)--(8v1);
\draw[<-] (8v1)--(9v1);
\draw[->] (9v1)--(14.5,1.5);
\draw[->, red, double, very thick] (-1.25,3)--(0v2);
\draw[<-] (0v2)--(1v2);
\draw[->] (1v2)--(2v2);
\draw[<-] (2v2)--(3v2);
\draw[->] (3v2)--(4v2);

\draw[->] (6v2)--(7v2);
\draw[<-] (7v2)--(8v2);
\draw[->] (8v2)--(9v2);
\draw[<-] (9v2)--(14.5, 3);

\draw[<-] (0,-1.0)--(0v0);
\draw[->] (0v0)--(0v1);
\draw[<-] (0v1)--(0v2);
\draw[->] (0v2)--(0,4.0);

\draw[->] (1.5,-1.0)--(1v0);
\draw[<-] (1v0)--(1v1);
\draw[->] (1v1)--(1v2);
\draw[<-] (1v2)--(1.5,4.0);

\draw[<-] (3,-1.0)--(2v0);
\draw[->] (2v0)--(2v1);
\draw[<-] (2v1)--(2v2);
\draw[->] (2v2)--(3,4.0);

\draw[->] (4.5,-1.0)--(3v0);
\draw[<-] (3v0)--(3v1);
\draw[->] (3v1)--(3v2);
\draw[<-] (3v2)--(4.5,4.0);

\draw[<-] (10.5,-1.0)--(7v0);
\draw[->] (7v0)--(7v1);
\draw[<-] (7v1)--(7v2);
\draw[->] (7v2)--(10.5,4.0);

\draw[->] (12,-1.0)--(8v0);
\draw[<-] (8v0)--(8v1);
\draw[->] (8v1)--(8v2);
\draw[<-] (8v2)--(12,4.0);

\draw[<-] (13.5,-1.0)--(9v0);
\draw[->] (9v0)--(9v1);
\draw[<-] (9v1)--(9v2);
\draw[->] (9v2)--(13.5,4.0);

\end{tikzpicture}

\vspace{0.5em}

\begin{tikzpicture}[
scale=0.5,
vertex/.style={circle, draw, minimum size=2.5mm, inner sep=0pt},
fvertex/.style={circle, draw, fill=black, minimum size=2.5mm, inner sep=0pt},
gvertex/.style={circle, draw, fill=gray!60, minimum size=2.5mm, inner sep=0pt},
>=Stealth,
thick
]
\node(t1) at (-5,1.5){After $20(\frac{n-1}{2} -1)$ steps};
\node[fvertex] (0v0) at (0,0) {};
\node[fvertex] (1v0) at (1.5,0) {};
\node[fvertex] (2v0) at (3,0) {};
\node[fvertex] (3v0) at (4.5,0) {};
\node (4v0) at (6,0) {};
\node (5v0) at (7.5,0) {};
\node (6v0) at (9,0) {};
\node[vertex] (7v0) at (10.5,0) {};
\node[vertex] (8v0) at (12,0) {};
\node[gvertex] (9v0) at (13.5,0) {};

\node[vertex] (0v1) at (0,1.5) {};
\node[vertex] (1v1) at (1.5,1.5) {};
\node[vertex] (2v1) at (3,1.5) {};
\node[vertex] (3v1) at (4.5,1.5) {};
\node (4v1) at (6,1.5) {};
\node(5v1) at (7.5,1.5) {};
\node (6v1) at (9,1.5) {};
\node[fvertex] (7v1) at (10.5,1.5) {};
\node[vertex] (8v1) at (12,1.5) {};
\node[vertex] (8v1) at (12,1.5) {};
\node[vertex] (9v1) at (13.5,1.5) {};

\node[vertex] (0v2) at (0,3) {};
\node[fvertex] (1v2) at (1.5,3) {};
\node[vertex] (2v2) at (3,3) {};
\node[fvertex] (3v2) at (4.5,3) {};
\node (4v2) at (6,3) {};
\node (5v2) at (7.5,3) {};
\node (6v2) at (9,3) {};
\node[vertex] (7v2) at (10.5,3) {};
\node[vertex] (8v2) at (12,3) {};
\node[vertex] (9v2) at (13.5,3) {};
\draw[gray] (-0.75, -0.5) -- (5.25, -0.5);
\draw[gray] (-0.75, 3.5) -- (5.25, 3.5);
\draw[gray, dotted] (5.25, -0.5) -- (9.75, -0.5);
\draw[gray, dotted] (5.25, 3.5) -- (9.75, 3.5);
\draw[gray] (9.75, -0.5) -- (14.25, -0.5);
\draw[gray] (9.75, 3.5) -- (14.25, 3.5);

\draw[gray] (-0.75,-0.5)--(-0.75,3.5);
\draw[gray] (0.75,-0.5)--(0.75,3.5);
\draw[gray] (2.25,-0.5)--(2.25,3.5);
\draw[gray] (3.75,-0.5)--(3.75,3.5);
\draw[gray, dotted] (5.25,-0.5)--(5.25,3.5);
\draw[gray, dotted] (9.75,-0.5)--(9.75,3.5);
\draw[gray] (11.25,-0.5)--(11.25,3.5);
\draw[gray] (12.75,-0.5)--(12.75,3.5);
\draw[gray] (14.25,-0.5)--(14.25,3.5);

\draw[ultra thick, dotted] (6, 1.5) -- (9, 1.5);

\draw[->] (-1.0,0)--(0v0);
\draw[<-] (0v0)--(1v0);
\draw[->] (1v0)--(2v0);
\draw[<-] (2v0)--(3v0);
\draw[->] (3v0)--(4v0);

\draw[->] (6v0)--(7v0);
\draw[<-] (7v0)--(8v0);
\draw[->] (8v0)--(9v0);
\draw[<-] (9v0)--(14.5, 0);
\draw[<-] (-1.0,1.5)--(0v1);
\draw[->] (0v1)--(1v1);
\draw[<-] (1v1)--(2v1);
\draw[->] (2v1)--(3v1);
\draw[<-] (3v1)--(4v1);

\draw[<-] (6v1)--(7v1);
\draw[->] (7v1)--(8v1);
\draw[<-] (8v1)--(9v1);
\draw[->] (9v1)--(14.5,1.5);
\draw[->] (-1.25,3)--(0v2);
\draw[<-] (0v2)--(1v2);
\draw[->] (1v2)--(2v2);
\draw[<-] (2v2)--(3v2);
\draw[->] (3v2)--(4v2);

\draw[red, ->, double, very thick] (6v2)--(7v2);
\draw[<-] (7v2)--(8v2);
\draw[->] (8v2)--(9v2);
\draw[<-] (9v2)--(14.5, 3);

\draw[<-] (0,-1.0)--(0v0);
\draw[->] (0v0)--(0v1);
\draw[<-] (0v1)--(0v2);
\draw[->] (0v2)--(0,4.0);

\draw[->] (1.5,-1.0)--(1v0);
\draw[<-] (1v0)--(1v1);
\draw[->] (1v1)--(1v2);
\draw[<-] (1v2)--(1.5,4.0);

\draw[<-] (3,-1.0)--(2v0);
\draw[->] (2v0)--(2v1);
\draw[<-] (2v1)--(2v2);
\draw[->] (2v2)--(3,4.0);

\draw[->] (4.5,-1.0)--(3v0);
\draw[<-] (3v0)--(3v1);
\draw[->] (3v1)--(3v2);
\draw[<-] (3v2)--(4.5,4.0);

\draw[<-] (10.5,-1.0)--(7v0);
\draw[->] (7v0)--(7v1);
\draw[<-] (7v1)--(7v2);
\draw[->] (7v2)--(10.5,4.0);

\draw[->] (12,-1.0)--(8v0);
\draw[<-] (8v0)--(8v1);
\draw[->] (8v1)--(8v2);
\draw[<-] (8v2)--(12,4.0);

\draw[<-] (13.5,-1.0)--(9v0);
\draw[->] (9v0)--(9v1);
\draw[<-] (9v1)--(9v2);
\draw[->] (9v2)--(13.5,4.0);

\end{tikzpicture}

\caption{For odd $n$ . The current position of the ant is indicated by the red double arrow. From the last configuration, the ant will escape the grid after $17$ additional steps.}
\label{fig:3maxodd}

\end{figure}

\subsubsection{Four-row Grids: Lower bound and conjecture}\label{sec:fourrow}

Explicit constructions (Figure \ref{fig:4max}) show that for every  $n \ge 14:$ 
$$ S_{4,n} \ge 34(n-5) + 20. $$
Moreover, our simulations indicate that $S_{4,n} = 34(n-5) + 20$ for even $n$ with $14\le n \le 18$ and odd $n$ with $9\le n \le 17$.  These observations lead us to conjecture the following linear upper bound.

\begin{conjecture}
For every integer $n \ge 2$, one has $S_{4,n} \le 34n$.
\end{conjecture}

\begin{figure}[htbp]
\centering

When $n$ is even:
\begin{tikzpicture}[
scale=0.5,
vertex/.style={circle, draw, minimum size=2.5mm, inner sep=0pt},
fvertex/.style={circle, draw, fill=black, minimum size=2.5mm, inner sep=0pt},
gvertex/.style={circle, draw, fill=gray!60, minimum size=2.5mm, inner sep=0pt},
>=Stealth,
thick
]

\node[gvertex] (0v0) at (0,0) {};
\node[vertex] (1v0) at (1.5,0) {};
\node[fvertex] (2v0) at (3,0) {};
\node[vertex] (3v0) at (4.5,0) {};
\node[vertex] (4v0) at (6,0) {};
\node[vertex] (5v0) at (7.5,0) {};
\node[vertex] (6v0) at (9,0) {};
\node (7v0) at (10.5,0) {};
\node (8v0) at (13.5,0) {};
\node[vertex] (9v0) at (15,0) {};
\node[vertex] (10v0) at (16.5,0) {};
\node[vertex] (11v0) at (18,0) {};
\node[vertex] (12v0) at (19.5,0) {};
\node[gvertex] (13v0) at (21,0) {};

\node[vertex] (0v1) at (0,1.5) {};
\node[fvertex] (1v1) at (1.5,1.5) {};
\node[vertex] (2v1) at (3,1.5) {};
\node[vertex] (3v1) at (4.5,1.5) {};
\node[vertex] (4v1) at (6,1.5) {};
\node[vertex] (5v1) at (7.5,1.5) {};
\node[fvertex] (6v1) at (9,1.5) {};
\node (7v1) at (10.5,1.5) {};
\node (8v1) at (13.5,1.5) {};
\node[vertex] (9v1) at (15,1.5) {};
\node[fvertex] (10v1) at (16.5,1.5) {};
\node[vertex] (11v1) at (18,1.5) {};
\node[fvertex] (12v1) at (19.5,1.5) {};
\node[vertex] (13v1) at (21,1.5) {};

\node[vertex] (0v2) at (0,3) {};
\node[fvertex] (1v2) at (1.5,3) {};
\node[fvertex] (2v2) at (3,3) {};
\node[vertex] (3v2) at (4.5,3) {};
\node[vertex] (4v2) at (6,3) {};
\node[vertex] (5v2) at (7.5,3) {};
\node[fvertex] (6v2) at (9,3) {};
\node (7v2) at (10.5,3) {};
\node (8v2) at (13.5,3) {};
\node[vertex] (9v2) at (15,3) {};
\node[fvertex] (10v2) at (16.5,3) {};
\node[vertex] (11v2) at (18,3) {};
\node[fvertex] (12v2) at (19.5,3) {};
\node[fvertex] (13v2) at (21,3) {};

\node[fvertex] (0v3) at (0,4.5) {};
\node[fvertex] (1v3) at (1.5,4.5) {};
\node[vertex] (2v3) at (3,4.5) {};
\node[fvertex] (3v3) at (4.5,4.5) {};
\node[fvertex] (4v3) at (6,4.5) {};
\node[fvertex] (5v3) at (7.5,4.5) {};
\node[fvertex] (6v3) at (9,4.5) {};
\node (7v3) at (10.5,4.5) {};
\node (8v3) at (13.5,4.5) {};
\node[fvertex] (9v3) at (15,4.5) {};
\node[fvertex] (10v3) at (16.5,4.5) {};
\node[vertex] (11v3) at (18,4.5) {};
\node[vertex] (12v3) at (19.5,4.5) {};
\node[fvertex] (13v3) at (21,4.5) {};

\draw[gray] (-0.75, -0.5) -- (9.75, -0.5);
\draw[gray] (-0.75, 5) -- (9.75, 5);

\draw[red, dotted] (9.75, -0.5) -- (14.25, -0.5);
\draw[red, dotted] (10.25, 2.5) -- (13.75, 2.5);
\draw[red, dotted] (9.75, 5) -- (14.25, 5);

\draw[gray] (14.25, -0.5) -- (21.75, -0.5);
\draw[gray] (14.25, 5) -- (21.75, 5);

\draw[gray] (-0.75,-0.5)--(-0.75,5);
\draw[gray] (0.75,-0.5)--(0.75,5);
\draw[gray] (2.25,-0.5)--(2.25,5);
\draw[gray] (3.75,-0.5)--(3.75,5);
\draw[gray] (5.25,-0.5)--(5.25,5);
\draw[gray] (6.75,-0.5)--(6.75,5);
\draw[gray] (8.25,-0.5)--(8.25,5);
\draw[gray, dotted] (9.75,-0.5)--(9.75,5);
\draw[gray, dotted] (14.25,-0.5)--(14.25,5);
\draw[gray] (15.75,-0.5)--(15.75,5);
\draw[gray] (17.25,-0.5)--(17.25,5);
\draw[gray] (18.75,-0.5)--(18.75,5);
\draw[gray] (20.25,-0.5)--(20.25,5);
\draw[gray] (21.75,-0.5)--(21.75,5);

\draw[red] (6.75,-0.5)--(6.75,5);
\draw[red] (9.75,-0.5)--(9.75,5);
\draw[red] (6.75,-0.5)--(9.75,-0.5);
\draw[red] (6.75,5)--(9.75,5);

\draw[red] (14.25,-0.5)--(14.25,5);
\draw[red] (17.25,-0.5)--(17.25,5);
\draw[red] (14.25,-0.5)--(17.25,-0.5);
\draw[red] (14.25,5)--(17.25,5);

\draw[<-] (-1.0,0)--(0v0);
\draw[->] (0v0)--(1v0);
\draw[<-] (1v0)--(2v0);
\draw[->] (2v0)--(3v0);
\draw[<-] (3v0)--(4v0);
\draw[->] (4v0)--(5v0);
\draw[<-] (5v0)--(6v0);
\draw[->] (6v0)--(7v0);
\draw[->] (8v0)--(9v0);
\draw[<-] (9v0)--(10v0);
\draw[->] (10v0)--(11v0);
\draw[<-] (11v0)--(12v0);
\draw[->] (12v0)--(13v0);
\draw[<-] (13v0)--(22,0);

\draw[->] (-1.0,1.5)--(0v1);
\draw[<-] (0v1)--(1v1);
\draw[->] (1v1)--(2v1);
\draw[<-] (2v1)--(3v1);
\draw[->] (3v1)--(4v1);
\draw[<-] (4v1)--(5v1);
\draw[->] (5v1)--(6v1);
\draw[<-] (6v1)--(7v1);
\draw[<-] (8v1)--(9v1);
\draw[->] (9v1)--(10v1);
\draw[<-] (10v1)--(11v1);
\draw[->] (11v1)--(12v1);
\draw[<-] (12v1)--(13v1);
\draw[->] (13v1)--(22,1.5);

\draw[<-] (-1.0,3)--(0v2);
\draw[->] (0v2)--(1v2);
\draw[<-] (1v2)--(2v2);
\draw[->] (2v2)--(3v2);
\draw[<-] (3v2)--(4v2);
\draw[->] (4v2)--(5v2);
\draw[<-] (5v2)--(6v2);
\draw[->] (6v2)--(7v2);
\draw[->] (8v2)--(9v2);
\draw[<-] (9v2)--(10v2);
\draw[->] (10v2)--(11v2);
\draw[<-] (11v2)--(12v2);
\draw[->] (12v2)--(13v2);
\draw[<-] (13v2)--(22,3);

\draw[->] (-1,4.5)--(0v3);
\draw[<-] (0v3)--(1v3);
\draw[->] (1v3)--(2v3);
\draw[<-] (2v3)--(3v3);
\draw[->] (3v3)--(4v3);
\draw[<-] (4v3)--(5v3);
\draw[->] (5v3)--(6v3);
\draw[<-] (6v3)--(7v3);
\draw[<-] (8v3)--(9v3);
\draw[->] (9v3)--(10v3);
\draw[<-] (10v3)--(11v3);
\draw[->] (11v3)--(12v3);
\draw[<-] (12v3)--(13v3);
\draw[->] (13v3)--(22,4.5);

\foreach \i in {0,1,2,3,4,5,6} {
    \pgfmathsetmacro{\x}{\i*1.5}
    \pgfmathsetmacro{\dir}{int(mod(\i,2)==0?1:0)}
    \ifnum\dir=1
        \draw[->] (\x,-1.0)--(\i v0);
        \draw[<-] (\i v0)--(\i v1);
        \draw[->] (\i v1)--(\i v2);
        \draw[<-] (\i v2)--(\i v3);
        \draw[->] (\i v3)--(\x,5.5);
    \else
        \draw[<-] (\x,-1.0)--(\i v0);
        \draw[->] (\i v0)--(\i v1);
        \draw[<-] (\i v1)--(\i v2);
        \draw[->] (\i v2)--(\i v3);
        \draw[<-] (\i v3)--(\x,5.5);
    \fi
}

\foreach \i in {9,10,11,12,13} {
    \pgfmathsetmacro{\x}{\i*1.5+1.5}
    \pgfmathsetmacro{\dir}{int(mod(\i,2)==1?1:0)}
    \ifnum\dir=1
        \draw[<-] (\x,-1.0)--(\i v0);
        \draw[->] (\i v0)--(\i v1);
        \draw[<-] (\i v1)--(\i v2);
        \draw[->] (\i v2)--(\i v3);
        \draw[<-] (\i v3)--(\x,5.5);
    \else
        \draw[->] (\x,-1.0)--(\i v0);
        \draw[<-] (\i v0)--(\i v1);
        \draw[->] (\i v1)--(\i v2);
        \draw[<-] (\i v2)--(\i v3);
        \draw[->] (\i v3)--(\x,5.5);
    \fi
}
\draw[red, ->, double, very thick] (1.5, 6)--(1v3);
\end{tikzpicture}

\vspace{1em}

When $n$ is odd:
\begin{tikzpicture}[
scale=0.5,
vertex/.style={circle, draw, minimum size=2.5mm, inner sep=0pt},
fvertex/.style={circle, draw, fill=black, minimum size=2.5mm, inner sep=0pt},
gvertex/.style={circle, draw, fill=gray!60, minimum size=2.5mm, inner sep=0pt},
>=Stealth,
thick
]

%
%
%
%

\node[gvertex] (0v0) at (0,0) {};
\node[vertex] (1v0) at (1.5,0) {};
\node[fvertex] (2v0) at (3,0) {};
\node[vertex] (3v0) at (4.5,0) {};
\node[vertex] (4v0) at (6,0) {};
\node[vertex] (5v0) at (7.5,0) {};
\node[vertex] (6v0) at (9,0) {};
\node (7v0) at (10.5,0) {};
\node (8v0) at (13.5,0) {};
\node[vertex] (9v0) at (15,0) {};
\node[vertex] (10v0) at (16.5,0) {};
\node[vertex] (11v0) at (18,0) {};
\node[fvertex] (12v0) at (19.5,0) {};
\node[fvertex] (13v0) at (21,0) {};
\node[vertex] (14v0) at (22.5,0) {};

\node[vertex] (0v1) at (0,1.5) {};
\node[fvertex] (1v1) at (1.5,1.5) {};
\node[vertex] (2v1) at (3,1.5) {};
\node[vertex] (3v1) at (4.5,1.5) {};
\node[vertex] (4v1) at (6,1.5) {};
\node[vertex] (5v1) at (7.5,1.5) {};
\node[fvertex] (6v1) at (9,1.5) {};
\node (7v1) at (10.5,1.5) {};
\node (8v1) at (13.5,1.5) {};
\node[vertex] (9v1) at (15,1.5) {};
\node[fvertex] (10v1) at (16.5,1.5) {};
\node[vertex] (11v1) at (18,1.5) {};
\node[fvertex] (12v1) at (19.5,1.5) {};
\node[vertex] (13v1) at (21,1.5) {};
\node[vertex] (14v1) at (22.5,1.5) {};

\node[vertex] (0v2) at (0,3) {};
\node[fvertex] (1v2) at (1.5,3) {};
\node[fvertex] (2v2) at (3,3) {};
\node[vertex] (3v2) at (4.5,3) {};
\node[vertex] (4v2) at (6,3) {};
\node[vertex] (5v2) at (7.5,3) {};
\node[fvertex] (6v2) at (9,3) {};
\node (7v2) at (10.5,3) {};
\node (8v2) at (13.5,3) {};
\node[vertex] (9v2) at (15,3) {};
\node[fvertex] (10v2) at (16.5,3) {};
\node[vertex] (11v2) at (18,3) {};
\node[fvertex] (12v2) at (19.5,3) {};
\node[vertex] (13v2) at (21,3) {};
\node[fvertex] (14v2) at (22.5,3) {};

\node[fvertex] (0v3) at (0,4.5) {};
\node[fvertex] (1v3) at (1.5,4.5) {};
\node[vertex] (2v3) at (3,4.5) {};
\node[fvertex] (3v3) at (4.5,4.5) {};
\node[fvertex] (4v3) at (6,4.5) {};
\node[fvertex] (5v3) at (7.5,4.5) {};
\node[fvertex] (6v3) at (9,4.5) {};
\node (7v3) at (10.5,4.5) {};
\node (8v3) at (13.5,4.5) {};
\node[fvertex] (9v3) at (15,4.5) {};
\node[fvertex] (10v3) at (16.5,4.5) {};
\node[fvertex] (11v3) at (18,4.5) {};
\node[fvertex] (12v3) at (19.5,4.5) {};
\node[fvertex] (13v3) at (21,4.5) {};
\node[gvertex] (14v3) at (22.5,4.5) {};

\draw[gray] (-0.75, -0.5) -- (9.75, -0.5);
\draw[gray] (-0.75, 5) -- (9.75, 5);

\draw[red, dotted] (9.75, -0.5) -- (14.25, -0.5);
\draw[red, dotted] (10.25, 2.5) -- (13.75, 2.5);
\draw[red, dotted] (9.75, 5) -- (14.25, 5);

\draw[gray] (14.25, -0.5) -- (23.25, -0.5);
\draw[gray] (14.25, 5) -- (23.25, 5);

\draw[gray] (-0.75,-0.5)--(-0.75,5);
\draw[gray] (0.75,-0.5)--(0.75,5);
\draw[gray] (2.25,-0.5)--(2.25,5);
\draw[gray] (3.75,-0.5)--(3.75,5);
\draw[gray] (5.25,-0.5)--(5.25,5);
\draw[gray] (6.75,-0.5)--(6.75,5);
\draw[gray] (8.25,-0.5)--(8.25,5);
\draw[gray, dotted] (9.75,-0.5)--(9.75,5);
\draw[gray, dotted] (14.25,-0.5)--(14.25,5);
\draw[gray] (15.75,-0.5)--(15.75,5);
\draw[gray] (17.25,-0.5)--(17.25,5);
\draw[gray] (18.75,-0.5)--(18.75,5);
\draw[gray] (20.25,-0.5)--(20.25,5);
\draw[gray] (21.75,-0.5)--(21.75,5);
\draw[gray] (23.25,-0.5)--(23.25,5);

\draw[red] (6.75,-0.5)--(6.75,5);
\draw[red] (9.75,-0.5)--(9.75,5);
\draw[red] (6.75,-0.5)--(9.75,-0.5);
\draw[red] (6.75,5)--(9.75,5);

\draw[red] (14.25,-0.5)--(14.25,5);
\draw[red] (17.25,-0.5)--(17.25,5);
\draw[red] (14.25,-0.5)--(17.25,-0.5);
\draw[red] (14.25,5)--(17.25,5);

\draw[<-] (-1.0,0)--(0v0);
\draw[->] (0v0)--(1v0);
\draw[<-] (1v0)--(2v0);
\draw[->] (2v0)--(3v0);
\draw[<-] (3v0)--(4v0);
\draw[->] (4v0)--(5v0);
\draw[<-] (5v0)--(6v0);
\draw[->] (6v0)--(7v0);
\draw[->] (8v0)--(9v0);
\draw[<-] (9v0)--(10v0);
\draw[->] (10v0)--(11v0);
\draw[<-] (11v0)--(12v0);
\draw[->] (12v0)--(13v0);
\draw[<-] (13v0)--(14v0);
\draw[->] (14v0)--(23.5,0);

\draw[->] (-1.0,1.5)--(0v1);
\draw[<-] (0v1)--(1v1);
\draw[->] (1v1)--(2v1);
\draw[<-] (2v1)--(3v1);
\draw[->] (3v1)--(4v1);
\draw[<-] (4v1)--(5v1);
\draw[->] (5v1)--(6v1);
\draw[<-] (6v1)--(7v1);
\draw[<-] (8v1)--(9v1);
\draw[->] (9v1)--(10v1);
\draw[<-] (10v1)--(11v1);
\draw[->] (11v1)--(12v1);
\draw[<-] (12v1)--(13v1);
\draw[->] (13v1)--(14v1);
\draw[<-] (14v1)--(23.5,1.5);

\draw[<-] (-1.0,3)--(0v2);
\draw[->] (0v2)--(1v2);
\draw[<-] (1v2)--(2v2);
\draw[->] (2v2)--(3v2);
\draw[<-] (3v2)--(4v2);
\draw[->] (4v2)--(5v2);
\draw[<-] (5v2)--(6v2);
\draw[->] (6v2)--(7v2);
\draw[->] (8v2)--(9v2);
\draw[<-] (9v2)--(10v2);
\draw[->] (10v2)--(11v2);
\draw[<-] (11v2)--(12v2);
\draw[->] (12v2)--(13v2);
\draw[<-] (13v2)--(14v2);
\draw[->] (14v2)--(23.5,3);

\draw[->] (-1,4.5)--(0v3);
\draw[<-] (0v3)--(1v3);
\draw[->] (1v3)--(2v3);
\draw[<-] (2v3)--(3v3);
\draw[->] (3v3)--(4v3);
\draw[<-] (4v3)--(5v3);
\draw[->] (5v3)--(6v3);
\draw[<-] (6v3)--(7v3);
\draw[<-] (8v3)--(9v3);
\draw[->] (9v3)--(10v3);
\draw[<-] (10v3)--(11v3);
\draw[->] (11v3)--(12v3);
\draw[<-] (12v3)--(13v3);
\draw[->] (13v3)--(14v3);
\draw[<-] (14v3)--(23.5,4.5);

\foreach \i in {0,1,2,3,4,5,6} {
    \pgfmathsetmacro{\x}{\i*1.5}
    \pgfmathsetmacro{\dir}{int(mod(\i,2)==0?1:0)}
    \ifnum\dir=1
        \draw[->] (\x,-1.0)--(\i v0);
        \draw[<-] (\i v0)--(\i v1);
        \draw[->] (\i v1)--(\i v2);
        \draw[<-] (\i v2)--(\i v3);
        \draw[->] (\i v3)--(\x,5.5);
    \else
        \draw[<-] (\x,-1.0)--(\i v0);
        \draw[->] (\i v0)--(\i v1);
        \draw[<-] (\i v1)--(\i v2);
        \draw[->] (\i v2)--(\i v3);
        \draw[<-] (\i v3)--(\x,5.5);
    \fi
}

\foreach \i in {9,10,11,12,13,14} {
    \pgfmathsetmacro{\x}{\i*1.5+1.5}
    \pgfmathsetmacro{\dir}{int(mod(\i,2)==1?1:0)}
    \ifnum\dir=1
        \draw[<-] (\x,-1.0)--(\i v0);
        \draw[->] (\i v0)--(\i v1);
        \draw[<-] (\i v1)--(\i v2);
        \draw[->] (\i v2)--(\i v3);
        \draw[<-] (\i v3)--(\x,5.5);
    \else
        \draw[->] (\x,-1.0)--(\i v0);
        \draw[<-] (\i v0)--(\i v1);
        \draw[->] (\i v1)--(\i v2);
        \draw[<-] (\i v2)--(\i v3);
        \draw[->] (\i v3)--(\x,5.5);
    \fi
}
\draw[red, ->, double, very thick] (1.5, 6)--(1v3);
\end{tikzpicture}

\caption{Configurations yielding lower bound for $S_{4,n}$. The initial position of the ant is indicated by the red double arrow.}
\label{fig:4max}
\end{figure}
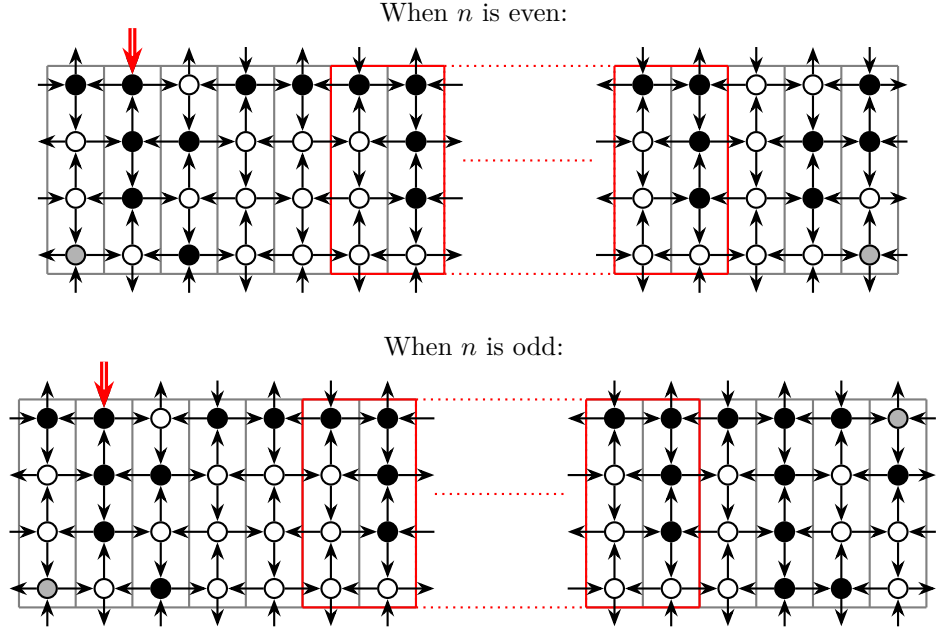

\subsection{Upper Bound for the Escaping Time of Rectangular Domains}\label{sec:rectupperbound}

Using the same inductive decomposition as in Section~\ref{sec:square}, 
adapted to the rectangular setting, we now derive an upper bound for 
$S_{k,n}$ that improves on the general $2^{kn}$ bound for fixed $k$.

\begin{theorem}\label{thm:rect}
For fixed $k \ge 2$ and $n \ge k$,
\[
S_{k,n} \;\le\;
\begin{cases}
(6n-4)\cdot\left(\dfrac{n+1}{\sqrt{2}}\right)^{k-2} & \text{if } k \text{ is even,}\\[6pt]
(10n-2)\cdot\left(\dfrac{n+1}{\sqrt{2}}\right)^{k-3} & \text{if } k \text{ is odd.}
\end{cases}
\]
In particular, $S_{k,n}  \le \left(\frac{n+1}{\sqrt{2}}\right)^k.$ 
\end{theorem}

\begin{proof}
We decompose $G_{k,n}$ as the union of the interior grid 
$G_{k-2,n}$ and two  $G_{1,n}$, one along the top 
boundary and one along the bottom boundary 
(Figure~\ref{fig:recgrid_decomposition}).

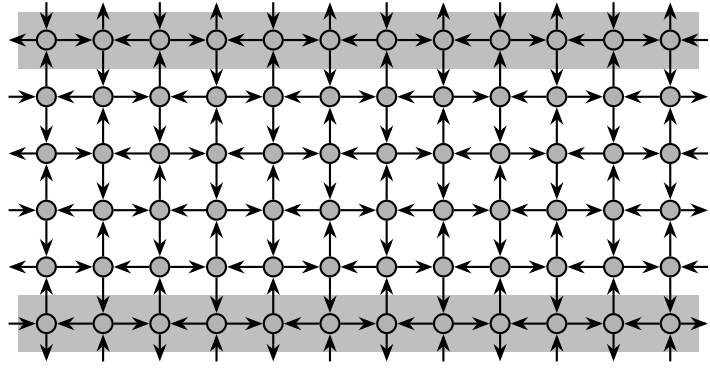
\begin{figure}[htbp]
\centering
\begin{tikzpicture}[
scale=0.5,
vertex/.style={circle, draw, minimum size=2.5mm, inner sep=0pt},
fvertex/.style={circle, draw, fill=black, minimum size=2.5mm, inner sep=0pt}, 
gvertex/.style={circle, draw, fill=gray!60, minimum size=2.5mm, inner sep=0pt}, 
>=Stealth,
thick
]
\def\k{6} 
\def\n{12} 
\def\spacing{1.5}
\def\externalarrow{1.0}

\pgfmathsetmacro{\gridleft}{-\spacing/2}
\pgfmathsetmacro{\gridright}{(\n-1)*\spacing+\spacing/2}
\pgfmathsetmacro{\gridbottom}{-\spacing/2}
\pgfmathsetmacro{\gridtop}{(\k-1)*\spacing+\spacing/2}

\fill[gray!50] (\gridleft, \gridbottom) rectangle (\gridright , \gridbottom+1.5) ;

\fill[gray!50] (\gridleft, \gridtop-1.5) rectangle (\gridright, \gridtop);


\foreach \i in {0,...,\numexpr\n-1} {
    \foreach \j in {0,...,\numexpr\k-1} {
        \pgfmathsetmacro{\x}{\i*\spacing}
        \pgfmathsetmacro{\y}{\j*\spacing}
        \node[gvertex] (\i v\j) at (\x,\y) {};
    }
}

\foreach \j in {0,...,\numexpr\k-1} {
    \pgfmathsetmacro{\y}{\j*\spacing}
    \pgfmathsetmacro{\dir}{int(mod(\j,2)==0?1:0)}
    \ifnum\dir=1
        \draw[->] (-\externalarrow,\y) -- (0v\j);
    \else
        \draw[<-] (-\externalarrow,\y) -- (0v\j);
    \fi
    \foreach \i in {0,...,\numexpr\n-2} {
        \pgfmathsetmacro{\arrowdir}{int(mod(\i+\j,2)==0?1:0)}
        \ifnum\arrowdir=1
            \draw[<-] (\i v\j) -- (\the\numexpr\i+1\relax v\j);
        \else
            \draw[->] (\i v\j) -- (\the\numexpr\i+1\relax v\j);
        \fi
    }
    \pgfmathsetmacro{\xright}{(\n-1)*\spacing+\externalarrow}
    \pgfmathsetmacro{\lastdir}{int(mod(\n-1+\j,2)==0?1:0)}
    \ifnum\lastdir=1
        \draw[<-] (\the\numexpr\n-1\relax v\j) -- (\xright,\y);
    \else
        \draw[->] (\the\numexpr\n-1\relax v\j) -- (\xright,\y);
    \fi
}

\foreach \i in {0,...,\numexpr\n-1} {
    \pgfmathsetmacro{\x}{\i*\spacing}
    \pgfmathsetmacro{\dir}{int(mod(\i,2)==0?1:0)}
    \ifnum\dir=1
        \draw[<-] (\x,-\externalarrow) -- (\i v0);
    \else
        \draw[->] (\x,-\externalarrow) -- (\i v0);
    \fi
    \foreach \j in {0,...,\numexpr\k-2} {
        \pgfmathsetmacro{\arrowdir}{int(mod(\i+\j,2)==0?1:0)}
        \ifnum\arrowdir=1
            \draw[->] (\i v\j) -- (\i v\the\numexpr\j+1\relax);
        \else
            \draw[<-] (\i v\j) -- (\i v\the\numexpr\j+1\relax);
        \fi
    }
    \pgfmathsetmacro{\ytop}{(\k-1)*\spacing+\externalarrow}
    \pgfmathsetmacro{\lastdir}{int(mod(\i+\k-1,2)==0?1:0)}
    \ifnum\lastdir=1
        \draw[->] (\i v\the\numexpr\k-1\relax) -- (\x,\ytop);
    \else
        \draw[<-] (\i v\the\numexpr\k-1\relax) -- (\x,\ytop);
    \fi
}
\end{tikzpicture}
\caption{Decomposition of $G_{6,12}$ as the union of the interior grid 
$G_{4,12}$ (white region) and two  $G_{1,12}$, one along the top 
boundary and one along the bottom boundary(gray shaded region). }
\label{fig:recgrid_decomposition}
\end{figure}

By Remark~\ref{rm:inb_max}, 
the ant starts from an in-boundary arc of $G_{k,n}$, reaches an 
in-boundary arc of $G_{k-2,n}$ in at most two steps, and then alternates 
between phases inside $G_{k-2,n}$ and bouncing phases on the two $G_{1,n}$, 
until it finally exits $G_{k,n}$ in at most two additional steps. If the 
ant returns to $G_{k-2,n}$ at most $Y_n$ times after its initial 
entry, we obtain
   $$ S_{k,n} \;\le\; (Y_n+1)\cdot S_{k-2,n} + 2Y_n + 4.$$

It remains to bound $Y_n$. Since the decomposition involves only the 
top and bottom $G_{1,n}$ (two sides rather than four), the same per-row 
bouncing analysis as in Section~\ref{sec:square} gives  $Y_n \le 2 \cdot \frac{n^2}{4} = \frac{n^2}{2}$. Substituting:
\[
    S_{k,n} \;\le\; \left(\frac{n^2}{2}+1\right) S_{k-2,n} + n^2 + 4 \text{ or equivalently } 2S_{k,n} \le (n^2+2)\,S_{k-2,n} + 2n^2 + 8
\]

\noindent Defining $T_{k,n} = S_{k,n}+2$, this rewrites as   $2T_{k,n} \;\le\; (n^2+2)\,T_{k-2,n} + 8.$

Since $(n^2+2) = (n+1)^2 - (2n-1)$ and $(2n-1)\,T_{k-2,n} > 8$ for all 
$k \ge 4$ and $n \ge k$ (which follows from $T_{2,n}=6(n-1)+2$ and the monotonicity of $(2n-1)\,T_{k-2,n}$), 
we obtain
\[
    2T_{k,n} \;\le\; (n+1)^2\,T_{k-2,n},
    \quad\text{i.e.,}\quad
    T_{k,n} \;\le\;  \left( \frac{n+1}{\sqrt{2}}\right)^2  \,T_{k-2,n} .
\]

Applying  this inequality by induction on $k$, using the base cases
$T_{2,n} \le 6n-4$ (for even $k$) and $T_{3,n} \le 10n-2$ (for odd $k$),
given by Theorems~\ref{pr:theotworow} and~\ref{pr:theothreerow} respectively,
yields the stated bound.
\end{proof}

\section{Conclusion and Open Questions}\label{sec:pers}

We studied the escaping time of Langton's ant on finite connected domains of
the square grid, establishing general structural properties and deriving upper
bounds. We obtained a factorial upper bound for square domains and linear bounds
for rectangular domains of height two and three, supported by exact values from simulations and matching or near-matching lower-bound constructions.
More generally, for a fixed height $k$,  we obtained an upper bound  $\left(\frac{n+1}{\sqrt{2}}\right)^k$ for rectangular domains $G_{k,n}$. 
Several directions remain open.

Our simulation results for square domains ($n \le 9$) suggest that $S_{n,n}$
grows significantly slower than the factorial upper bound $(n+1)!$, with
values consistent with cubic growth $S_{n,n} \in O(n^3)$. Can the factorial
bound be improved to a polynomial bound in $n$?

Extending the column-by-column analysis to grids of height four appears
feasible, although the case analysis becomes significantly more complex.
Preliminary simulations suggest a linear bound of the form $S_{4,n} \le 34n$,
and we have identified constructions achieving values close to this bound.
Can this conjecture be proved?

More generally, while $\left(\frac{n+1}{\sqrt{2}}\right)^k$ guarantees a polynomial upper
bound in $n$ for every fixed $k$ for $S_{k,n}$, the degree grows with $k$. Is there a
function $f$ such that $S_{k,n} \le f(k)\,n$ for all $k \in \mathbb{N}$?
In other words, is the escaping time linear in $n$ for every fixed height $k$?

Computing exact values of $S_{k,n}$ for larger grids (e.g., $10 \times 10$
and beyond) could help refine these conjectures, and designing constructions
yielding stronger lower bounds remains an important combinatorial challenge.

Finally, a better understanding of escaping time may contribute to clarifying
the computational complexity of Langton's ant on finite grids, particularly
for reachability problems.

\bibliography{biblio}

\end{document}